\documentclass[11pt,letterpaper]{article}
\usepackage{fullpage}
\usepackage[top=2cm, bottom=4.5cm, left=2.5cm, right=2.5cm]{geometry}
\usepackage{amsmath,amsthm,amsfonts,amssymb,amscd}
\usepackage{lastpage}
\usepackage{enumerate}
\usepackage{fancyhdr}
\usepackage{mathrsfs}
\usepackage{xcolor}
\usepackage{graphicx}
\usepackage{listings}
\usepackage[colorlinks=true, urlcolor=blue, linkcolor=blue, citecolor=magenta]{hyperref}
\usepackage{cite}
\usepackage{multirow}
\usepackage{comment}
\usepackage{mathpazo}
\usepackage{todonotes}
\usepackage{complexity}
\usepackage{setspace}

\hypersetup{
  colorlinks=true,
  linkcolor=red,
  linkbordercolor={0 0 1}
}

\lstdefinestyle{Matlab}{
    language        = matlab,
    frame           = lines, 
    basicstyle      = \footnotesize,
    keywordstyle    = \color{blue},
    stringstyle     = \color{green},
    commentstyle    = \color{red}\ttfamily
}
\newtheorem{theorem}{Theorem}[section]
\newtheorem{definition}[theorem]{Definition}

\newtheorem{claim}[theorem]{Claim}

\newtheorem{remark}[theorem]{Remark}

\newtheorem{lemma}[theorem]{Lemma}
\newtheorem{hypothesis}[theorem]{Hypothesis}
\newtheorem{corollary}[theorem]{Corollary}

\newcommand{\wone}{\mathsf{W[1]}}
\newcommand{\fpt}{\mathsf{FPT}}

\newcommand{\kclique}{\textsc{$k$-Clique}}

\DeclareMathOperator{\MAJ}{\mathrm{Maj}}
\DeclareMathOperator*{\Ex}{\mathbb{E}}
\DeclareMathOperator*{\Var}{\text{Var}}

\title{Improved Hardness of Approximating $k$-Clique under ETH}

\author{
Bingkai Lin \thanks{State Key Laboratory for Novel Software Technology, Nanjing University. Email: \texttt{lin@nju.edu.cn}}
\and
Xuandi Ren\thanks{University of California, Berkeley. Email: \texttt{xuandi\_ren@berkeley.edu}}
\and
Yican Sun\thanks{Peking University. Email: \texttt{sycpku@pku.edu.cn}}
\and 
Xiuhan Wang\thanks{Tsinghua University. Email: \texttt{wangxh19@mails.tsinghua.edu.cn}}
}
\date{}

\begin{document}

\maketitle

\begin{abstract}
In this paper, we prove that assuming the exponential time hypothesis (ETH), there is no $f(k)\cdot n^{k^{o(1/\log\log k)}}$-time algorithm that can decide whether an $n$-vertex graph contains a clique of size $k$ or contains no clique of size $k/2$, and no FPT algorithm can decide whether an input graph has a clique of size $k$ or no clique of size $k/f(k)$, where $f(k)$ is some function in $k^{1-o(1)}$. Our results significantly improve the previous works~\cite{Lin21,DBLP:conf/icalp/LinRSW22}. The crux of our proof is a framework to construct gap-producing reductions for the \kclique{} problem. More precisely, we show that given an error-correcting code $C:\Sigma_1^k\to\Sigma_2^{k'}$ that is locally testable and smooth locally decodable in the parallel setting,  one can construct a reduction which on input a graph $G$ outputs a graph $G'$ in $(k')^{O(1)}\cdot n^{O(\log|\Sigma_2|/\log|\Sigma_1|)}$ time such that
 \begin{itemize}
     \item if $G$ has a clique of size $k$, then $G'$ has a clique of size $K$, where $K = (k')^{O(1)}$.
     \item if $G$ has no clique of size $k$, then $G'$ has no clique of size $(1-\varepsilon)\cdot K$ for some constant $\varepsilon\in(0,1)$.
 \end{itemize}
 We then construct such a code with $k'=k^{\Theta(\log\log k)}$ and $|\Sigma_2|=|\Sigma_1|^{k^{0.{54}}}$, establishing the hardness result above. Our code generalizes the derivative code~\cite{WY07} into the case with a super constant order of derivatives. 
\end{abstract}
\clearpage

\begin{spacing}{0.8}
\tableofcontents
\end{spacing}
\clearpage

\section{Introduction}
In this work, we study the \kclique{} problem: given a simple graph $G$ with $n$ vertices, decide whether $G$ contains a clique of size $k$. The \kclique{} problem is one of the most fundamental computational problems in complexity theory. It is $\mathsf{NP}$-hard~\cite{kar72}, which means that there is no polynomial time algorithm for \kclique{} assuming $\mathsf{P}\ne \mathsf{NP}$.

It is natural to consider \emph{approximation} to circumvent the intractability. 
The decision version of approximating \kclique{}, also known as the $\gamma(k)$-gap \kclique{} problem for some approximation ratio $\gamma(k)>1$, aims to distinguish between graphs with a  $k$-clique and those without  clique of size $k/\gamma(k)$. 
Unfortunately, there has been a long line of work~\cite{bellare1993efficient, bellare1994improved,feige1996interactive, hastad1996clique,  goldwasser1998introduction, feige2000two, Zuc07} showing that even for the nearly tight ratio $n^{1-\varepsilon}$, the $n^{1-\varepsilon}$-gap \kclique{} problem remains $\mathsf{NP}$-hard for every $\varepsilon>0$.

Besides approximation, \emph{parameterization} is also widely used to bypass $\mathsf{NP}$-hardness. In the parameterized regime, instead of using polynomial-time algorithms, we treat $k$ as a parameter independent of $n$, and allow algorithms with running time $f(k)\cdot n^{g(k)}$, where $f,g$ are some computable functions (e.g., $k!$) independent of $n$. In addition, we say an algorithm is \emph{fixed parameter tractable} (FPT) if $g(k)=O(1)$. However, \kclique{} is still hard in this setting:
\begin{itemize}
    \item \kclique{} has been identified as the canonical complete problem in $\wone$, which is the parameterized analogy of $\mathsf{NP}$~\cite{downey1995fixed}. Hence, \kclique{} admits no FPT algorithm assuming $\wone \ne \fpt$.
    \item Assuming the \emph{exponential time hypothesis} (ETH), we can rule out $f(k)\cdot n^{o(k)}$ time algorithms for \kclique{} for every computable function $f(k)$~\cite{ChenHKX06}.
\end{itemize}

In parallel with the classic setting, it is natural to ask whether \kclique{} admits an efficient approximation algorithm in the parameterized regime. Note that $\gamma(k)$-gap \kclique{} has a trivial $n^{O(k/\gamma(k))}$-time algorithm which simply enumerates all subsets with size $k/\gamma(k)$. Thus we consider the following question:
\begin{center}
    \textit{Does $\gamma(k)$-gap \kclique{} admits an algorithm that runs in $f(k)\cdot n^{o(k/\gamma(k))}$ time for some computable function $f(k)$?}
\end{center}
 This problem is fundamental due to the importance of \kclique{} in the parameterized complexity theory. It  is also connected with the \emph{parameterized inapproximability hypothesis} (PIH)~\cite{LRSZ20}, a central conjecture in the parameterized complexity theory. The conjecture states that no algorithm can constantly approximate the constraint satisfaction problem (CSP) with $k$ variables and alphabet size $n$ in FPT time\footnote{Note that the original PIH in \cite{LRSZ20} states that approximating the constraint satisfaction problem parameterized by the number of variables within a constant ratio is $\mathsf{W[1]}$-hard. Here we actually use a relaxed form.}. As shown in~\cite{DBLP:conf/icalp/LinRSW22}, PIH can be implied by a lower bound of $f(k)\cdot n^{\omega(k/\log k)}$ for constant-gap \kclique{}\footnote{Throughout this paper, the lower bound for constant-gap \kclique{} means that there \emph{exists} a constant $C>1$ such that $C$-gap \kclique{} has such a lower bound. This is harmless because we can apply expanders to amplify the gap to every constant.}. Thus, an immediate open problem is to improve the lower bound for constant-gap \kclique{} to prove PIH (under ETH). 
 
 The study of the parameterized  gap \kclique{} has received increasing attention recently. First, the work \cite{CCK+17} rules out all algorithms within running time $f(k)\cdot n^{o(k/\gamma(k))}$ for every $\gamma(k) = o(k)$, assuming Gap-ETH~\cite{Din16,MR16}. Recently, Lin \cite{Lin21} shows that constant-gap \kclique{} has no FPT algorithms assuming $\mathsf{W[1]}\ne \mathsf{FPT}$. The work also implicitly rules out $f(k)\cdot n^{o(\sqrt[6]{\log k})}$ time algorithm for constant-gap 
\kclique{} assuming ETH. Lin's reduction is further improved in two follow-up works: In \cite{DBLP:conf/icalp/LinRSW22},   the lower bound for constant-gap \kclique{} is improved from $f(k)\cdot n^{\Omega(\sqrt[6]{\log k})}$ to $f(k)\cdot n^{\Omega({\log k})}$ under ETH. On another orthogonal  direction, Karthik and Khot~\cite{DBLP:conf/coco/SK22} improve the FPT inapproximability ratio from constant to $k^{o(1)}$ under $\mathsf{W[1]}\ne \mathsf{FPT}$.

In this work, we make significant progress in proving the lower bound of constant-gap \kclique{} under ETH.

\begin{theorem}[Theorem~\ref{thm:eth-gap-clique-lowerbound}]\label{thm:intro-A}
Assuming ETH, for some constant $\lambda>0$, constant-gap \kclique{} admits no algorithm within running time $f(k)\cdot n^{k^{\lambda /\log \log k}}$ for any computable function $f$.    
\end{theorem}
By using a disperser argument to amplify the gap, we have the following:
\begin{theorem}[Theorem~\ref{thm:eth-gap-clique-fpt}]\label{thm:intro-B}
Assuming ETH, for some constant $\lambda>0$, $(k/(\log k)^{\lambda\log\log\log k})$-gap \kclique{} has no FPT algorithm.
\end{theorem}
On the other hand, we note that the parameter choice in~\cite{DBLP:conf/icalp/LinRSW22} regarding the relationship between lower bound of constant-gap \kclique{} and PIH is not optimal. By optimizing the parameters, we relax the lower bound requirement from $f(k)\cdot n^{\omega(k/\log k)}$ to $f(k)\cdot n^{\omega(\sqrt k)}$ and obtain the following:
\begin{theorem}[See Appendix~\ref{app:lrsw22}]\label{thm:pih}
If for any computable function $f(k)$, constant-gap \kclique{} has no algorithm in time $f(k)\cdot n^{O(\sqrt k)}$, then PIH is true.
\end{theorem}

The highlights  of this paper are listed below.

\begin{itemize}
    \item Compared with~\cite{CCK+17}, our results are obtained under the weaker gap-free hypothesis ETH, and the main difficulty here is ``gap-producing''.  In contrast, ~\cite{CCK+17} assumes Gap-ETH, where the gap is inherent in the assumption, and the problem is to preserve the gap. Thus, their approach does not apply to our setting.
    \item Compared with~\cite{DBLP:conf/icalp/LinRSW22}, our results significantly improve the lower bound of constant-gap \kclique{} under ETH from $f(k)\cdot n^{\Omega(\log k)}$ to $f(k)\cdot n^{k^{\Omega(1/\log \log k)}}$. 
    \item Compared with~\cite{Lin21} and~\cite{DBLP:conf/coco/SK22}, though our results rely on a stronger but still gap-free assumption, we significantly improve the inapproximability ratio for \kclique{} in FPT time, from $k^{o(1)}$ to some function in $k^{1-o(1)}$. 
    \item Our technique generalizes the previous proofs~\cite{Lin21,DBLP:conf/icalp/LinRSW22} into a \emph{systematic framework} for deriving the lower bound of constant-gap \kclique{}. In a nutshell, we show that \emph{a nice code is all you need} for a gap-producing reduction for \kclique{}.
\end{itemize}

\noindent\textbf{Our reduction framework.} Previous approaches~\cite{Lin21,DBLP:conf/icalp/LinRSW22}  use the classic Hadamard code or Reed-Muller code to encode the solution of an intermediate CSP problem, and derive the lower bound for constant-gap \kclique{} using some implicit properties of these codes. Therefore, it is unclear how to improve the lower bound by using other error-correcting codes.

Our new framework explicitly separates the code from the reduction to constant-gap \kclique{}. Specifically, our framework requires a certain type of error-correcting code $C: \Sigma_1^{k} \to \Sigma_2^{k'}$, which we call ``parallel locally testable and decodable codes’'~(PLTDCs, Definition~\ref{def:nice}). Any such a code can be plugged into our framework to obtain a lower bound. Intuitively, the higher rate and lower alphabet size blow-up of the code $C$ has, the stronger lower bound we can get. See Theorems~\ref{thm:reduction-intro} and~\ref{thm:vcsp2gapclique} for more details.

To prove Theorem~\ref{thm:intro-A}, we construct such codes with $k' = k^{\Theta(\log \log k)}$ and $|\Sigma_2| = |\Sigma_1|^{k^{0.{54}}}$. Our construction generalizes the derivative code~\cite{WY07} into the case with higher-order derivatives. We construct such codes by choosing the derivative order to be $\Theta(\log k)$.

In the end, our framework does not completely incorporate the reduction in~\cite{DBLP:conf/coco/SK22} due to its stronger requirements for the local decodability. 

However, one of the ingredients in our framework, namely the vector 2CSP, may serve as an alternative to the intermediate problem there, and can thus simplify their proof. In fact, Chen, Feng, Laekhanukit and Liu \cite{CFLL23} already  gave a simple and elegant proof which extends the result in~\cite{DBLP:conf/coco/SK22}. Our proof for the hardness of vector 2CSP is similar to  the \emph{Sidon Set} technique used in~\cite{CFLL23}.

\medskip\noindent\textbf{Future work.}
We propose two future directions. First, by our Theorem~\ref{thm:pih}, to prove PIH under ETH, we need to improve the lower bound for constant-gap \kclique{} from $f(k)\cdot n^{k^{\Omega(1/\log \log k)}}$ to $f(k)\cdot n^{\omega(\sqrt k)}$.  It is worth considering whether this can be done by constructing a more efficient PLTDC. 

The second direction is to make our framework more fruitful and incorporate more approaches, such as~\cite{DBLP:conf/coco/SK22}.

\section{Proof Overview}

In this section, we present an overview of our results. All definitions and claims in this part are \emph{informal}. We will give formal descriptions when we make some definitions and claims.

\subsection{Key Idea}
 Approximating \kclique{} can be recognized as a weaker variant of approximating CSPs. We follow the well-established strategy for proving the $\mathsf{NP}$-hardness of approximating CSPs (i.e., the PCP theorems)~\cite{ASLMSS98,hastad3querypcp,kozen2006theory,OM06}. 
The previous reduction starts with arithmetizing 3SAT into a CSP with $k$ variables and alphabet $\Sigma_1$.
Then, it takes a locally testable code $C: \Sigma_1^k\to \Sigma_2^{k'}$ and sets the proof to be the codeword that encodes a solution $\varphi:[k]\to \Sigma_1$ of the CSP problem. 

Unfortunately, this strategy does not work in the parameterized setting. The previous proofs set $k'$ to be either $|\Sigma_1|^{\Omega(k)}$~\cite{ASLMSS98,kozen2006theory,OM06} or $\Omega(2^{|\Sigma_1|}k)$~\cite{hastad3querypcp}.
But in the parameterized reduction, the length $k'$ must be independent of $|\Sigma_1|$, making the above strategy inapplicable.

To address this, our idea is to treat the alphabet $\Sigma_1$ as a vector space $\mathbb F^d$ and the solution $\varphi$ as $d$ independent messages in $\mathbb F^k$. 

Then, we construct a ``good'' code $C:\mathbb F^k\to (\mathbb F')^{k'}$, apply $C$ to each message, and combine each symbol of the $d$ codewords into an element in $({\mathbb F}')^d$. This way, we make the output $k'$ independent of $|\Sigma_1|$. To implement this idea, our approach consists of three parts:
\begin{itemize}
\item We first need an intermediate CSP problem whose alphabet is a vector space. For this purpose, we introduce the vector 2CSP~(Definition~\ref{def:bvcsp}) as the intermediate problem.
\item We propose a new type of codes, namely \emph{parallel locally testable and decodable code} (PLTDC, Definition~\ref{def:nice}), to encode the solution for a vector 2CSP instance in FPT time.
\item We apply the modified FGLSS reduction from vector 2CSP to constant-gap \kclique{}  by encoding the solution for a vector 2CSP instance~(Theorem \ref{thm:vcsp2gapclique}).
\end{itemize}

In Sections \ref{sec:v2csp-inro} -- \ref{sec:fglss-intro}, we illustrate three parts one by one. Then, in Section~\ref{sec:construct-intro}, we overview how to construct an efficient PLTDC using polynomials. Finally, in Section~\ref{sec:hardness-intro}, we establish the hardness of constant-gap \kclique{} as Theorems~\ref{thm:intro-A} and~\ref{thm:intro-B}. 

\subsection {Vector 2CSP}\label{sec:v2csp-inro}
Given a finite field $\mathbb F$ and a dimension $d$, a vector 2CSP instance consists of $k$ variables $\{x_1,\dots,x_k\}$.
The alphabet is the vector space $\mathbb F^d$. For each $i\in [k]$, there is a unary constraint ``$x_i\in S_i$'' for some $S_i\subseteq \mathbb F^d$. For each $(i,j)\in \binom{[k]}{2}$, there is a binary constraint ``$x_i+x_j\in S_{ij}$'' for some $S_{ij}\subseteq \mathbb F^d$. The following lemma establishes the hardness of vector 2CSPs. 

\begin{lemma}[Corollaries~\ref{cor:w1} and~\ref{cor:eth}]\label{lem:vcsp-intro}
For every integer $k$ and finite field $\mathbb F$, vector 2CSP problem with $k$ variables and  dimension $\Theta(\frac{\log n}{\log |\mathbb F|})$ over the field $\mathbb F$ is $\wone$-hard and has running time lower bound $f(k)\cdot n^{\Omega(k)}$ under ETH.
\end{lemma}

Compared with the {\sc $k$-VectorSum} problem, which was used as the intermediate problem in proving the hardness of gap \kclique{} previously~\cite{Lin21,DBLP:conf/icalp/LinRSW22,DBLP:conf/coco/SK22}, and whose hardness requires the dimension $d=\Omega(k^{2}\log n)$, the hardness of vector 2CSP only requires $d=O(\log n)$. This avoids the random projection step in the previous proofs for dimension reduction and simplifies the proof for the hardness of gap \kclique{}.

\begin{proof}[Proof of Lemma~\ref{lem:vcsp-intro}]
We reduce \kclique{} to this CSP problem by embedding the vertices into low-dimension vectors with a special hash function $H$. Given the graph $G = (V,E)$ for \kclique{}, the hash function $H : V\to \mathbb F^{\frac{5\log |V|}{\log |\mathbb F|}}$ has two properties:
\begin{itemize}
    \item For any two distinct vertices $v, v'\in V$, $H(v)\neq H(v')$.
    \item For any two distinct pairs $(v, w), (v', w')\in V\times V$, $H(v) + H(w)\neq H(v') + H(w')$.
\end{itemize}
The two properties ensure that $H(v)$ uniquely determines a vertex $v$ and $H(v)+H(w)$ uniquely determines a pair $(v,w)$. Such functions can be constructed by viewing the name of each vertex as a vector over $\mathbb{F}$ and multiplying it by a random matrix. We can further use the method of conditional probabilities to derandomize the construction.

We then construct the following instance of vector 2CSP. Let $S = \{H(v): v\in V\}$ and $W = \{H(v)+H(w): v,w\in V\land (v,w)\in E\}$. We set up $k$ variables $\{x_1,\dots,x_k\}$ with the alphabet $\mathbb F^{\frac{5\log |V|}{\log |\mathbb F|}}$ of size $|V|^5$, where each $x_i$ is supposed to be the vector obtained by applying the hash function $H(\cdot)$ to the $i$-th vertex in the clique. Then, we set two types of constraints:
\begin{align*}
    \textstyle {\rm (C1)}\ \forall i\in [k]\ x_i\in S. & & {\rm (C2)}\ \forall i,j\in \binom{[k]}{2}\ x_i+x_j\in W.
\end{align*}
Intuitively, (C1) guarantees each $x_i$ stands for some unique vertex, and (C2) ensures an edge connects the $i$-th and the $j$-th vertex. By the two properties above of $H$, the CSP is satisfiable if and only if the input graph $G$ has a clique of size $k$, and thus Lemma~\ref{lem:vcsp-intro} follows. 
\end{proof}

\subsection{Parallel Locally Testable and Decodable Codes (PLTDCs)}\label{sec:pltdc-intro}
\medskip\noindent\textbf{Parallel encoding.} Before illustrating PLTDCs, we first introduce the parallel encoding. We present an illustration of parallel encoding in Figure~\ref{fig:eccfig}. In detail, given a code $C: \Sigma_1^{k}\to \Sigma_2^{k'}$ and an integer $\ell$ for the degree of parallelization, its parallel encoding gives the code $C^{\odot \ell}: \left(\Sigma_1^{\ell}\right)^{k}\to \left(\Sigma_2^{\ell}\right)^{k'}$ that maps $k$ vectors with length $\ell$ over $\Sigma_1$ into $k'$ vectors with length $\ell$ over $\Sigma_2$.
$C^{\odot \ell}$ is defined as follows.
Given the input message $(\mathbf{a}_1,\dots, \mathbf{a}_k)\in \left(\Sigma_1^{\ell}\right)^{k}$, we denote by $(\mathbf{z}_1,\dots, \mathbf{z}_{k'})$  the output codeword of $C^{\odot \ell}$ and denote by $\mathbf{a}_j[i]$ the $i$-th entry ($i\in [\ell]$) of the vector $\mathbf{a}_j$. Similar definition holds for  $\mathbf{z}_j$. Then, for every $i\in [\ell]$, we define the message $\mathbf{m}_i:=(\mathbf{a}_1[i],\dots, \mathbf{a}_k[i])\in \Sigma_1^{k}$ that extracts the $i$-th entry of the input, and define:
\begin{equation}
    (\mathbf{z}_1[i],\dots,\mathbf{z}_{k'}[i]) := C(\mathbf{m}_i)=C(\mathbf{a}_1[i],\dots, \mathbf{a}_k[i]).
\end{equation}
Note that the equation above uniquely defines the output codeword $(\mathbf{z}_1,\dots,\mathbf{z}_{k'})$. Intuitively, the parallel encoding treats the input $k$ length-$\ell$ vectors over $\Sigma_1$ as $\ell$ parallel messages in $\Sigma_1^{k}$, and applies the code $C$ to each message, finally zips the $\ell$ output codewords into $k'$ vectors with length $\ell$. As a toy example, consider $k=2$, $k'=3$, and $\ell=2$. For $\mathbf{m}_1 = (\mathbf{a}_1[1],\mathbf{a}_2[1])$ and $\mathbf{m}_2 = (\mathbf{a}_1[2],\mathbf{a}_2[2])$, the codeword is:
\begin{equation*}
    {\huge C^{\odot 2}}{\small\left(\begin{bmatrix}\mathbf{a}_1[1]\\\mathbf{a}_1[2]\end{bmatrix}, \begin{bmatrix}\mathbf{a}_2[1]\\\mathbf{a}_2[2]\end{bmatrix}\right) = \left(\begin{bmatrix}C(\mathbf{m}_1)_1\\C(\mathbf{m}_2)_1\end{bmatrix},\begin{bmatrix}C(\mathbf{m}_1)_2\\C(\mathbf{m}_2)_2\end{bmatrix},\begin{bmatrix}C(\mathbf{m}_1)_3\\C(\mathbf{m}_2)_3\end{bmatrix}\right)}.
\end{equation*}

 \begin{figure}[h]
\centering
\includegraphics[width=0.4\textwidth]{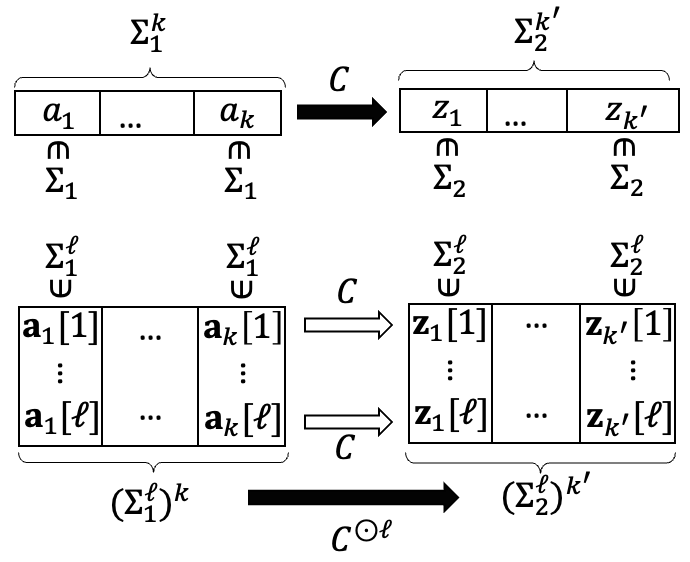}
\caption{Classic Encoding vs. Parallel Encoding}
\label{fig:eccfig}
\end{figure}

\medskip\noindent\textbf{PLTDCs.} Intuitively, $C$ is a parallel locally testable and decodable code (PLTDC) if all its $\ell$-parallel versions $C^{\odot \ell}$ are constant-query locally testable and 2-query locally decodable, for every $\ell\ge 1$. Below, we informally define a PLTDC $C$.
We fix a given word $w\in \left(\Sigma_2^{\ell}\right)^{k'}$ as the corrupted version of some codeword $C^{\odot \ell}(\mathbf{a}_1,\dots,\mathbf{a}_k)$.
\begin{itemize}
    \item \textit{($q$-query parallel locally testable)} For every $\ell\ge 1$, we can randomly test whether $w$ is close to some codeword by $q$ queries on $w$. The testing algorithm has the following properties.
    \begin{itemize}
        \item It randomly chooses $r\in [R_T]$, where  $R_T=\left(k'\right)^{O(1)}$, and queries $q$ positions based on $r$.
        \item {\it (Perfect completeness)} If $w$ is a codeword, then it always accepts $w$.
        \item{\it (Local testability)} If it accepts $w$ with probability at least $\varepsilon_T$, then $ \Delta\left(w,{\rm Im}\left(C^{\odot \ell}\right)\right)\le \delta$, where $\varepsilon_T$ and $\delta$ are independent of $\ell$. 
    \end{itemize}
    \item \textit{(2-query parallel smooth locally decodable)} For every $\ell\ge 1$, we can decode every $\mathbf{a}_i$ and $\mathbf{a}_i+\mathbf{a}_j$ by two queries on $w$. The decoding algorithm has the following properties.
    \begin{itemize}
        \item (\textit{Smoothness}) It chooses a randomness $r\in [R_D]$, where  $R_D=\left(k'\right)^{O(1)}$, and queries two positions based on $r$. Each position $\in [k']$ is queried with equal probability $2/k'$.
        \item (\textit{Perfect completeness}) If $w=C^{\odot \ell}(\mathbf{a}_1,\dots,\mathbf{a}_k)$, then it succeeds with probability 1.
        \item (\textit{Local decodability}) If $\Delta(w,C^{\odot \ell}(\mathbf{a}_1,\dots,\mathbf{a}_k))\le \delta_D$, then it succeeds with probability at least $\varepsilon_D$, where $\varepsilon_D$ is independent of $\ell$, and $\delta_D \ge \delta$.
    \end{itemize}
\end{itemize}

We discuss the restriction of PLTDCs in Remark~\ref{rem:restriction}, after finishing the overview of our framework. Below, we review the classic Hadamard code and show that it is a PLTDC in sketch. We present the formal treatments in Section~\ref{sec:hard-pltdc}. In Section~\ref{sec:2pir}, we show that the derivative code~\cite{WY07} is also a PLTDC.

\noindent\textbf{Example: the Hadamard code.}
We show that the Hadamard code is a PLTDC.
Given a field $\mathbb F$, let $H:\mathbb{F}^{k}\to\mathbb{F}^{\mathbb{F}^{k}}$ be the Hadamard code, i.e., for every $\mathbf{a}=(a_1,\dots,a_k)^T\in \mathbb F^k,\mathbf{x}\in\mathbb{F}^{k}$,
\[
H(a_1,\dots,a_k)[\mathbf{x}]=\mathbf{a}^T\mathbf{x},
\]
where we treat the codeword $H(\mathbf{a})[\cdot]$ as a function from $\mathbb{F}^{k}$ to $\mathbb{F}$. Similarly, consider the $\ell$-parallel code $H^{\odot \ell}: \left(\mathbb F^{\ell}\right)^{k}\to \left(\mathbb F^{\ell}\right)^{\mathbb F^k}$, given $\mathbf a \in \left(\mathbb F^{\ell}\right)^{k}$, we treat $\mathbf a$ as an element in $\mathbb F^{k\times \ell}$ and $H^{\odot \ell}(\mathbf{a})[\cdot]$ can be viewed as a function from $\mathbb{F}^{k}$ to $\mathbb{F}^\ell$, mapping $\mathbf x \in \mathbb F^{k}$ to $\mathbf a^T\mathbf x$. 

From \cite{blum1993self}, if a function $f$ satisfies $\Pr_{\mathbf{a},\mathbf{b}}[f(\mathbf{a})+f(\mathbf{b})=f(\mathbf{a}+\mathbf{b})]\ge 1-\frac{\delta}{2}$ for some $\delta$, then there exists some $\mathbf{a}_0\in \mathbb F^{k\times \ell}$ such that $f$ is at most $\delta$-far from $\mathbf{a}_0^T\mathbf{x}$, showing that the Hadamard code is parallel locally testable.

We show that one can locally decode a single input symbol or the sum of two input symbols in parallel from a Hadamard codeword. If $f$ satisfies $\Pr_\mathbf{x}[f(\mathbf{x})=\mathbf a^T\mathbf x]\ge 1-\delta$ for some $\mathbf a \in \mathbb F^{k \times \ell}$, which is also an element in $\left(\mathbb F^{\ell}\right)^k$, then for any $i\ne j\in [k]$, we have 
\begin{align*}
\Pr_{\mathbf{x}}[\mathbf{a}[i]=f(\mathbf{x}+\mathbf{e}_i)-f(\mathbf{x})]&\ge 1-2\delta,\\
\Pr_{\mathbf{x}}[\mathbf{a}[i] + \mathbf{a}[j]=f(\mathbf{x}+\mathbf{e}_i+\mathbf{e}_j)-f(\mathbf{x})]&\ge 1-2\delta,
\end{align*}
where we use $\mathbf a[i]\in \mathbb F^{\ell}$ to denote the $i$-th symbol of the message, and $\mathbf{e}_i$ is the $i$-th unit vector. Thus, the Hadamard code is parallel locally decodable. For details, please refer to Section~\ref{sec:hard-pltdc}.

\subsection{The Modified FGLSS Redcution}\label{sec:fglss-intro}
In this part we construct the reduction from vector 2CSP to constant-gap \kclique{} via PLTDCs. Our reduction is a variant of the classic FGLSS reduction~\cite{feige1996interactive}, and generalizes the previous reductions~\cite{Lin21, DBLP:conf/icalp/LinRSW22} from the Hadamard code into any PLTDC.
\begin{theorem}[Theorem~\ref{thm:vcsp2gapclique}]\label{thm:reduction-intro}
    Given a PLTDC $C: \Sigma_1^{k}\to \Sigma_2^{k'}$ defined as above, and a vector 2CSP instance with $k$ variables, alphabet $\Sigma_1$ and dimension $t$,
    we can construct a graph $G$ within time $(k')^{O(1)}\cdot |\Sigma_2|^{O(tq)}$ such that:
    \begin{itemize}
        \item $G$ consists of $R_T=(k')^{O(1)}$ parts of independent sets.
        \item (Completeness) If the vector 2CSP instance is satisfiable, then $G$ has a clique of size $R_T$.
        \item (Soundness) If $\delta<\frac1{12}$, and the instance is unsatisfiable, then $G$ has no clique of size $\varepsilon_TR_T$.
    \end{itemize}
\end{theorem}

Our reduction needs a disperser-like property of the PLTDC, for which we have a lemma as follows.
\begin{lemma}[Lemma~\ref{lem:ltc_num_var}]\label{lem:ltc-num-var-intro}
For every subset of randomness $R\subseteq [R_T]$ with size $|R|\ge \varepsilon_TR_T$, the number of indices queried by the testing algorithm with at least one randomness in $R$ is at least $(1-2\delta)\cdot k'$.
\end{lemma}
\begin{proof}[Proof sketch]
    If for some $R\subseteq [R_T]$ with $|R|\ge \varepsilon_TR_T$, the number of queried variables is less than $(1-2\delta)\cdot k'$, then we can fool the testing algorithm and obtain two different codewords of $C^{\odot t}$ with small distance, which contradicts the local decodability.
\end{proof}

\begin{proof}[Proof of Theorem~\ref{thm:reduction-intro}]
Given a vector 2CSP instances with $ k$ variables, $x_1,\dots,x_{k}$,
we construct $G$ based on the $t$-parallel code $C^{\odot t}: \left(\Sigma_1^t\right)^k\to \left(\Sigma_2^t\right)^{k'}$ as follows.

\medskip\noindent\textbf{Vertices of $G$.} $G$ consists of $R_T$ parts of vertices. For each $r\in [R_T]$, the $r$-th part corresponds to the randomness $r$ in the locally testing algorithm. It is a subset of $\left(\Sigma_2^t\right)^q$ and consists of all accepting configurations of the $q$ symbols queried by the locally testing algorithm of $C^{\odot t}$ under $r$.%, which . 

\medskip\noindent\textbf{Edges of $G$.} 
Intuitively, we construct edges so that a large clique in $G$ corresponds to a word close to the codeword $C^{\odot t}(x_1,\dots,x_{k})\in \left(\Sigma_2^t\right)^{k'}$, where $(x_1,\dots,x_{k})\in \left(\Sigma_1^t\right)^k$ is a solution of the vector 2CSP instance.
For illustration, the edges are specified by removing edges from a complete graph as follows.
\begin{enumerate}[(1)]
    \item First, note that each vertex specifies the values of $q$ indices in the word. We remove all edges between inconsistent vertices, i.e., two vertices that specify different values to the same index.
    \item For every two consistent vertices $u,v$, they specify the values of $\le 2q$ indices $I\subseteq [k']$. Then, we remove the edge between $u$ and $v$ if, under some randomness $r\in [R_D]$, both of the two indices queried by algorithm decoding $x_i$ (or $x_i+x_j$) fall into $I$, but the decoding result violates the constraint $x_i\in S_i$ (or $x_i+x_j\in S_{ij}$).
\end{enumerate}
We first remove inconsistent edges in (1) to ensure any large clique of $G$ induces a word $w\in \left(\Sigma_2^t\right)^{k'}$ that passes a large fraction of local tests. By the local testability of $C^{\odot t}$, $w$ should be close to some codeword $C^{\odot t}(x_1,\dots,x_k)$. We then apply the 2-query local decodability  on $w$ to ensure every $x_i\in S_i$ and every $x_i+x_j\in S_{ij}$, so that $x_1,\dots,x_{k}$ is a solution.

\noindent\textbf{Completeness.} Suppose the vector 2CSP has a solution $x_1,\dots,x_{k}$. For each of the $r\in [R_T]$ parts, we can choose the unique vertex in this part that is consistent with $C^{\odot t}(x_1,\dots,x_{k})$. Since both the testability and decodability have perfect completeness, the clique size is exactly $R_T$.

\noindent\textbf{Soundness.} For soundness analysis, we follow our intuition above. If there is a clique $D$ with size $|D|\ge \varepsilon_TR_T$, we can construct the solution of the vector 2CSP instance as follows. 

First, by combining Lemma~\ref{lem:ltc-num-var-intro} and the local testability, there is a codeword $C^{\odot t}(x_1,\dots,x_{k})$ such that, for at least $1-3\delta$ fraction of indices $I\subseteq [k']$, and every index $i\in I$, the value of the $i$-th index is not only specified by some vertex in $D$ but also consistent with this codeword.

Second, by the smoothness of decoding and the union bound, the probability that some index queried by the decoding of $x_i$ (or $x_i+x_j$) falls out of $I$ is at most $1-12\delta$. Thus, if $\delta<\frac1{12}$, there must be some randomness $r\in [R_D]$, so that under randomness $r$, both of the two queried indices for decoding fall into $I$. By the second step in the edge-removing procedure, we must have that $x_i\in S_i$ (or $x_i+x_j\in S_{ij}$). Hence, $(x_1,\dots,x_{k})
$ is a solution.
\end{proof}

\begin{remark}\label{rem:restriction}
 The restrictions of PLTDCs mainly fall into the aspect of local decodability in two folds.
 \begin{enumerate}
     \item We require such codes to be able to decode not only each input symbol but also the sum of two input symbols. The restriction is not essential and is only to make our code compatible with the intermediate vector 2CSP problem, which sets up constraints $x_i\in S_i$ and $x_i+x_j\in S_{ij}$. If we choose another intermediate problem, the requirement for its decoding ability could be changed correspondingly.
     \item  We require two-query decoding. This requirement also appears in previous works~\cite{Lin21,DBLP:conf/icalp/LinRSW22,DBLP:conf/coco/SK22}. As stated in~\cite{DBLP:conf/icalp/LinRSW22}, this is because we can encode the two-query decoding into an edge in the graph. If the decoding algorithm queries $q$ positions, we will need $q$-ary hyperedges to encode the decoding procedure, making the output problem of the reduction as constant-gap $k$-{\sc HyperClique} instead of \kclique{}. Moreover, It is still open how to do gap-preserving reductions from $k$-{\sc HyperClique} into \kclique{}.
 \end{enumerate}
\end{remark}

\subsection{Constructing a More Efficient PLTDC}\label{sec:construct-intro}
Hadamard codes have exponential blow-up in the codeword length, making the lower bound no tighter than $f(k)n^{\omega(\log k)}$~\cite{Lin21,DBLP:conf/icalp/LinRSW22,DBLP:conf/coco/SK22}. A natural direction for improvement is to consider the Reed-Muller code, the polynomial extension of the Hadamard code. To apply our framework, we need to prove that the RM code is a PLTDC, which has the following challenges.
\begin{itemize}
    \item For testability, we need to test whether the given word is close to a degree $d$ polynomial instead of linear functions.
    \item For decodability, the folklore interpolation-based decoding algorithm of degree $d$ polynomials requires $d+1$ queries on the codeword. However, the PLTDC requires 2-query decoding.
\end{itemize}

For parallel testing of low-degree polynomials, the previous work~\cite{DBLP:conf/icalp/LinRSW22} extends the classic Rubinfield-Sudan test~\cite{rs-test} that works on constant degree polynomials into the parallel case. In this work, to handle polynomials with a super constant degree, we extend the classic line vs. point test~\cite{fstest} into the parallel case (Lemma~\ref{lem:line-vs-point}).

For parallel decoding, we borrow the idea of the derivative code~\cite{WY07} to reduce the query number in the decoding procedure. As a simplified illustration, below, we present the essential ideas of the decoding procedure and show how to use partial derivatives to decode $f(\mathbf{x}_0)$ where $f(x_1,\dots,x_m): \mathbb F^m\to \mathbb F$ is a $m$-variate degree-$3$ polynomial, and $\mathbf{x}_0$ is a point in $\mathbb F^m$. 

For each $\mathbf{x}\in \mathbb F^m$, suppose we can get access to not only $f(\mathbf{x})$ but also all partial derivatives $\frac{\partial f}{\partial x_1}(\mathbf{x}),\dots,\frac{\partial f}{\partial x_m}(\mathbf{x})$. Then, to decode $f(\mathbf{x}_0)$, we randomly select a direction $\mathbf{h}\in \mathbb F^m$, and consider the univariate degree-3 polynomial $g(\lambda):=f(\mathbf{x}_0+\lambda\mathbf{h})$. We can represent $g(\lambda)$ as $\sum_{i=0}^{3}c_i\lambda^i$. We then sample two different values $(\lambda_1,\lambda_2)\in \binom{\mathbb F\setminus\{0\}}{2}$ at random, and query the values and all partial derivatives at two coordinates $\mathbf{x}_0+\lambda_1\mathbf{h}$ and $\mathbf{x}_0+\lambda_2\mathbf{h}$. According to the chain rule, we can get $g(\lambda_1), g'(\lambda_1), g(\lambda_2)$ and $g'(\lambda_2)$, which leads to the following equation:
\[\begin{pmatrix}g(\lambda_1)\\ g'(\lambda_1)\\ g(\lambda_2)\\ g'(\lambda_2)\end{pmatrix} = \begin{pmatrix} 1 & \lambda_1& \lambda_1^2 & \lambda_1^3\\0 & 1& 2\lambda_1 & 3\lambda_1^2\\1 & \lambda_2& \lambda_2^2 & \lambda_2^3\\0 & 1& 2\lambda_2 & 3\lambda_2^2 \end{pmatrix} \begin{pmatrix}c_{0}\\ c_{1}\\ c_{2}\\ c_{3}\end{pmatrix}.\]
The $4\times 4$ matrix has determinant $(\lambda_2-\lambda_1)^4\neq 0$ and is thus invertible. As a result, we obtain all coefficients of $g$ and can compute $f(\mathbf{x}_0)$ by $g(0)$. We can verify that the decoding is smooth.

In this work, we generalize the idea above into decoding $f(\mathbf{x})$ where $f$ is a polynomial of arbitrary odd degree $2r+1$. The decoding is similar but requires the codeword to incorporate all higher-order partial derivatives of order $\le r$ for each point $\mathbf{x}\in \mathbb F^m$. Furthermore, we consider the parallel decoding and prove the decoding procedure above adapts to the parallel case.

There are still two gaps for PLTDC. (i) We need extra testing to check the partial derivatives. We use directional derivatives to test this. 
(ii) To further support decoding the sum of two input symbols, we make some slight manipulations.

We implement the idea above into two codes. In Section~\ref{sec:pldtc1}, we extend the derivative code into polynomials with a constant degree using the parallel version of the Rubinfield-Sudan test. In Section~\ref{sec:pldtc2}, we then extend the code into super constant degrees using the parallel version of the line vs. point test. Below, we present the code that works for polynomials with a super constant degree.

\begin{theorem}[Theorem~\ref{thm:pldtc2}]\label{thm:PLTDC-intro}
For every integer $m,r,k>0$ and finite field $\mathbb F$ with a prime number size, let $d=2r+1$, if these parameters satisfy $|\mathbb F|> \max\{3d,9216\}$ and $k\le \binom{m+d}{d}\le |\mathbb F|^m$, then there is a PLTDC $C: \mathbb F^{k}\to \left(\mathbb F^{(d+1)\times \binom{2m+r}{r}}\right)^{\mathbb F^{4m}}$ that can be constructed explicitly in time depending only on $|\mathbb F|, m, d, k$.
\end{theorem}

\subsection{Improved Hardness of Constant-Gap \kclique{}}\label{sec:hardness-intro}

In the end, we prove Theorems~\ref{thm:intro-A} and \ref{thm:intro-B}. Note that Theorem~\ref{thm:intro-B} follows from Theorem~\ref{thm:intro-A} with a standard disperser argument to amplify the gap, which is the same as~\cite{CCK+17,DBLP:conf/icalp/LinRSW22}. Thus, we only present the proof of Theorem~\ref{thm:intro-A} below.
\begin{proof}[Proof of Theorem~\ref{thm:intro-A}]
The proof plugs the code~(Theorem~\ref{thm:PLTDC-intro}) into our reduction framework~(Theorem~\ref{thm:reduction-intro}) and obtains a stronger lower bound of constant-gap \kclique{} under ETH. 

In detail, we first fix an arbitrary sufficiently large $d$, then set $m=2^{100}\cdot d$ and $k = \binom{m+d}{d}$. We choose $\mathbb F$ to be the smallest field with prime size $> 3d$. By Bertrand's postulate, we have $3d+1\le |\mathbb F|\le 6d$. Since $(\frac{m}{d})^d\le \binom{m+d}{d}\le \left(\frac{e(m+d)}{d}\right)^d$, we have that $d=\Theta(\log k)$ and $|\mathbb F|^m\ge \binom{m+d}{d}$. Hence, we invoke Theorem~\ref{thm:PLTDC-intro} and construct the code $C$ who has block length $k'= |\mathbb F|^{4m}$ and alphabet size $|\Sigma_2|=|\mathbb F|^{(d+1)\binom{2m+r}{r}}$. Note that we have $k'= k^{\Theta(\log \log k)}$ and $|\Sigma_2| = |\mathbb F|^{O(k^{\mu})}$ for $\mu= 0.54$. 

By Lemma~\ref{lem:vcsp-intro}, under ETH, vector 2CSP problem with $k$ variables and dimension $t=\Theta\left(\frac{\log n}{\log |\mathbb F|}\right)$ over the field $\mathbb F$ has lower bound $f(k)\cdot n^{\Omega(k)}$. We then invoke Theorem~\ref{thm:reduction-intro} to obtain a constant-gap {$k'$-\textsc{Clique}} instance with size $n^{O(k^{\mu})}$ and $k'=k^{\Theta(\log \log k)}$ parts. As a result, assuming ETH, we establish the lower bound $f(k)n^{k^{\Omega(1/\log \log k)}}$ for constant-gap \kclique{}.
\end{proof}
\section{Preliminaries}
In this section, we introduce the problems and hypotheses. Our central focus is the approximation of parameterized \kclique{} problem (Definition \ref{def:kclique}) and parameterized 2CSP problem (Definition \ref{def:2csp}). Below, we first provide some background on parameterized complexity theory, then define the two problems. At the end, we introduce the \emph{exponential time hypothesis} (ETH, Hypothesis \ref{assumption:eth}), which our main results are based on.

In parameterized complexity theory, we consider a language $L\subseteq\{0,1\}^*$ equipped with a computable function $\kappa : \{0,1\}^*\to\mathbb{N}$, which returns a parameter $\kappa(x)\in\mathbb{N}$ for every input instance $x\in \{0,1\}^*$. 
A parameterized problem $(L,\kappa)$ is \emph{fixed parameter tractable} (FPT) if it has an algorithm,  which for every input $x\in\{0,1\}^*$, decides if $x\in L$ in $f(\kappa(x))\cdot |x|^{O(1)}$ time for some computable function $f$. An FPT reduction from problem $(L,\kappa)$ to $(L',\kappa')$ is an algorithm $A$ which on every input $x$, outputs an instance $x'$ in $f(\kappa(x))\cdot |x|^{O(1)}$ time for some computable function $f$, such that
\begin{itemize}
    \item there exists a universal computable function $g$ such that $\kappa'(x') = g(\kappa(x))$;
    \item $x'\in L'$ if and only if $x\in L$.
\end{itemize}

There are two important complexity classes in the parameterized regime, namely $\mathsf{FPT}$ and $\mathsf{W[1]}$. The class $\mathsf{FPT}$ consists of problems that admit FPT algorithms, and $\mathsf{W[1]}$ consists of problems that admit non-deterministic FPT algorithms.
The classes $\mathsf{W[1]}$ and $\mathsf{FPT}$ are analogies of $\mathsf{P}$ and $\mathsf{NP}$ in classic complexity theory. Similarly, a problem $(L,\kappa)$ is $\mathsf{W[1]}$-hard if every problem in $\mathsf{W[1]}$ can be FPT-reduced to it. A problem is $\mathsf{W[1]}$-complete if it is in $\mathsf{W[1]}$ and also $\mathsf{W[1]}$-hard.

Now we define the \kclique{} problem and the 2CSP problem. Throughout our paper, we only consider their parameterized versions.

\begin{definition}[\kclique{}]\label{def:kclique}
An instance of parameterized \kclique{} problem is a graph $G = (V,E)$, where the vertex set is partitioned into $k$ disjoint parts, each of which is an independent set. In other words, $V = V_1\dot\cup V_2\dot\cup \ldots \dot\cup V_k$ and $E$ only contains edges that cross two different parts. Without loss of generality, we assume each $V_i$ has an equal size $n$. The parameter is the number of disjoint parts $k$. The exact version of the \kclique{} problem is to decide whether $G$ has a clique of size $k$ or not. Given a function $\gamma:\mathbb N \to (0,1)$, the $\gamma(k)$-gap \kclique{} problem is to distinguish between the two cases that $G$ has a clique of size $k$, or all cliques in $G$ have size no more than $k/\gamma(k)$.
\end{definition}

\begin{definition}[2CSP]\label{def:2csp}
An instance of parameterized arity-2 constraint satisfaction problem (2CSP) is a tuple $\Pi=(X,\Sigma,\Phi)$, where:
\begin{itemize}
\item $X=\{x_1,\ldots,x_k\}$ is the set of variables;
\item $\Sigma$ is the alphabet.
\item $\Phi=\{\phi_1,\ldots,\phi_m\}$ is the set of constraints. Each constraint $\phi_i$ is a function $\Sigma^{\le 2}\to\{0,1\}$ on either one variable $x_{t_i}$ or two variables $x_{t_{i,1}},x_{t_{i,2}}$, where $t_i,t_{i,1},t_{i,2} \in [k]$. A constraint $\phi_i$ is satisfied if $\phi_i(x_{t_i})=1$ (for unary constraints), or $\phi_i(x_{t_{i,1}},x_{t_{i,2}})=1$ (for binary constraints).
\end{itemize}
A 2CSP instance $\Pi$ is satisfiable if there exists an assignment $\sigma: X \to \Sigma$ such that all constraints in $\Phi$ are satisfied.
The goal of the 2CSP problem is to decide whether $\Gamma$ is satisfiable or not. We use $n$ to denote the alphabet size, and the parameter $k$ is the number of variables.
\end{definition}

To introduce the exponential time hypothesis (ETH), we first define the 3SAT problem.

\begin{definition}[3SAT]\label{def:3sat}
A 3CNF formula $\varphi$ is a conjunction of $m$ disjunctive formulas over $n$ variables, where each disjunctive formula, termed as a clause, is on three variables or their negations. In other words, $\varphi$ is of the form $(x_{1,1}\lor x_{1,2}\lor x_{1,3})\land (x_{2,1}\lor x_{2,2}\lor x_{2,3})\land \ldots \land (x_{m,1}\lor x_{m,2}\lor x_{m,3})$, where each $x_{i,j} (i\in [m], j\in [3])$ is either one of the $n$ variables or its negation. The goal of the 3SAT problem is to decide whether $\varphi$ is satisfiable or not. 
\end{definition}

\begin{hypothesis}[ETH\cite{IP01}]
\label{assumption:eth}
There is no algorithm which can solve 3SAT in $2^{o(n)}$ time.
\end{hypothesis}

In the end, we present the hardness of \kclique{}.
\kclique{} is the canonical $\mathsf{W[1]}$-complete problem. Moreover, 
by applying the standard reduction~\cite{ChenHKX06} from 3SAT to \kclique{}, we have that: 
\begin{lemma}[\cite{ChenHKX06}]\label{thm:exact-clique-eth}
Assuming ETH, there is no algorithm which can solve \kclique{} in $f(k)n^{o(k)}$ time, for any computable function $f$.
\end{lemma}

\section{Vector 2CSP}
\label{sec:bvcsp}
In this section, we introduce vector 2CSP, an intermediate problem used in proving hardness of gap \kclique{}. We first define this special 2CSP (Definition~\ref{def:bvcsp}), then show it is $\wone$-hard (Corollary~\ref{cor:w1}) and has $|\Sigma|^{\Omega(k)}$ lower bound under ETH (Corollary~\ref{cor:eth}), where $\Sigma$ and $k$ stand for the alphabet and the number of variables, respectively.

\begin{definition}[Vector 2CSP]
\label{def:bvcsp}
    A vector 2CSP instance $(X,\Sigma,\Phi)$ is a special 2CSP instance, where:
\begin{itemize}
    \item The variables are $X=\{x_1,\ldots,x_k\}$.
    \item $\Sigma = \mathbb F^d$. The alphabet consists of all $d$-dimensional vectors over a field $\mathbb F$.
    \item $\Phi=\{\phi_i\}_{i=1}^k \bigcup \{\phi_{i,j}\}_{1\le i<j\le k}$ consists of $k+\binom{k}{2}$ constraints, where:
    \begin{itemize}
        \item For each $i\in [k]$, the constraint $\phi_i$ concerns a single variable $x_i$ and is of the form $x_i\in S_i$, where $S_i$ is a prescribed subset of $\mathbb F^d$. 
        \item For each $1\le i<j\le k$, the constraint $\phi_{i,j}$ concerns two variables $x_i,x_j$ and is of the form $x_i+x_j\in S_{i,j}$, where $S_{i,j}$ is also a prescribed subset of $\mathbb F^d$.
    \end{itemize}
\end{itemize}
\end{definition}

Below, we present a linear reduction from \kclique{} to the satisfiability of vector 2CSPs, thus prove the $\mathsf{W}[1]$-hardness and tight ETH-based lower bound for this problem. 
\begin{theorem}\label{thm:clique2vcsp}
    There is an FPT algorithm which, on input a \kclique{} instance $G=(V=V_1\dot\cup V_2\dot\cup \ldots \dot\cup V_k,E)$ with $n$ vertices, outputs a vector 2CSP instance $\Pi=(X,\Sigma,C)$ such that:
    \begin{itemize}
        \item $|X|=k$, $\Sigma = \mathbb F^{d}$ where $\mathbb F$ is an arbitrary finite field, and $d=\frac{5\cdot \log n}{\log |\mathbb F|}$.
        \item(Completeness) If $G$ has a clique of size $k$, then $\Pi$ is satisfiable. 
        \item(Soundness) If $G$ does not have any clique of size $k$, then $\Pi$ is not satisfiable. 
    \end{itemize}
\end{theorem}

Our reduction relies on the following lemma.
\begin{lemma}\label{lem:hash}
Given any finite field $\mathbb F$ and a \kclique{} instance $G=(V=V_1\dot\cup V_2\dot\cup \ldots \dot\cup V_k,E)$ with $n$ vertices, we can construct a hash function $H: V\to \mathbb F^{d}$, where $d = 5\cdot \frac{\log n}{\log |\mathbb F|}$, such that: 
\begin{itemize}
    \item For each $i\in [k]$ and distinct $v_i,v_i'\in V_i$, $H(v_i)\neq H(v_i')$.
    
    \item For any distinct $i,j\in[k]$ and distinct pairs $(v_i,v_j),(v_i',v_j')\in 
    V_i\times V_j$, $H(v_i)+H(v_j)\neq H(v_i')+H(v_j')$.
\end{itemize}
\end{lemma}

\begin{proof}
We present a random construction which succeeds with probability at least $1-kn^{-3}-k^2n^{-1}$. This construction can be easily derandomized using the conditional expectation technique similar to previous work~\cite{Lin21}.

The construction of $H$ is as follows. For each $i\in [k]$, we pick a $d \times \frac{\log n}{\log |\mathbb F|}$ matrix $M_i$ uniformly at random. For every $v\in V_i$, we treat it as a vector in $\mathbb F^{\frac{\log n}{\log |\mathbb F|}}$, and define $H(v):=M_iv\in \mathbb F^{d}$. Below, we show that with probability at least $1-kn^{-3}-k^2n^{-1}$, $H$ satisfies the two conditions above. 

The probability that the first property does not hold is upper bounded by 
    \begin{align*}
        &\Pr\left[\exists 1\le i\le k, \exists v_i,v_i'\in V_i, v_i\ne v_i'\land M_iv_i = M_iv_i'\right]\\\le\  &k |V_i|^2\cdot \Pr\left[v_i\ne v_i'\land M_iv_i = M_iv_i'\right]\\=\  &k |V_i|^2\cdot \Pr\left[v_i\ne v_i'\land M_i(v_i-v_i')=0\right]\\ \le\  &k n^2\cdot |\mathbb F|^{-d}\le kn^2\cdot n^{-5}=kn^{-3}.
    \end{align*}

The probability that the second property does not hold is upper bounded by
    \begin{align*}
        &\Pr[\exists 1\le i<j\le k, \exists (v_i,v_j),(v_i',v_j')\in V_i\times V_j, \\
        &(v_i,v_j)\ne (v_i',v_j')\land M_iv_i+M_jv_j = M_iv_i'+M_jv_j']\\
        \le\  &\binom{k}{2} n^4\cdot \Pr\left[(v_i,v_j)\ne (v_i',v_j')\land M_iv_i+M_jv_j= M_iv_i'+M_jv_j'\right]\\=\  &\binom{k}{2} n^4\cdot \Pr\left[(v_i,v_j)\ne (v_i',v_j')\land \begin{pmatrix} M_i & M_j\end{pmatrix}\cdot \begin{pmatrix} v_i-v_i' \\ v_j-v_j'\end{pmatrix}=0\right]\\ \le\  &\binom{k}{2} n^4\cdot n^{-5}= \binom{k}{2}n^{-1}\le k^2n^{-1}.
    \end{align*}
\end{proof}

Then, we apply Lemma~\ref{lem:hash} to prove Theorem~\ref{thm:clique2vcsp}.
\begin{proof}[Proof of Theorem~\ref{thm:clique2vcsp}]
Let $H$ be the hash function constructed in Lemma~\ref{lem:hash}. For each $i\in [k]$, we construct a set $T_i$ which encodes the vertex set $V_i$: 
\[T_i := \{H(v_i)\mid v_i\in V_i\}.\]
For every $1\le i<j\le k$, we construct a set $U_{i,j}\subseteq \mathbb F^d$ which encodes the edges between $V_i$ and $V_j$: 
\[U_{i,j} := \{H(v_i) + H(v_j)\in \mathbb F^d \mid v_i\in V_i, v_j\in V_j, (v_i,v_j) \in E\}.\] 

For each $i\in [k]$, we add the constraint $x_i\in T_i$, and for each $1\le i<j\le k$, we add a constraint $x_i+x_j\in U_{i,j}$. 

For the completeness, suppose there are $k$ vertices $v_1\in V_1,\ldots, v_k\in V_k$ that form a $k$-clique, then we set $x_i= H(v_i)$. They clearly satisfy all the constraints above.

For the soundness, we fix an arbitrary assignment $\{x_1,\dots,x_k\}$. We show that if this assignment satisfies all the constraints, then there is a $k$-clique in the original graph $G$. First, for each $i\in [k]$, by the construction of $T_i$ and the first item of Lemma~\ref{lem:hash}, there exists a unique vertex $v_i\in V_i$ such that $x_i = H(v_i)$.
Next, for each different $i,j \in [k]$, we have that $H(v_i)+H(v_j)\in U_{i,j}$. By the construction of $U_{i,j}$ and the second item of Lemma~\ref{lem:hash}, we conclude that there is an edge between $v_i\in V_i$ and $v_j\in V_j$. Thus $\{v_1,\ldots,v_k\}$ is a size-$k$ clique in $G$. 
\end{proof}

Given Theorem \ref{thm:clique2vcsp}, we immediately have the following corollaries:
\begin{corollary}\label{cor:w1} For any integer $k$ and any finite field $\mathbb F$, vector 2CSP with $k$ variables 
and dimension $\Theta(\frac{\log n}{\log |\mathbb F|})$ is $\wone$-hard. 
\end{corollary}

\begin{corollary}\label{cor:eth}
Assuming ETH, for any integer $k$, any finite field $\mathbb F$ and  any computable function $f$, vector 2CSP problem with $k$ variables and dimension $\Theta(\frac{\log n}{\log |\mathbb F|})$ admits no algorithm with running time $f(k)\cdot n^{o(k)}$.
\end{corollary}

\section{Parallel Locally Testable and Decodable Codes}
\label{sec:code}

In this section, we first present some backgrounds in coding theory, then define the $\ell$-parallelization of an error-correcting code $C$ to capture the parallel setting (Definition \ref{def:parallel}), and extend classic notions of locally decodable codes and locally testable codes respectively into the $\ell$-parallelization case (Definition \ref{def:parallel-decode} and \ref{def:parallel-test}). After that, we define \emph{parallel locally testable and decodable codes} (PLTDC, Definition \ref{def:nice}), and illustrate two examples of PLTDCs in Section \ref{sec:hard-pltdc} and \ref{sec:2pir}.
We will give the construction of a more rate-efficient parallel locally testable and decodable code in Section~\ref{sec:eccgood}. 

\subsection{Motivation}
Error-correcting code is a key ingredient in the proof of classic PCP theorem, and an important combinatorial tool to create a gap. To prove the PCP theorem, a natural strategy is to design a locally testable and decodable code $C:\Sigma^k \to \Sigma^{k'}$ for an $\mathsf{NP}$-complete language $L$, and require the prover to provide $C(\mathbf x)$ for some certificate $\mathbf x$. By locally reading a constant number of positions in $C(\mathbf x)$, we can verify whether $\mathbf x$ is a true certificate or not.

In the parameterized regime, the alphabet $\Sigma$ of a certificate is as large as $n$, which makes it hard to design an efficient locally  testable and decodable code. Take the Hadamard code $H:\mathbb F^k\to \mathbb F^{|\mathbb F|^k}$ as an example. If we choose $|\mathbb F|$ to be greater than $n$, then the block length, which corresponds to the number of variables in the gap instance, would be $n^k$, and thus the reduction will not be FPT.

To address this issue, we come up with the notion of parallel locally testable and  decodable codes. Suppose we have a locally testable and decodable code $C:\mathbb F_2^k\to \mathbb F_2^{k'}$, and want to encode a certificate $\mathbf x\in \Sigma^k$ over a large alphabet $\Sigma$. Think of $\Sigma=\mathbb F_2^\ell$ where $\ell=\log_{2}|\Sigma|$. For each $i \in [\ell]$, we concatenate the $i$-th bit of $\mathbf x[1],\ldots,\mathbf x[k]$ to a string $\mathbf x^i$ in $\mathbb F_2^k$, and encode it to a codeword in $\mathbb F_2^{k'}$ using $C$. For each $i \in [k']$, we combine the $i$-th symbol of the $\ell$ codewords, $C(\mathbf x^1)[i],\ldots,C(\mathbf x^\ell)[i]$, back to a symbol in $\Sigma$, and concatenate them to obtain a codeword in $\Sigma^{k'}$. Denote this parallelized code as $C^{\odot \ell}:\Sigma^k \to \Sigma^{k'}$. To decode a symbol $\mathbf x[i]$, we naturally want to decode its $\ell$ bits $\mathbf x^1[i],\ldots,\mathbf x^\ell[i]$ from $C(\mathbf x^1),\ldots,C(\mathbf x^\ell)$ respectively and combine them. However, to make the decoding local, we can only read $q$ symbols in $C^{\odot \ell}(\mathbf x)$, which give us access to the same $q$ positions in $C(\mathbf x^1),\ldots,C(\mathbf x^\ell)$. As the codewords may be corrupted and we want the $q$ queried positions in $C(\mathbf x^1),\ldots,C(\mathbf x^{\ell})$ to be all clean, the local test should also give us the promise that errors in $C(\mathbf x^1),\ldots,C(\mathbf x^{\ell})$ happen in a parallel way, i.e., take place at the same set of positions. 

\subsection{The Definition of PLTDCs} 
Before introducing PLTDCs, we first introduce some notations in coding theory. 

Given two strings $\mathbf x,\mathbf y\in \Sigma^n$, we denote by $\Delta(\mathbf x,\mathbf y)$ their relative hamming distance, i.e., $\Delta(\mathbf x,\mathbf y) = \frac{1}{n}|\{i|1\le i\le n, \mathbf x[i]\ne \mathbf y[i]\}|$. The distance of a string $\mathbf x\in \Sigma^n$ from a set of strings $S\subseteq \Sigma^n$ is denoted by $\Delta(\mathbf x,S) = \min_{\mathbf y \in S} \Delta(\mathbf x,\mathbf y)$. We say $\mathbf x$ is $\delta$-far from (respectively, $\delta$-close to) $S$ if $\Delta(\mathbf x,S) > \delta$ (respectively, $\Delta(\mathbf x,S)\le \delta$). 
 Below, we introduce error-correcting codes (ECCs), and define the $\ell$-parallelization of an error-correcting code $C$.
\begin{definition}[Error-correcting codes]
   A mapping $C: \Sigma_1^{k}\to \Sigma_2^{n}$ is an error-correcting code with message length $k$, block length $n$, and relative distance $\delta$ if for every distinct $\mathbf x,\mathbf y\in \Sigma_1^k$, we have $\Delta(C(\mathbf x),C(\mathbf y))\ge \delta$. We  use $\mathrm{Im}(C)$ to represent the image of $C$, i.e., $\mathrm{Im}(C) = \{C(\mathbf x)\mid \mathbf x\in \Sigma_1^k\}$.
\end{definition}

\begin{definition}[$\ell$-parallelization]
\label{def:parallel}
Given a ECC $C: \Sigma_1^{k}\to \Sigma_2^{n}$ and an integer $\ell$, the $\ell$-parallelization of $C$, denoted by $C^{\odot \ell}: (\Sigma_1^{\ell})^k\to (\Sigma_2^\ell)^{n}$, is defined as follows.
Given any $\mathbf{v}_1,\ldots,\mathbf{v}_{k}\in\Sigma_1^\ell$, for every $i\in [n]$,
define $\mathbf{u}_{i}\in \Sigma_2^\ell$ to be the vector such that for every $j\in [\ell]$:
\begin{equation}\label{eq:parallel}
\mathbf{u}_{i}[j]=C(\mathbf{v}_1[j],\ldots,\mathbf{v}_{k}[j])[i]\in \Sigma_2.
\end{equation}
Then, $C^{\odot \ell}(\mathbf{v}_1,\ldots,\mathbf{v}_{k})$ is defined as $(\mathbf{u}_1,\ldots,\mathbf{u}_n)\in (\Sigma_2^{\ell})^n$. 
\end{definition}

We put an illustration of $\ell$-parallelization in Figure~\ref{fig:eccfig}. The intuition is to encode a message with large alphabet $\Sigma$ using an error-correcting code which is defined over small alphabet $\Sigma_1$. Take $\ell=\lceil \log_{|\Sigma_1|}|\Sigma| \rceil$ and regard each character $v \in \Sigma$ as a vector $\mathbf v \in \Sigma_1^\ell$. Thus, a message $(v_1,\ldots,v_k) \in \Sigma^k$ is regarded as a list of vectors $(\mathbf v_1,\ldots,\mathbf v_k) \in (\Sigma_1^\ell)^k$. A codeword in $C^{\odot \ell}$ is also a list of vectors $(\mathbf u_1,\ldots,\mathbf u_n)\in (\Sigma_2^\ell)^n$. For every $j \in [\ell]$, we concatenate the $j$-th entry of the $\mathbf v_1,\ldots,\mathbf v_k$ to a string in $\Sigma_1^k$, encode it using $C$, disintegrate the codeword into $n$ characters in $\Sigma_2$, and put the $n$ characters to the $j$-th entry of $\mathbf u_1,\ldots,\mathbf u_n$, respectively.

Note that we can extend the definition of $\ell$-parallelization to all functions. Given a function $f: \Sigma^k\to \Sigma$ , its $\ell$-parallelization is a function $f^{\odot \ell}: (\Sigma^{\ell})^k\to \Sigma^{\ell}$.

Below we extend the notion of classic locally decodable codes into the $\ell$-parallelization case. Throughout our paper, in addition to decoding each symbol of the message, we also need to decode some linear combination of two symbols. Therefore, we slightly generalize the definition of local decodability. 
\begin{definition}[$(q,R,\delta,\varepsilon)$-(smooth) parallel local decodability]
\label{def:parallel-decode}
    Given a function $f:\Sigma_1^k \to \Sigma_1$, an error-correcting code $C:\Sigma_1^k \to \Sigma_2^{n}$ is said to be $(q,R,\delta,\varepsilon)$-parallel locally decodable with respect to $f$ if there exists a randomized  algorithm $\mathcal D_f$, 
    such that for every $\ell>0, \mathbf m \in (\Sigma_1^{\ell})^k$ and $\mathbf z \in (\Sigma_2^{\ell})^{n}$, 
    \begin{itemize}
        \item $\mathcal D_f$ draws a random number $r$ uniformly from $[R]$, queries $q$ positions of $\mathbf z$ according to the randomness $r$, and  outputs an element in $\Sigma_1^{\ell}$. Note that the selection of the positions to be queried is independent of $\mathbf z$.
       \item (Perfect completeness) If $\Delta(\mathbf z,C^{\odot \ell}(\mathbf m))=0$, then $\Pr[\mathcal D_f(\mathbf z)=f^{\odot \ell}(\mathbf m)] =1$.
        \item (Smoothness) Each of the $n$ symbols in $\mathbf z$ is equally likely to be queried, i.e., for each $1\le j\le n$, $\Pr[\mathcal D_f(\mathbf z)\text{ queries }\mathbf z[j]] = \frac{q}{n}$.
        
        \item (Local decodability)\footnote{Note that perfect completeness and smoothness imply local decodability with $\varepsilon=1-q\delta$. However, we still put local decodability into the definition for convenience.} If $\Delta(\mathbf z,C^{\odot \ell}(\mathbf m))\le \delta$, then $\Pr[\mathcal D_f(\mathbf z)=f^{\odot \ell}(\mathbf m)] \ge \varepsilon$.
    \end{itemize}
\end{definition}

Note that classic locally decodable codes correspond to our generalized locally decodable codes with $\ell=1$ and $f$ being dictator functions $\chi_i(x_1, x_2, \ldots, x_k) = x_i$. 

Next follows the definition of parallel locally testable codes. 

\begin{definition}[$(q,R,\delta,\varepsilon)$-parallel local testability]
\label{def:parallel-test}

An error-correcting code $C:\Sigma_1^k \to \Sigma_2^{n}$ is $(q,R,\delta,\varepsilon)$-parallel locally testable if and only if there exists a randomized algorithm $\mathcal T$, such that for any $\ell>0$ and any $\mathbf w \in (\Sigma_2^\ell)^{n}$, 
    \begin{itemize}
        \item $\mathcal T$ draws a random number $r$ uniformly from $[R]$, queries $q$ positions of $\mathbf w$ according to the randomness $r$, and outputs one bit.
        \item (Perfect completeness) If $\mathbf w\in \mathrm{Im}(C^{\odot \ell})$, then $\Pr[\mathcal T(\mathbf w)=1] = 1$.
        \item (Soundness) If $\Pr[\mathcal T(\mathbf w)=1] \ge \varepsilon$, then $\Delta(\mathbf w,\mathrm{Im}(C^{\odot \ell})) \le \delta$.
    \end{itemize}
\end{definition}
Note that classic locally testable codes correspond to the $\ell=1$ case, while parallel locally testable codes should work for every $\ell>0$, and the parameters $\delta$ and $\varepsilon$ are independent of $\ell$.

Now we are ready to define our key ingredient, namely \emph{parallel locally testable and decodable code} (PLTDC).
\begin{definition}[PLTDC]
\label{def:nice}
An error-correcting code $C:\Sigma_1^{k}\to\Sigma_2^{n}$ is a PLTDC if there exist constants $q>0$, $0<\delta_T,\delta_D, \varepsilon_T, \varepsilon_D<1$ and two functions $R_T(n),R_D(n) = n^{O(1)}$ such that:
\begin{enumerate}
\item $C$ is $(q,R_T,\delta_T,\varepsilon_T)$-parallel locally testable.
\item For any $i \in [k]$, let $\chi_i:\Sigma_1^{k}\to \Sigma_1$ be the $i$-th dictator function, i.e., $\chi_i(a_1,\ldots,a_{k})=a_i$. $C$ is $(2,R_D,\delta_D,\varepsilon_D)$-smooth parallel locally decodable with respect to $\chi_i$. In other words, we can locally decode any symbol of the message from the (corrupted) codeword in the parallel sense.

\item For any $1\le i<j\le k$, let $\psi_{i,j}:\Sigma_1^{k}\to \Sigma_1$ be the function that maps $(a_1,\ldots,a_k)$ to $a_i+a_j$. $C$ is $(2,R_D,\delta_D,\varepsilon_D)$-smooth parallel locally decodable with respect to $\psi_{i,j}$. In other words, we can locally decode any $a_i+a_j$ from the (corrupted) codeword of $(a_1,\ldots,a_k)$ in the parallel sense.

\end{enumerate}
\end{definition}

\subsection{Example I: Hadamard Code}\label{sec:hard-pltdc}
We show that the Hadamard code is a PLTDC.
Let $H:\mathbb{F}^{k}\to\mathbb{F}^{\mathbb{F}^{k}}$ be the Hadamard Code, i.e. for every $\mathbf{x},\mathbf{a}\in\mathbb{F}^{k}$,
\[
H(\mathbf{x})[\mathbf{a}]=\mathbf{a}^T\mathbf{x}.
\]
We treat the codeword $H(\mathbf{x})[\cdot]$ as a function from $\mathbb{F}^{k}$ to $\mathbb{F}$. Similarly, the codeword for $\mathbf x \in \mathbb F^{k \times \ell}$ in the $\ell$-parallelization of the Hadamard code can be viewed as a function from $\mathbb{F}^{k}$ to $\mathbb{F}^\ell$, mapping $\mathbf a \in \mathbb F^{k}$ to $\mathbf a^T\mathbf x$. We say a function $f$ is a homomorphism if $f(\mathbf a)+f(\mathbf b)=f(\mathbf a+\mathbf b)$. It is easy to see a function from $\mathbb F^{k}$ to $\mathbb F^\ell$ is a homomorphism if and only if it is of the form $f(\mathbf a)=\mathbf a^T\mathbf x$ for some $\mathbf x \in \mathbb F^{k \times \ell}$, thus if and only if $f$ is a codeword of $H^{\odot \ell}$. 

From \cite{blum1993self}, if a function $f$ satisfies $\Pr_{\mathbf{a},\mathbf{b}}[f(\mathbf{a})+f(\mathbf{b})=f(\mathbf{a}+\mathbf{b})]\ge 1-\frac{\delta}{2}$ for some $\delta$, then it is at most $\delta$-far from a homomorphism function $g$. This test shows that the Hadamard code is $\left(3,|\mathbb{F}|^{2k},\delta,1-\frac{\delta}{2}\right)$-parallel locally testable.

Next, we turn to prove that it is parallel locally decodable and satisfies  conditions 2--3 in Definition~\ref{def:nice}. 
If $f$ satisfies $\Pr_\mathbf{a}[f(\mathbf{a})=\mathbf a^T\mathbf x]\ge 1-\delta$ for some $\mathbf x \in \mathbb F^{k \times \ell}$, then for any $i\ne j\in [k]$, we have 
\[
\Pr_{\mathbf{a}}[\mathbf{x}[i]=f(\mathbf{a}+\mathbf{e}_i)-f(\mathbf{a})]\ge 1-2\delta
\]
and
\[
\Pr_{\mathbf{a}}[\mathbf{x}[i] + \mathbf{x}[j]=f(\mathbf{a}+\mathbf{e}_i+\mathbf{e}_j)-f(\mathbf{a})]\ge 1-2\delta,
\]
where $\mathbf x[i]\in \mathbb F^{\ell}$, the $i$-th row of $\mathbf x$, is the $i$-th symbol of the message, and $\mathbf{e}_i$ is the $i$-th unit vector. This shows that for every $i \ne j \in [k]$, Hadamard code is $(2,|\mathbb F|^{k},\delta,1-2\delta)$-smooth parallel locally decodable with respect to all $\chi_i$ and $\psi_{i,j}$.

\subsection{Example II: The Derivative Code from 
\cite{WY07}}\label{sec:2pir}

Below, we present another PLTDC constructed in \cite{WY07}, which is based on classic Reed-Muller codes \cite{Muller54}.

We fix a field $\mathbb F= \mathbb F_7$.
For every integer $m\ge 1$, we set $k = \binom{m+3}{3}$ and prescribe $k$ points $\mathbf{u}_1,\dots,\mathbf{u}_k\in \mathbb F^m$ such that any evaluation on these points uniquely determines an $m$-variate degree-3 polynomial. This is possible since the number of coefficients of such polynomials is $\binom{m+3}{3}=k$. We then define the code $C: \mathbb F^{k}\to \left(\mathbb  F^{2m+1} \right)^{|\mathbb F|^{2m}}$. Given a message $(a_1,\dots,a_k)\in \mathbb F^{k}$, its codeword $C(a_1,\dots,a_k)[\cdot,\cdot]$ is regarded as a function from $\mathbb F^{m}\times \mathbb F^{m}$ to $\mathbb F^{2m+1}$, which is defined as follows.
\begin{enumerate}
    \item Let $p:\mathbb F^m \to \mathbb F$ be the the unique $m$-variate degree-3 polynomial such that for every $i \in [k], p(\mathbf u_i)=a_i$, and let $q:\mathbb F^{m}\times \mathbb F^{m} \to \mathbb F$ be defined as $q(\mathbf{x},\mathbf{y}):=p(\mathbf{x})+p(\mathbf{y})$. Using $q$, we can obtain $a_i$ (resp. $a_i+a_j$) from $q(\mathbf{u}_i,\mathbf{0})$ (resp., $q(\mathbf{u}_i,\mathbf{u}_j)$), which reduces the task of decoding $a_i,a_i+a_j$ to the task of computing some $q(\mathbf{x},\mathbf{y})$.
    \item For every $\mathbf{x},\mathbf{y}\in \mathbb F^m$,
\begin{equation}\label{eq:construct3}
C(a_1,\dots,a_k)[\mathbf{x},\mathbf{y}]
:= \left(q,\frac{\partial q}{\partial x_1},\dots,\frac{\partial q}{\partial x_m},\frac{\partial q}{\partial y_1},\dots,\frac{\partial q}{\partial y_m}\right)(\mathbf{x},\mathbf{y})\in \mathbb F^{2m+1}
\end{equation}
which stores $q(\mathbf{x},\mathbf{y})$ and all its partial derivatives $\frac{\partial q}{\partial x_i}(\mathbf{x},\mathbf{y})$ and $\frac{\partial q}{\partial y_i}(\mathbf{x},\mathbf{y})$. 
\end{enumerate}

We first prove that $C$ is parallel locally testable. Let $f: \mathbb F^{m}\times \mathbb F^{m}\to \left(\mathbb F^{2m+1}\right)^{\ell}$ be a plausible codeword in the $\ell$-parallelization of $C$.

We treat each element $\sigma \in (\mathbb F^{2m+1})^{\ell}$ as an $\ell \times (2m+1)$ table, where for $i\in [\ell], j\in [2m+1]$, we use $\sigma[i,j]\in \mathbb F$ to denote its $(i,j)$-th entry. Note that $f(\cdot, \cdot)[i,j]$ is supposed to be a $(2m)$-variate polynomial with degree $\le 3$. 
The testing algorithm is as follows.
\begin{itemize}
    \item We first test whether $\forall i\in [\ell],j\in [2m+1], f(\cdot, \cdot)[i,j]$ is close to a polynomial of degree $\le 3$ in the parallel sense. We apply constant-degree parallel low-degree tests ~\cite{DBLP:conf/icalp/LinRSW22}. If the rejection probability is smaller than $\frac{\delta}{2}$ 
    for some $0<\delta<\frac{1}{25}$, then there exists a $(2m)$-variate polynomial tuple $\widetilde f(\cdot,\cdot)$ with degree $\le 3$ and a set $S\subseteq \mathbb F^m\times \mathbb F^m$ with $|S|\ge (1-\delta)|\mathbb F|^{2m}$, such that for every $i\in [\ell], j\in [2m+1]$, $f(\cdot,\cdot)[i,j]$ equals to $\widetilde f(\cdot,\cdot)[i,j]$ on $S$. Since the distance between any two different degree-$(\le 3)$ polynomials is at least $1-\frac{3}{|\mathbb F|} > 2\delta$, the polynomial tuple $\widetilde f(\cdot,\cdot)$ is unique.
    
    \item Next, we test whether $\forall i\in [\ell], 2\le j\le 2m+1, \widetilde f(\cdot, \cdot)[i,j]$ is a partial derivative of $\widetilde f(\cdot, \cdot)[i,1]$ as in (\ref{eq:construct3}). In detail, we randomly pick $\mathbf{x}, \mathbf{y}, \mathbf{d}_1, \mathbf{d}_2\in \mathbb F^m$ and query $4$ entries $\{f(\mathbf{x}+\lambda\mathbf{d}_1,\mathbf{y}+\lambda\mathbf{d}_2)\mid 0\le \lambda\le 3\}$. Using the queried results, we can interpolate a univariate degree-3 polynomial tuple $g(\lambda):=f(\mathbf{x}+\lambda\mathbf{d}_1,\mathbf{y}+\lambda\mathbf{d}_2)\in \mathbb F^{\ell\times (2m+1)}$. The algorithm accepts if $\forall i\in [\ell]$,
    \begin{equation}\label{eq:construct4}
        g'(0)[i,1] = \sum_{j=1}^{m} f(\mathbf{x},\mathbf{y})[i,j+1]\cdot \mathbf{d}_1[j]+
        \sum_{j=1}^{m} f(\mathbf{x},\mathbf{y})[i,j+m+1]\cdot \mathbf{d}_2[j].
    \end{equation}
    For completeness, assume $f=\widetilde f$ is a codeword and thus for every $i\in [\ell],2 \le j \le 2m+1$, $\widetilde f(\cdot, \cdot)[i,j]$ is is the partial derivative of $\widetilde f(\cdot, \cdot)[i,1]$ as in (\ref{eq:construct3}). Then both LHS and RHS of (\ref{eq:construct4}) compute the directional derivative of $\widetilde f$, with respect to direction $(\mathbf d_1,\mathbf d_2)$, on point $(\mathbf x,\mathbf y)$.
    
    \noindent For soundness, we first note that with probability $\ge 1-4\delta$,  all the queried points lie in $S$ and thus $f$ agrees with $\widetilde f$ on them. Note that for each $i \in [\ell],j \in [2m+1]$, $\widetilde f(\cdot,\cdot)[i,j]$ is a polynomial with degree $\le 3$. Thus, if $\exists i\in [\ell], 2\le j\le 2m+1$ such that $\widetilde f(\cdot,\cdot)[i,j]$ is not the partial derivative of $\widetilde f(\cdot,\cdot)[i,1]$ as in (\ref{eq:construct3}), then the LHS and RHS of (\ref{eq:construct4}) are different degree-$(\le 4)$ polynomials on variables  $\mathbf{x},\mathbf{y},\mathbf{d}_1,\mathbf{d}_2$. Therefore, the rejection probability is at least $1-\frac{4}{|\mathbb F|}-4\delta$ by Schwartz-Zippel Lemma.
    \item At the end, we test if for every $i \in [\ell]$, $\widetilde f(\mathbf{x},\mathbf{y})[i,1]$ is of the form $p(\mathbf{x})+p(\mathbf{y})$. This is equivalent to testing whether for every $\mathbf x,\mathbf y \in \mathbb F^m$:
    \begin{equation}\label{eq:construct5}
        \left\{
        \begin{aligned}
        \widetilde f(\mathbf{x},\mathbf{y})[i,1] & = \widetilde f(\mathbf{x},\mathbf{0})[i,1]+\widetilde f(\mathbf{0},\mathbf{y})[i,1] \\
        \widetilde f(\mathbf{x},\mathbf{0})[i,1] & = \widetilde f(\mathbf{0},\mathbf{x})[i,1]
        \end{aligned}
        \right..
    \end{equation} 
    Similarly, if we pick $\mathbf x,\mathbf y ,\mathbf d_1,\mathbf d_2\in \mathbb F^m$ uniformly at random and do local interpolation from $\{f(\mathbf x+\lambda \mathbf d_1,\mathbf y+\lambda \mathbf d_2) \mid 1 \le \lambda \le 4\}$, we can get $\widetilde f(\mathbf x,\mathbf y)$ with probability at least $1-4\delta$. We randomly pick $\mathbf x,\mathbf y \in \mathbb F^m$ and test whether (\ref{eq:construct5}) holds using $4 \times 4=16$ queries. 
    The analysis is the same as the previous bullet, so we omit the details.
\end{itemize}

Thus, the code is $\left(25,|\mathbb F|^{18m},\delta,1-\frac{\delta}{2}\right)$-parallel locally testable for every $0<\delta<\frac{1}{25}$. We now turn to prove the parallel local decodability. 
From our discussion above, to decode $\ell$-parallelized messages, it suffices to locally correct the $\ell$ function values $(f(\mathbf{x},\mathbf{y})[1,1],\dots, f(\mathbf{x},\mathbf{y})[\ell,1])$ for some $(\mathbf {x},\mathbf {y})$.

 We randomly pick $\mathbf{d}_1,\mathbf{d}_2\in \mathbb F^m$ and different $\lambda_1,\lambda_2\in \mathbb F\setminus \{0\}$, and query $f(\mathbf{x}+\lambda_1\mathbf{d}_1,\mathbf{y}+\lambda_1\mathbf{d}_2)$ and $f(\mathbf{x}+\lambda_2\mathbf{d}_1,\mathbf{y}+\lambda_2\mathbf{d}_2)$.
Consider the function $g(\lambda) := (f(\mathbf{x}+\lambda\mathbf{d}_1,\mathbf{y}+\lambda\mathbf{d}_2)[1,1],\ldots,f(\mathbf{x}+\lambda\mathbf{d}_1,\mathbf{y}+\lambda\mathbf{d}_2)[\ell,1])\in \mathbb F^\ell$. For each $i \in [\ell]$, $g(\cdot)[i]$ is supposed to be a univariate degree-3 polynomial. 
Formally, for every $i\in [\ell]$, there exists $c_{i,0},\dots,c_{i,3}\in \mathbb F$ such that $g(\lambda)[i] = \sum_{j=0}^{3} c_{i,j}\lambda^j$. 
Since the codeword  consists of point values of the function and its partial derivatives, we can obtain $g(\lambda_1)$, $g(\lambda_2)$ and $g'(\lambda_1)$, $g'(\lambda_2)$ by the chain rule. Thus, for every $1\le i\le \ell$, we solve the following linear equation system for $(c_{i,0},c_{i,1},c_{i,2},c_{i,3})$:
\[\begin{pmatrix}g(\lambda_1)[i]\\ g'(\lambda_1)[i]\\ g(\lambda_2)[i]\\ g'(\lambda_2)[i]\end{pmatrix} = \begin{pmatrix} 1 & \lambda_1& \lambda_1^2 & \lambda_1^3\\0 & 1& 2\lambda_1 & 3\lambda_1^2\\1 & \lambda_2& \lambda_2^2 & \lambda_2^3\\0 & 1& 2\lambda_2 & 3\lambda_2^2 \end{pmatrix} \begin{pmatrix}c_{i,0}\\ c_{i,1}\\ c_{i,2}\\ c_{i,3}\end{pmatrix}.\]
The $4\times 4$ matrix has determinant $(\lambda_2-\lambda_1)^4\neq 0$ and is thus invertible. As a result, we obtain all coefficients of $g$ and can compute $(f(\mathbf{x},\mathbf{y})[1,1],\dots, f(\mathbf{x},\mathbf{y})[\ell,1])$ by $g(0)$. 

It is straightforward that every point in $\mathbb F^m\times \mathbb F^m$ has an equal probability of being queried. Furthermore, if $\Delta(f, {\rm Im}(C^{\odot \ell}))<\delta$, then with probability at least $1-2\delta$, the 2 queried points $f(\mathbf{x}+\lambda_1\mathbf{d}_1,\mathbf{y}+\lambda_1\mathbf{d}_2)$ and $f(\mathbf{x}+\lambda_2\mathbf{d}_1,\mathbf{y}+\lambda_2\mathbf{d}_2)$ are not corrupted, and thus we can correctly decode   $(f(\mathbf{x},\mathbf{y})[1,1],\dots, f(\mathbf{x},\mathbf{y})[\ell,1])$. As a result, the code $C$ is $(2,|\mathbb F|^{2m+2}, \delta,1-2\delta)$-smooth parallel locally decodable for every $0<\delta<\frac{1}{2}$.

In Section~\ref{sec:eccgood}, we construct PLTDCs that are more rate-efficient than the two codes here. By applying our PLTDCs, we can derive tighter lower bounds for constant-gap \kclique{}.
\section{Reduction from Vector 2CSP to Gap \kclique{} via PLTDCs}
\label{sec:reduction-vcsp2gapclique}

Before giving the reduction from vector 2CSP to constant-gap \kclique{}, we introduce the following lemma about PLTDCs.

\begin{lemma}\label{lem:ltc_num_var}
Let $C:\Sigma_1^{k}\to \Sigma_2^n$ be a PLTDC such that
\begin{itemize}
    \item It is $(q,R_T,\delta,\varepsilon_T)$-parallel locally testable. 
    \item It is $(2,R_D,5\delta,\varepsilon_D)$-parallel locally decodable with respect to all dictator functions $\chi_i$.
\end{itemize}
Here $q>1,0<\delta,\varepsilon_D,\varepsilon_T<1$ are constants and $n>\frac{1}{\delta}$. For each $r \in [R_T]$, let $\text{var}(r)\subseteq [n]$ be the set of $q$ positions queried by the local testing algorithm $\mathcal T$ under randomness $r$. Then for any set $I \subseteq [R_T]$ with size $\ge \varepsilon_T R_T$, we have
\[\left|\bigcup_{r \in I} \text{var}(r)\right| \ge (1-2\delta)\cdot n.\]
\end{lemma}
\begin{proof}
    Suppose by contradiction that the statement is not true, i.e., for some $I \subseteq [R_T]$ with size $\ge \varepsilon_T R_T$, the set $\text{var}_I:=\bigcup_{r \in I}\text{var}(r)$ has size $<(1-2\delta)\cdot n$.
    We first show that the distance of code $C$ is at most $5\delta$, then show a code with small distance cannot be locally decodable.

    Take an arbitrary codeword $\mathbf c \in \text{Im}(C)$. We construct two strings $\mathbf s_1,\mathbf s_2 \in \Sigma_2^n$ by setting $\mathbf s_1[i]=\mathbf s_2[i]=\mathbf c[i]$ for every $i \in \text{var}_I$, and making $2\delta<\Delta(\mathbf s_1,\mathbf s_2)\le 2\delta+\frac{1}{n}<3\delta$. This is possible since $|\text{var}_I|<(1-2\delta)\cdot n$ and $\Delta(\mathbf s_1,\mathbf s_2)$ is an integer multiple of $\frac{1}{n}$. Then, since for any randomness $r \in I$, $\mathbf s_1$ and $\mathbf s_2$ agree with $\mathbf c$ on positions $\text{var}(r)\subseteq \text{var}_I$, and $\mathcal T$ accepts $\mathbf c$ with probability 1, we have
    \[\left\{\begin{aligned}
    \Pr_{r \in [R_T]}[\mathcal T(\mathbf s_1)=1]\ge \varepsilon_T\\
    \Pr_{r \in [R_T]}[\mathcal T(\mathbf s_2)=1]\ge \varepsilon_T 
    \end{aligned}\right.,\]
    which implies there exists $\mathbf c_1,\mathbf c_2 \in \text{Im}(C)$ satisfying $\Delta(\mathbf c_1,\mathbf s_1) \le \delta$ and $\Delta(\mathbf c_2,\mathbf s_2) \le \delta$. Note that $\mathbf c_1,\mathbf c_2$ must be different codewords since $\Delta(\mathbf s_1,\mathbf s_2) >2\delta$. Furthermore, by triangle inequality, 
    \[\Delta(\mathbf c_1,\mathbf c_2) \le \Delta(\mathbf c_1,\mathbf s_1)+\Delta(\mathbf s_1,\mathbf s_2)+\Delta(\mathbf s_2,\mathbf c_2) < \delta+3\delta+\delta = 5\delta.\]

    Now suppose $\mathbf c_1=C(\mathbf m_1)$ and $\mathbf c_2=C(\mathbf m_2)$ for some different messages $\mathbf m_1,\mathbf m_2 \in \Sigma_1^{k}$, and $\mathbf m_1[i]\neq \mathbf m_2[i]$ for some $i \in [k]$. By the perfect completeness of the local decoding algorithm $\mathcal D_{\chi_i}$, we have $\Pr[\mathcal D_{\chi_i}(\mathbf c_1)=\mathbf m_1[i]]=1$. However, since $\mathbf c_1$ is $(5\delta)$-close to $\mathbf c_2$, by the local decodability of $C$, we have $\Pr[\mathcal D_{\chi_i}(\mathbf c_1)=\mathbf m_2[i]] \ge \varepsilon_D>0$, a contradiction.
    
\end{proof}

\begin{theorem}\label{thm:vcsp2gapclique}
Let $C : \Sigma_1^{k}\to \Sigma_2^{k'}$ be a PLTDC such that
\begin{itemize}
    \item It is $(q,R_T,\delta,\varepsilon_T)$-parallel locally testable.
    \item It is $(2,R_D,5\delta,\varepsilon_D)$-smooth parallel locally decodable with respect to all $\chi_i$ and $\psi_{i,j}$.
\end{itemize}
Here $q>1,0<\varepsilon_D,\varepsilon_T<1, 0<\delta<\frac{1}{12}$ are constants, and $k',R_T,R_D$ satisfy  $\frac{1}{\delta}<k'\le R_T,R_D\le \left(k'\right)^{O(1)}$ and $\varepsilon_T R_T\ge 2$.
Then, there exists a reduction which takes as input any vector 2CSP instance $\Pi=(X,\Sigma,\Phi)$ with $|X|= k$, $\Sigma=\Sigma_1^{t}$, and outputs a graph $G=(V,E)$ with $|V| \le K\cdot |\Sigma_2|^{tq}$ and the following properties in $|V|^{O(1)}$-time, where $K=R_T$:
\begin{itemize}
\item If $\Pi$ is satisfiable, then $G$ contains a size-$K$ clique;
\item If $\Pi$ is not satisfiable, then $G$ contains no clique of size $\varepsilon_T K$.
\end{itemize}
\end{theorem}

\begin{proof}
Let $X=\{x_1,\ldots,x_{k}\}$ be the variable set in $\Pi$, and let $\sigma:X \to \Sigma_1^t$ be any assignment of $X$. The vertices in $G$ are supposed to represent the local testing results of the $t$-parallelized codeword $C^{\odot t}(\sigma(x_1),\ldots,\sigma(x_{k}))$. The edges in $G$ are used to check the consistency between different tests and to check whether the locally decoded value $\sigma(x_i)$ (resp. $\sigma(x_i)+\sigma(x_j)$) satisfy the constraints in $\Phi$.

Below, we present the construction of $G$ by specifying its vertices and edges, then prove the completeness and soundness.

\medskip
    \noindent \textbf{Vertices of $G$.}
    The vertex set $V$ consists of $R_T$ parts $V=V_1\dot\cup\ldots\dot\cup V_{R_T}$, one for each randomness of the $t$-parallelized local testing algorithm $\mathcal T$. For each randomness $r \in [R_T]$, $\mathcal T$ queries $q$ positions in the codeword based on $r$, and either accepts or rejects depending on the $q$ queried symbols in $\Sigma_2^t$. We let $V_r$ contain all accepting configurations of the $q$ symbols under randomness $r$, which is a subset of $(\Sigma_2^t)^q$. For each vertex $v \in V_r$, we use ${\rm var}(v)$ to denote the $q$ queried positions under randomness $r$, and use a  function $\pi_v:\text{var}(v)\to \Sigma_2^t$ to denote its partial assignment.
    
\medskip

    \noindent \textbf{Edges of $G$.}
    For illustration, we specify the edge set by removing edges from a complete graph. The procedure is as follows.
    \begin{enumerate}[(1)]
        \item We remove edges that violate the consistency. Specifically, we remove the edge between any $v,w \in V$ if there exists some position $i\in \text{var}(v) \cap \text{var}(w)$ such that $\pi_{v}(i) \neq \pi_{w}(i)$. Note that after this step, each $V_r$ is an independent set.
        
        \item We remove edges whose locally decoding result violates the constraints in $\Phi$. For any different $v,w \in V$, if edge $(v,w)$ still remains before this step, then $\pi_{v}$ and $\pi_{w}$ must be consistent on $\text{var}(v) \cap \text{var}(w)$. We use $\pi_{v\cup w}:\text{var}(v)\cup \text{var}(w) \to \Sigma_2^t$ to denote the union of two partial assignments. From the parallel locally decodability, we have a 2-query locally decoding algorithm for each $\chi_i$ and $\psi_{i,j}$. Suppose for some $\chi_i$ (resp. $\psi_{i,j}$) and some randomness $r \in R_D$, the decoding algorithm $\mathcal D_{\chi_i}$ (resp. $\mathcal D_{\psi_{i,j}}$) queries positions $a_1,a_2 \in [k']$. If $\{a_1,a_2\}\subseteq \text{var}(v) \cup \text{var}(w)$ and the decoding result based on $\pi_{v\cup w}(a_1)$ and $\pi_{v\cup w}(a_2)$ does not satisfy the constraint in $\Pi$, i.e., is not in $S_i$ (resp. $S_{i,j}$), then we remove the edge between $v$ and $w$. 

    \end{enumerate}

\medskip
    \noindent \textbf{Completeness.} Suppose $\Pi$ is satisfiable and $\sigma:X \to \Sigma_1^t$ is a solution. Let $\mathbf c=C^{\odot t}(\sigma(x_1),\ldots,\sigma(x_{k}))$ be the encoding of the solution. For each randomness $r \in [R_T]$, we pick the vertex $v \in V_r$ such that for every $i \in \text{var}(v)$, $\pi_v(i)=\mathbf c[i]$. By our construction and the perfect completeness of local testing and  decoding, it is easy to verify such $R_T$ vertices exist and form a clique of size $K=R_T$.
    
\medskip
    \noindent \textbf{Soundness.} Let $D$ be the maximal clique in $G$. Below we prove that if $|D|>\varepsilon_TK\ge 2$, then there is a satisfying assignment of $\Pi$. 

    We define \[\text{var}_D:=\bigcup_{v \in D} \text{var}(v).\] By item (1) in the edge description, the partial assignments corresponding to vertices in $D$ must be consistent on $\text{var}_D$. Thus, we can also define a function $\pi_{D}:\text{var}_D \to \Sigma_2^t$ such that for every $i \in \text{var}_D$, \[\pi_{D}(i):=\pi_{v}(i)\]
    where $v$ is a vertex in $V_r \cap D$ such that $\text{var}(r)$ contains $i$. We extend $\pi_D$ to the entire $[k']$ by setting $\pi_D(i):=\mathbf 0$ for every $i \not\in \text{var}_D$.

    Note that by our construction of the vertices, $\pi_D=(\pi_D(1),\ldots,\pi_D(k'))\in (\Sigma_2^t)^{k'}$ is a string which can pass the local testing with probability at least $\frac{|D|}{R_T}\ge \varepsilon_T$. Thus, since $C$ is $(q,R_T,\delta,\varepsilon_T)$-parallel locally testable, we have $\Delta(\pi_D,\text{Im}(C^{\odot t}))\le \delta$. In other words, there exists an assignment $\sigma':X \to \Sigma_1^t$ and a set $I$ of indices with $|I| \ge (1-\delta)\cdot k'$ such that
    \[\pi_D(i)=C^{\odot t}(\sigma'(x_1),\ldots,\sigma'(x_{k}))[i], \forall i \in I.\]
    We further define $I'=\text{var}_D \cap I$. By Lemma \ref{lem:ltc_num_var}, we have $|\text{var}_D| \ge (1-2\delta)\cdot k'$ and thus $|I'|\ge (1-3\delta)\cdot k'$.

    In the following, we prove $\sigma'$ is a satisfying assignment of $\Pi$ by showing that $\sigma'(x_i) \in S_i$ (resp. $\sigma'(x_i)+\sigma'(x_j) \in S_{i,j}$) for every $i,j \in [k]$. Consider the local decoding algorithm $\mathcal D_{\chi_i}$ (resp. $\mathcal D_{\psi_{i,j}}$). From the smoothness of $\mathcal D_{\chi_i}$ (resp. $\mathcal D_{\psi_{i,j}}$), the probability that the two queried positions $a_1,a_2\in[k']$ both lie in $I'$ can be lower bounded as:
    \begin{align*}
        & \Pr_{r \in [R_D]}[a_1\in I' \land a_2 \in I'] \\ = &\ 1-\Pr_{r \in [R_D]}[a_1 \notin I' \lor a_2 \notin I'] \\
        \ge &\ 1-2\sum_{i \notin I'} \Pr_{r \in [R_D]}[i\text{ is queried}] \\
        = &\ 1-2 (k'-|I'|)\cdot \frac{2}{k'} \\
        \ge &\ 1-12\delta>0.
    \end{align*}
    
    Thus, there exists some $r \in [R_D]$ such that the queried positions $a_1,a_2$ under randomness $r$ are both in $I'$. By the perfect completeness of $\mathcal D_{\chi_i}$ (resp. $\mathcal D_{\psi_{i,j}}$),  the decoding result from them is exactly $\sigma'(x_i)$ (resp. $\sigma'(x_i)+\sigma'(x_j)$). 
    Since $I'\subseteq \text{var}_D=\bigcup_{v \in D}\text{var}(v)$ and the clique size $\varepsilon_T K \ge 2$, there exists two different vertices $v,w \in D$ such that $a_1,a_2\in \mathrm{var}(v)\cup \mathrm{var}(w)$, and there is an edge between $v$ and $w$. Hence, by the item (2) in the edge description, we have that $\sigma'(x_i) \in S_i$ (resp. $\sigma'(x_i)+\sigma'(x_j) \in S_{i,j}$).

    Therefore, $\sigma'$ is a satisfying assignment of $\Pi$, contradicting the fact that $\Pi$ is not satisfiable.
\end{proof}
\section{Construction of PLTDCs}
\label{sec:eccgood}
In this section, we construct two new PLTDCs which generalize the Derivative code in \cite{WY07} (see Section \ref{sec:2pir}) from degree-3 polynomials to higher odd-degree polynomials. We first specify some notations in Section \ref{sec:pldtc-notation}. Then, we give a construction of the first PLTDC in Section \ref{sec:pldtc1}, which uses polynomials with constant degree. After that, we show the construction of the second PLTDC in Section \ref{sec:pldtc2}, which generalizes the first one by lifting the degree to super-constant, using a more sophisticated local testing regime.

\subsection{Notations}\label{sec:pldtc-notation}
Throughout this section, we consider $t$-tuple of $m$-variate degree-$d$ polynomials over a finite field $\mathbb F$ of prime size. Let $\mathbb{F}^{\le d}[x_1,\ldots,x_m]$ be the set of $m$-variable polynomials with degree at most $d$.
When $t,m,d,\mathbb F$ are clear in the context, we use 
\[
 \mathcal F= \left\{(f_1,\ldots,f_t) : \forall i\in [t], f_i\in\mathbb{F}^{\le d}[x_1,\ldots,x_m] \right\}
\]
to denote the set of such polynomial tuples. Furthermore, we use $\mathcal L^m=\{\ell(\mathbf x,\mathbf h)\left| \mathbf x,\mathbf h\in \mathbb F^m\right.\}$ to denote the set of all lines in $\mathbb F^m$, where $\ell(\mathbf x,\mathbf h)=\{\mathbf x+\lambda \mathbf h\left| \lambda \in \mathbb F\right.\}$. Note that different $(\mathbf x,\mathbf h)$ pairs may lead to the same line, for example, we have $\ell(\mathbf x,2\mathbf h)=\ell(\mathbf x,\mathbf h)=\ell(\mathbf x+\mathbf h,\mathbf h)$.

Let $f\in \mathcal F$ be an $m$-variate degree-$d$ polynomial tuple. For any $\mathbf x,\mathbf h \in \mathbb F^m$, the function $g:\mathbb F \to \mathbb F^t$ defined by $g(\lambda):=f(\mathbf x+\lambda \mathbf h)$ is a univariate degree-$d$ polynomial tuple \cite{RS96}. In other words, there exist coefficients tuples $\mathbf c_0,\ldots,\mathbf c_d \in \mathbb F^{t}$, such that for every $\lambda \in \mathbb F$, $g(\lambda)=f(\mathbf x+\lambda\mathbf h) = \sum_{i=0}^{d} \mathbf c_i \cdot \lambda^i$. 

For every $r \in \mathbb N$, we use $S^m_r$ to denote the set of non-negative integer vectors with sum equal to $r$, i.e., $S^m_r:=\{\mathbf v \in \mathbb N^m \mid \sum_{i=1}^m \mathbf v[i]=r\}$. Let $S^m_{\le r}=\bigcup_{0 \le i \le r} S^m_i$, we have $|S^m_r|=\binom{m+r-1}{r}$ and $|S^m_{\le r}|=\sum_{i=0}^r |S^m_i|=\binom{m+r}{r}$. 

We further use $\mathbf x^{\mathbf v}$ to denote the monomial $\prod_{i=1}^m \mathbf x[i]^{\mathbf v[i]}$, and let $f$ be an $m$-variate polynomial on $\mathbf x$, we use $\frac{\partial^r f}{\partial \mathbf x^{\mathbf v}}$ to denote the order-$r$ partial derivative $\frac{\partial^r f}{(\partial \mathbf x[1])^{\mathbf v[1]}\ldots(\partial \mathbf x[m])^{\mathbf v[m]}}$ for every $\mathbf v \in S^m_r$.

\subsection{PLTDC I: Derivative Code with Constant Degree}
\label{sec:pldtc1}
\begin{theorem}
For any integer $r>0$, let $d=2r+1$ and let $\mathbb F$ be a finite field of prime size $|\mathbb F| > 2 d$. For every integer $m>0$ satisfying $|\mathbb F|^m \ge \binom{m+d}{d}$, there is a PLTDC $C:\mathbb F^{\binom{m+d}{d}} \to \left(\mathbb F^{\binom{2m+r}{r}}\right)^{|\mathbb F|^{2m}}$ such that:
\begin{itemize}
    \item It is $\left(6d+7,|\mathbb F|^{18m},\delta,1-\frac{\delta}{2}\right)$-parallel locally testable for every $0<\delta<\frac{1}{(d+2)^2}$.
    \item It is $(2,|\mathbb F|^{2m+2},\delta,1-2\delta)$-smooth parallel locally decodable with respect to all $\chi_i$ and $\psi_{i,j}$, for every $0<\delta<\frac{1}{2}$.
\end{itemize}
Furthermore, $C$ can be constructed explicitly in time depending only on $|\mathbb F|,m,d$.
\end{theorem}
We first show the construction of $C$, then prove its local testability and local decodability, respectively.

\subsubsection{Construction}\label{sec:pldtc1-construct}
Since $|\mathbb F| >d$, the $m$-variate monomials $\{\mathbf x^{\mathbf v} \mid \mathbf v \in S^m_{\le d}\}$ are linearly independent functions. Together with the fact that $|\mathbb F|^m \ge \binom{m+d}{d}=|S^m_{\le d}|$, this implies the $(|\mathbb F|^m) \times \binom{m+d}{d}$ sample matrix $A$, where for every $\mathbf u \in \mathbb F^m, \mathbf v \in S^m_{\le d}$, $A[\mathbf u,\mathbf v]=\mathbf u^{\mathbf v}$, has full column rank. Thus there exists a set of points $\mathbf u_1,\ldots, \mathbf u_{\binom{m+d}{d}} \in \mathbb F^m$, such that the value on those points can uniquely determine an $m$-variate degree-$d$ polynomial. Such sets of points could be obtained in time $|\mathbb F|^{m\cdot \binom{m+d}{d}}$ by enumerating all such possible sets of points.

For a message $(a_1,\ldots,a_{\binom{m+d}{d}}) \in \mathbb F^{\binom{m+d}{d}}$, let $p$ be the unique $m$-variate degree-$d$ polynomial satisfying $p(\mathbf u_i)=a_i$ for every $1 \le i \le \binom{m+d}{d}$. We define a function $q:\mathbb F^{2m} \to \mathbb F$ such that $q(\mathbf x,\mathbf y)=p(\mathbf x)+p(\mathbf y)$ for every $\mathbf x,\mathbf y \in \mathbb F^m$. 
    
The code $C$ maps the message to a string of length $|\mathbb F|^{2m}$, where for every $\mathbf z \in \mathbb F^{2m}$, the $\mathbf z$-th symbol of the codeword is the collection of all order-$(\le r)$ partial derivatives of $q$ evaluated at $\mathbf z$. Specifically, denote the variables of $q$ as $\mathbf x$, then for every $\mathbf z \in \mathbb F^{2m}$,
\begin{equation}\label{eq:pldtc1}
    C\left(a_1,\ldots,a_{\binom{m+d}{d}}\right)[\mathbf z]
    =\left(\frac{\partial q^i}{\partial \mathbf x^{\mathbf v}}\ \middle|\  0 \le i \le r,\mathbf v \in S_i\right)[\mathbf z] \in \mathbb F^{|S^{2m}_{\le r}|}=  \mathbb F^{\binom{2m+r}{r}}.
\end{equation}

\subsubsection{Proof of Local Testability}
We first introduce the following parallel low-degree test, which is a key component in our local testing algorithm.

\begin{lemma}[\cite{DBLP:conf/icalp/LinRSW22}]\label{lem:pldt}
    Suppose $|\mathbb F|>2d$. There is an algorithm $\mathcal T$ which given access to a function $f:\mathbb F^m \to \mathbb F^t$, makes $(d+2)$ uniform queries, and outputs a bit such that
    \begin{itemize}
        \item (Completeness) If $f \in \mathcal F$, then $\Pr[\mathcal T(f)=1]=1$.
        \item (Soundness) If $\Delta(f,\mathcal F) \ge \delta$, then $\Pr[\mathcal T(f)=0]\ge \frac{\delta}{2}$ for every $0<\delta<\frac{1}{(d+2)^2}$.
    \end{itemize}
\end{lemma}

Now we prove the parallel local testability of the code $C$.

\begin{lemma}\label{lem:ltc1}
    The code $C$ defined in Section \ref{sec:pldtc1-construct} is $\left(6d+7,|\mathbb F|^{18m},\delta,1-\frac{\delta}{2}\right)$-parallel locally testable for every $0<\delta<\frac{1}{(d+2)^2}$.
\end{lemma}

\begin{proof}
Let $f:\mathbb F^{2m}\to \mathbb F^{t \times \binom{2m+r}{r}}$ be the function representation of a plausible codeword in the $t$-parallelization of $C$. For any $\sigma \in \mathbb F^{t \times \binom{2m+r}{r}}$, we treat $\sigma$ as a $t \times |S^{2m}_{\le r}|$ table and use $\sigma[i,\mathbf v]$ to denote its $(i,\mathbf v)$-th entry for every $i \in [t],\mathbf v \in S^{2m}_{\le r}$. 

The local testing algorithm is as follows.
\begin{enumerate}
    \item First, we check whether for every $i \in [t], \mathbf v \in S_{\le r}^{2m}$, $f(\cdot)[i,\mathbf v]$ is close to a $(2m)$-variate degree-$d$ polynomial $\widetilde f(\cdot)[i,\mathbf v]$ in the parallel sense. This can be done via the parallel low-degree test in Lemma \ref{lem:pldt}. If the rejection probability is smaller than $\frac{\delta}{2}$ for some $0<\delta<\frac{1}{(d+2)^2}$, then there exists a $(2m)$-variate polynomial tuple $\widetilde f(\cdot)$ with degree at most $d$ and a set $T\subseteq \mathbb F^{2m}$ with $|T|\ge (1-\delta)|\mathbb F|^{2m}$, such that for every $i\in [t], \mathbf v\in S^{2m}_{\le r}$ and $\mathbf x\in T$, $f(\mathbf x)[i,\mathbf v]=\widetilde f(\mathbf x)[i,\mathbf v]$. Since the distance between any two different degree-$d$ polynomials is at least $1-\frac{d}{|\mathbb F|}>2\delta$ by Schwartz-Zippel Lemma, the polynomial tuple $\widetilde f(\cdot)$ is unique.
    
    \item Next, we test whether $\forall i \in [t], j \in [r]$ and $\mathbf v \in S^{2m}_{j}$, $\widetilde f(\cdot)[i,\mathbf v]=\frac{\partial^j \widetilde f(\cdot)[i,\mathbf 0]}{\partial \mathbf x^{\mathbf v}}$. To do this, we pick $\mathbf x,\mathbf h \in \mathbb F^{2m}$ uniformly at random, and query $f$ at $(d+1)$ positions $\{\mathbf x+\lambda \mathbf h| 0 \le \lambda \le d\}$. Using the results, we can interpolate a univariate degree-$d$ polynomial tuple $g(\lambda):=f(\mathbf x+\lambda \mathbf h) \in \mathbb F^{t \times \binom{2m+r}{r}}$. The algorithm accepts if $\forall i \in [t]$,
    \begin{equation}\label{eq:pldtc1-ltc}
    \left\{
        \begin{aligned}
        g'(0)[i,\mathbf 0] & =\sum_{\mathbf v \in S_1} f(\mathbf x)[i,
        \mathbf v] \cdot \mathbf h^{\mathbf v} \\
        g''(0)[i,\mathbf 0] & = \sum_{\mathbf v \in S_2} f(\mathbf x)[i,\mathbf v]\cdot \mathbf h^{\mathbf v} \\
        & \ldots \\
        g^{(r)}(0)[i,\mathbf 0] & = \sum_{\mathbf v \in S_r} f(\mathbf x)[i,\mathbf v]\cdot \mathbf h^{\mathbf v}
        \end{aligned}
    \right..
    \end{equation}
    For completeness, $f=\widetilde f$ and $\widetilde f$ satisfies the property that for every $i \in [t],j \in [r]$ and $\mathbf v \in S^{2m}_{j}$, $\widetilde f(\cdot)[i,\mathbf v]=\frac{\partial^j \widetilde f(\cdot)[i,\mathbf 0]}{\partial \mathbf x^{\mathbf v}}$. Then both left-hand-side (LHS) and right-hand-side (RHS) of the $k$-th equation in (\ref{eq:pldtc1-ltc}) equal to the order-$k$ directional derivative of $\widetilde f$ with respect to direction $\mathbf h$ on point $\mathbf x$.

    For soundness, note that with probability $\ge 1-(d+1)\delta$, $f$ agrees with $\widetilde f$ on the $(d+1)$ queried points. This means the interpolated univariate polynomial tuple $g(\lambda)$ is $\widetilde f(\mathbf x+\lambda \mathbf h)$ and we have $f(\mathbf x)=\widetilde f(\mathbf x)$. Note that for each $i \in [t], \mathbf v \in S^{2m}_{\le r}$, $\widetilde f(\cdot)[i,\mathbf v]$ is a polynomial of degree at most $d$. Thus, if $\exists i \in [t],j \in[r]$ and $\mathbf v\in S^{2m}_j$ such that $\widetilde f(\cdot)[i,\mathbf v]\neq \frac{\partial^j \widetilde f(\cdot)[i,\mathbf 0]}{\partial \mathbf x^{\mathbf v}}$, then the LHS and RHS of (\ref{eq:pldtc1-ltc}) would be different polynomials on variables $(\mathbf x,\mathbf h)$ with degree at most $d+r \le \frac{3}{2}d$. By Schwartz-Zippel Lemma, at least $1-\frac{3d}{2|\mathbb F|}$ fraction of $(\mathbf x,\mathbf h)$ will make the algorithm reject. Therefore, the rejection probability is at least $1-\frac{3d}{2|\mathbb F|}-(d+1)\delta$ in this case.

    \item At the end, we test if for every $i \in [t]$, $\widetilde f(\mathbf x,\mathbf y)[i,\mathbf 0]$ is of the form $p(\mathbf x)+p(\mathbf y)$ for some $m$-variate degree-$d$ polynomial $p$. This is equivalent to testing whether 
    \begin{equation}\label{eq:pldtc1:form}
        \left\{
        \begin{aligned}
            \widetilde f(\mathbf x,\mathbf y)[i,\mathbf 0] & =\widetilde f(\mathbf x,\mathbf 0)[i,\mathbf 0]+\widetilde f(\mathbf 0,\mathbf y)[i,\mathbf 0] \\
            \widetilde f(\mathbf x,\mathbf 0)[i,\mathbf 0] & =\widetilde f(\mathbf 0,\mathbf x)[i,\mathbf 0]
        \end{aligned}
        \right.
    \end{equation}
    for every $\mathbf x,\mathbf y \in \mathbb F^m$, since $\widetilde f(\cdot,\cdot)[i,\mathbf 0]$ is itself a $(2m)$-variate degree-$d$ polynomial. For any fixed $\mathbf z \in \mathbb F^{2m}$, by picking $\mathbf h \in \mathbb F^{2m}$ uniformly at random and do local interpolation from $\{f(\mathbf z+\lambda \mathbf h) \mid 1 \le \lambda \le d+1\}$, we can get $\widetilde f(\mathbf z)$ with probability at least $1-(d+1)\delta$. We pick $\mathbf x,\mathbf y \in \mathbb F^m$ uniformly at random and test whether (\ref{eq:pldtc1:form}) holds using $4(d+1)$ queries. If for either of the equation in (\ref{eq:pldtc1:form}), LHS and RHS are different functions on $\mathbf x,\mathbf y$, then they differ on at least $1-\frac{d}{|\mathbb F|}$ fraction of points by Schwartz-Zippel Lemma. The rejection probability is therefore at least $1-\frac{d}{|\mathbb F|}-4(d+1)\delta$ in this case.    
\end{enumerate}

Overall, let $\delta=\Delta(f,\text{Im}(C^{\odot t}))$, the local testing algorithm generates $2+2+5=9$ uniformly random elements in $\mathbb F^{2m}$, queries $f$ at $(d+2) + (1+4)\cdot (d+1)=6d+7$ positions, and rejects with probability at least $\min\left(\frac{\delta}{2},1-\frac{3d}{2|\mathbb F|}-(d+1)\delta,1-\frac{d}{|\mathbb F|}-4(d+1)\delta\right)=\frac{\delta}{2}$ since $|\mathbb F|>2d$ and $0<\delta <\frac{1}{(d+2)^2}$.
\end{proof}

\subsubsection{Proof of Local Decodability}
\label{sec:pldtc1-ldc}

\begin{lemma}\label{lem:ldc1}
    The code $C$ defined in Section \ref{sec:pldtc1-construct} is $(2,|\mathbb F|^{2m+2},\delta,1-2\delta)$-smooth parallel locally decodable with respect to all $\chi_i$ and $\psi_{i,j}$, for every $0<\delta<\frac{1}{2}$.
\end{lemma}
\begin{proof}
    For any message $(a_1,\ldots,a_{\binom{m+d}{d}}) \in \mathbb F^{\binom{m+d}{d}}$, we have $a_i=p(\mathbf u_i)=q(\mathbf u_i,\mathbf 0)$ and $a_i+a_j=p(\mathbf u_i)+p(\mathbf u_j)=q(\mathbf u_i,\mathbf u_j)$. Thus, to prove the parallel local decodability of $C$, it suffices to show we can locally correct $q(\mathbf x)$ for any $\mathbf x \in \mathbb F^{2m}$ in parallel. 
    
    Let $f:\mathbb F^{2m} \to \mathbb F^{t \times \binom{2m+r}{r}}$ be the function representation of a string which is $\delta$-close to a codeword in $\text{Im}(C^{\odot t})$, and let $\widetilde f:\mathbb F^{2m} \to \mathbb F^{t \times \binom{2m+r}{r}}$ be the function representation of that codeword. We still treat every $\sigma \in \mathbb F^{t \times \binom{2m+r}{r}}$ as a $t \times |S^{2m}_{\le r}|$ table and use $\sigma[i,\mathbf v]$ to denote its $(i,\mathbf v)$-th entry for every $i \in [t],\mathbf v \in S^{2m}_{\le r}$. 
    
    Our goal is to locally correct the $t$ function values $(\widetilde f(\mathbf x)[1,\mathbf 0],\ldots,\widetilde f(\mathbf x)[t,\mathbf 0])$ for any given $\mathbf x \in \mathbb F^{2m}$. To do this, we pick $\mathbf h \in \mathbb F^{2m}$ and different $\lambda_1,\lambda_2 \in \mathbb F \setminus \{0\}$ uniformly at random, and query $f(\mathbf x+\lambda_1\mathbf h)$ and $f(\mathbf x+\lambda_2\mathbf h)$. Suppose $f$ agrees with $\widetilde f$ on both positions, which happens with probability at least $1-2\delta$. Then by the definition of the code, we correctly get the value of 
    \begin{equation}\label{eq:pldtc1-ldc}
        \left(\frac{\partial^j \widetilde f(\cdot)[i,\mathbf 0]}{\partial \mathbf x^{\mathbf v}}\mid i \in [t],0 \le j \le r,\mathbf v \in S^{2m}_j\right)
    \end{equation} 
    on $\mathbf x+\lambda_1\mathbf h_1$ and $\mathbf x+\lambda_2 \mathbf h_2$.
    
    Define a function $g(\lambda):=(\widetilde f(\mathbf x+\lambda \mathbf h)[1,\mathbf 0],\ldots,\widetilde f(\mathbf x+\lambda \mathbf h)[t,\mathbf 0])$. For each $i \in [t]$, $g(\cdot)[i]$ is a univariate degree-$d$ polynomial and there exists $c_{i,0},\ldots,c_{i,d} \in \mathbb F$ such that $g(\lambda)[i]=\sum_{j=0}^dc_{i,j}\lambda^j$. From (\ref{eq:pldtc1-ldc}), we can compute $g(\lambda_1),g'(\lambda_1),\ldots,g^{(r)}(\lambda_1)$ and $g(\lambda_2),g'(\lambda_2),\ldots,g^{(r)}(\lambda_2)$ by the chain rule, and thus for each $i \in [t]$, we can obtain $c_{i,0},\ldots,c_{i,d}$ by solving a linear equation system. The  non-singularity of the $(d+1) \times (d+1)$ coefficient matrix is guaranteed by Lemma \ref{lem:det-maintext}, for which we defer the proof to Appendix \ref{lem:det-appendix}.

    After getting the coefficients of $g$, we can obtain $(\widetilde f(\mathbf x)[1,\mathbf 0],\ldots,\widetilde f(\mathbf x)[t,\mathbf 0])$ by computing $g(0)$. As a result, $C$ is $(2,|\mathbb F|^{2m+2},\delta,1-2\delta)$-smooth parallel locally decodable for every $0<\delta<\frac{1}{2}$.
    
\end{proof}

\begin{lemma}
\label{lem:det-maintext}
    Let $r > 0$ be an integer and $d = 2r+1$, $\mathbb F$ be a finite field with prime size $|\mathbb F| > r$, and $\lambda_1,\lambda_2$ be two distinct elements in $\mathbb F$. The values $g(\lambda_1), g'(\lambda_1), \ldots, g^{(r)}(\lambda_1)$ and $g(\lambda_2), g'(\lambda_2), \ldots, g^{(r)}(\lambda_2)$ uniquely determine a univariate polynomial $g: \mathbb{F}\to \mathbb{F}$ with degree at most $d$.
\end{lemma}

\subsection{PLTDC II: Derivative Code Based on Line vs. Point Test}
\label{sec:pldtc2}
For any $m\in\mathbb{N}$, let $\mathcal L^m$ be the set of lines in $\mathbb F^m$. 
\begin{theorem}\label{thm:pldtc2}
For any integer $r>0$, let $d=2r+1$ and let $\mathbb F$ be a finite field of prime size $|\mathbb F|> \max(3d,9216)$. For every integer $m>0$ satisfying $|\mathbb F|^m \ge \binom{m+d}{d}$, there is a PLTDC $C:\mathbb F^{\binom{m+d}{d}} \to \left(\mathbb F^{(d+1)\cdot \binom{2m+r}{r}}\right)^{\mathcal L^{2m}}$ such that:
\begin{itemize}
    \item It is  $(11,|\mathbb F|^{34m+8},\delta,1-\frac{2}{13}\delta)$-parallel locally testable for every $0.035<\delta<0.06$.
    \item It is $(2,|\mathbb F|^{6m+2},\delta,1-2\delta)$-smooth parallel locally decodable with respect to all $\chi_i$ and $\psi_{i,j}$ for every $0<\delta<\frac{1}{2}$.
\end{itemize}
Furthermore, $C$ can be constructed explicitly in time depending only on $\mathbb F,m,d$.    
\end{theorem}

We remark that the block length is bounded by $|\mathcal L^{2m}| \le |\{(\mathbf x,\mathbf h) \in \mathbb F^{2m} \times \mathbb F^{2m}\}| \le |\mathbb F|^{4m}$.

\subsubsection{Construction}\label{sec:pldtc2-construct}
%Recall that $\mathcal L^m$ is the set of lines in $\mathbb F^m$. 
For a line\footnote{Here we abuse the notation a little by using $\ell$ to denote an unspecified line and using $\ell(\mathbf x,\mathbf h)$ to denote the line $\{\mathbf x+\lambda \mathbf h \mid \lambda \in \mathbb F\}$.} $\ell \in \mathcal L^m$, we fix the lexicographically smallest pair $(\mathbf x,\mathbf h) \in \mathbb F^m \times \mathbb F^m$ such that $\ell(\mathbf x,\mathbf h)=\ell$, and say a function $f:\ell \to \mathbb F^t$ is of degree $d$ if and only if the function $g:\mathbb F \to \mathbb F^t$ defined by $g(\lambda):=f(\mathbf x+\lambda \mathbf h)$ is a $t$-tuple of univariate degree-$d$ polynomials. Any degree-$d$ function $f:\ell \to \mathbb F^t$ can be encoded by an element in $\mathbb F^{t\times (d+1)}$ by writing down the $(d+1)$ coefficient tuples of $g$. 

 Recall the function $q:\mathbb F^{2m} \to \mathbb F$ defined in Section \ref{sec:pldtc1-construct}. Our previous code encodes the table of values of every $\frac{\partial^i q}{\partial \mathbf x^{\mathbf v}}$ where $0 \le i \le r,\mathbf v \in S^{2m}_i$. Here we generalize this construction by further including the restriction of those polynomials on every line in $\mathcal L^{2m}$. Specifically, for each $\ell \in \mathcal L^{2m}$,
 \[
    C(a_1,\ldots,a_M)[\ell]=\left(\left(\frac{\partial^i q}{\partial \mathbf x^{\mathbf v}}\right)_{|\ell} \ \middle|\ 0 \le i \le r, \mathbf v \in S^{2m}_i\right),
\]
where the restriction of each of the $|S^{2m}_{\le r}|=\binom{2m+r}{r}$ polynomials on $\ell$ is of degree at most $d$, and thus can be stored in $\mathbb F^{(d+1)}$ as mentioned  above. The alphabet of $C$ is therefore $\left(\mathbb F^{d+1}\right)^{\binom{2m+r}{r}}$.

\subsubsection{Proof of Local Testability}

We first introduce the following line vs. point test, which can tell whether a given function $f:\mathbb F^m \to \mathbb F^t$ is close to an $m$-variate degree-$d$ polynomial tuple, using only 2 queries.

\begin{lemma}[Line vs. Point Test]\label{lem:line-vs-point}
Suppose $|\mathbb F|> \max(3d,9216)$. Define an algorithm $\mathcal T$ as follows. 
\begin{enumerate}
    \item Take a function $f:\mathbb F^m \to \mathbb F^t$ and a proof $\pi: \mathcal L^m \to \mathbb F^{t \times (d+1)}$ as input.
    \item Pick random $\mathbf x,\mathbf h \in \mathbb F^m$, and decode a degree-$d$ function $P_{\ell(\mathbf x,\mathbf h)}:\ell(\mathbf x,\mathbf h) \to \mathbb F^t$ from the univariate degree-$d$ polynomial tuple defined by $\pi[\ell(\mathbf x,\mathbf h)] \in \mathbb F^{t \times (d+1)}$. Note that when $\mathbf h=\mathbf 0$, we decode $P_{\ell(\mathbf x,\mathbf 0)}(\mathbf x)$ by plugging 0 into the univariate polynomial tuple.
    \item Output whether $f(\mathbf x)=P_{\ell(\mathbf x,\mathbf h)}(\mathbf x)$.
\end{enumerate}

Then $\mathcal T$ has the following properties.
\begin{itemize}
    \item (Completeness) If $f \in \mathcal F$, then let $\pi[\ell]$ encodes $f_{|\ell}:\ell \to \mathbb F^m$ for every $\ell \in \mathcal L^m$, we have $\Pr[\mathcal T(f,\pi)=1]=1$. 
    \item (Soundness) If $\Delta(f,\mathcal F)\ge \delta$, then for any proof $\pi$, $\Pr[\mathcal T(f,\pi)=0] \ge \frac{\delta}{2}$ for every $0<\delta<\frac{1}{24}$.
\end{itemize}
\end{lemma}
We defer the proof of Lemma \ref{lem:line-vs-point} to Appendix \ref{app:line-vs-point}.

\begin{lemma}
    The code $C$ defined in Section \ref{sec:pldtc2-construct} is $(11,|\mathbb F|^{34m+8},\delta,1-\frac{2}{13}\delta)$-parallel locally testable for every $0.035<\delta<0.06$.
\end{lemma}
\begin{proof}
Let $\pi:\mathcal L^{2m} \to \mathbb F^{t \times \binom{2m+r}{r} \times (d+1)}$ be the function representation of a plausible codeword in $\text{Im}(C^{\odot t})$. For every $\ell \in \mathcal L^{2m}$, we use $P_{\ell}(\cdot):\ell \to \mathbb F^{t \times \binom{2m+r}{r}}$ to denote the degree-$d$ function decoded from $\pi(\ell)$. Furthermore, we define a function $f:\mathbb F^{2m} \to \mathbb F^{t \times \binom{2m+r}{r}}$ by setting $f(\mathbf x):=P_{\ell(\mathbf x,\mathbf 0)}(\mathbf x)$ for every $\mathbf x \in \mathbb F^{2m}$. For each $\sigma \in \mathbb F^{t \times \binom{2m+r}{r}}$, we treat $\sigma$ as a $t \times |S^{2m}_{\le r}|$ table and use $\sigma[i,\mathbf v]$ to denote its $(i,\mathbf v)$-th entry for every $i \in [t],\mathbf v \in S^{2m}_{\le r}$.

The local testing algorithm is as follows.
\begin{enumerate}
    \item First, we invoke the line vs. point test in Lemma \ref{lem:line-vs-point} by setting $m\gets 2m$, $t\gets t\times\binom{2m+r}{r}$ and feeding $(f,\pi)$ into the algorithm $\mathcal T$. If $\mathcal T$ rejects with probability at most $\varepsilon$ for some $0<\varepsilon<\frac{1}{48}$, then we have:
    \begin{itemize}
        \item $\Pr_{\mathbf x,\mathbf h \in \mathbb F^{2m}}\left[f(\mathbf x)\neq P_{\ell(\mathbf x,\mathbf h)}(\mathbf x)\right]\le  \varepsilon$.
        \item There exists a $(2m)$-variate degree-$d$ polynomial tuple $\widetilde f:\mathbb F^{2m} \to \mathbb F^{t \times \binom{2m+r}{r}}$ and a set $T \subseteq \mathbb F^{2m}$ with $|T| \ge (1-2\varepsilon) |\mathbb F|^{2m}$, such that $f$ agrees with $\widetilde f$ on $T$.
    \end{itemize}
    Thus, by a union bound we get
    \[\Pr_{\mathbf x,\mathbf h \in \mathbb F^{2m}}\left[\widetilde f(\mathbf x)\neq P_{\ell(\mathbf x,\mathbf h)}(\mathbf x)\right]\le 3\varepsilon.\]
    Note that the distribution of $\mathbf y$, where we uniformly pick $\mathbf x,\mathbf h \in \mathbb F^{2m}$ and $\mathbf y \in \ell(\mathbf x,\mathbf h)$, is also uniform over $\mathbb F^{2m}$. Thus,
    \begin{align*}
        & \Ex_{\mathbf x,\mathbf h \in \mathbb F^{2m}}\left[\Pr_{\mathbf y \in \ell(\mathbf x,\mathbf h)}\left[\widetilde f(\mathbf y) \neq P_{\ell(\mathbf x,\mathbf h)}(\mathbf y)\right]\right]\\
        = & \Ex_{\mathbf x,\mathbf h \in \mathbb F^{2m}}\left[\Pr_{\mathbf y \in \ell(\mathbf x,\mathbf h)}\left[\widetilde f(\mathbf y) \neq P_{\ell(\mathbf y,\mathbf h)}(\mathbf y)\right]\right]\\
        = & \Pr_{\mathbf y,\mathbf h \in \mathbb F^{2m}}\left[\widetilde f(\mathbf y) \neq P_{\ell(\mathbf y,\mathbf h)}(\mathbf y)\right] \le 3\varepsilon.
    \end{align*}
    Let
    {\footnotesize
    \[
    B=\left\{(\mathbf x,\mathbf h)\in \mathbb F^{2m}\times \mathbb F^{2m} \mid \Pr_{\mathbf y \in \ell(\mathbf x,\mathbf h)}\left[\widetilde f(\mathbf y) \neq P_{\ell(\mathbf x,\mathbf h)}(\mathbf y)\right]\ge \frac{2}{3}\right\}.
    \]}
    By Markov's inequality, we have \[\Pr_{\mathbf x,\mathbf h \in \mathbb F^{2m}}\left[(\mathbf x,\mathbf h) \in B\right] \le \frac{9}{2} \varepsilon.\]
    For any $(\mathbf x,\mathbf h) \notin B$, we have $\Delta(\widetilde f_{|\ell(\mathbf x,\mathbf h)},P_{\ell(\mathbf x,\mathbf h)}) <\frac{2}{3}$. Since the distance between the table of values of two different degree-$d$ polynomials is at least $1-\frac{d}{|\mathbb F|}>\frac{2}{3}$, we have $P_{\ell(\mathbf x,\mathbf h)}=\widetilde f_{|\ell(\mathbf x,\mathbf h)}$.

    Conditioned on $\mathbf h \neq \mathbf 0$, $\ell(\mathbf x,\mathbf h)$ is uniform over 
    $\mathcal L^{2m}_* := \mathcal L^{2m} \setminus \{\ell(\mathbf x,\mathbf 0) \mid \mathbf x \in \mathbb F^{2m}\}$. Therefore,
    \begin{align*}
        & 1 - \Delta(\pi, \mathrm{Im}(C^{\odot t}))\\
        \ge &\Pr_{\ell \in \mathcal L^{2m}}\left[P_{\ell}=\widetilde f_{|\ell}\right] \\
        \ge & \Pr_{\mathbf x,\mathbf h \in \mathbb F^{2m}}\left[P_{\ell(\mathbf x,\mathbf h)}=\widetilde f_{|\ell(\mathbf x,\mathbf h)}\mid \mathbf h \neq \mathbf 0\right]\cdot \frac{|\mathcal L^{2m}_*|}{|\mathcal L^{2m}|} \\
        \ge & \left(\Pr_{\mathbf x,\mathbf h \in \mathbb F^{2m}} \left[ P_{\ell(\mathbf x,\mathbf h)} = \widetilde f_{\ell(\mathbf x,\mathbf h)}\right]-\Pr_{\mathbf x,\mathbf h\in\mathbb F^{2m}}\left[\mathbf h=\mathbf 0\right]\right)\cdot \\
        &\left(1-\frac{|\mathbb F|^{2m}}{|\mathbb F|^{4m-2}-|\mathbb F|^{2m-2}+|\mathbb F|^{2m}}\right) \\
        \ge & \left(1-\frac{9}{2}\cdot \varepsilon - \frac{1}{|\mathbb F|^{2m}}\right)\left(1-\frac{1}{|\mathbb F|^{2m-2}}\right) \\
        \ge& 0.99-\frac{9}{2}\varepsilon\\
        >& 1-\frac{13}{2}\varepsilon,
    \end{align*}
    thus $\Delta(\pi, \mathrm{Im}(C^{\odot t}))\le \frac{13}2\varepsilon$. For the last inequality above, we further assumed $\varepsilon>0.005$.

    \item Second, we test whether for every $i \in [t],j \in [r]$ and $\mathbf v \in S^{2m}_j$, $\widetilde f(\cdot)[i,\mathbf v]=\frac{\partial^j \widetilde f(\cdot)[i,\mathbf 0]}{\partial \mathbf x^{\mathbf v}}$. To do this, we pick $\mathbf x,\mathbf h \in \mathbb F^{2m}$ uniformly at random, decode the univariate degree-$d$ polynomial $g(\lambda):=P_{\ell(\mathbf x,\mathbf h)}(\mathbf x+\lambda \mathbf h)\in \mathbb F^{t \times \binom{2m+r}{r}}$, and accepts if $\forall i \in [t]$,
    \begin{equation}\label{eq:pldtc2-ltc}
    \left\{
        \begin{aligned}
        g'(0)[i,\mathbf 0] & =\sum_{\mathbf v \in S_1} g(0)[i,
        \mathbf v] \cdot \mathbf h^{\mathbf v} \\
        g''(0)[i,\mathbf 0] & = \sum_{\mathbf v \in S_2} g(0)[i,\mathbf v]\cdot \mathbf h^{\mathbf v} \\
        & \ldots \\
        g^{(r)}(0)[i,\mathbf 0] & = \sum_{\mathbf v \in S_r} g(0)[i,\mathbf v]\cdot \mathbf h^{\mathbf v}
        \end{aligned}
    \right..
    \end{equation}
    For completeness, there exists a $(2m)$-variate degree-$d$ polynomial tuple $\widetilde f:\mathbb F^{2m} \to \mathbb F^{t \times \binom{2m+r}{r}}$ satisfying $\widetilde f(\cdot)[i,\mathbf v]=\frac{\partial^j \widetilde f(\cdot)[i,\mathbf 0]}{\partial \mathbf x^{\mathbf v}}$ for every $i \in [t],j \in [r],\mathbf v \in S^{2m}_j$, and we have $P_{\ell(\mathbf x,\mathbf h)}=\widetilde f_{|\ell(\mathbf x,\mathbf h)}$ for every $\mathbf x,\mathbf h \in \mathbb F^{2m}$. Therefore, both LHS and RHS of the $k$-th equation in \ref{eq:pldtc2-ltc} equal to the order-$k$ directional derivative of $\widetilde f$ with respect to direction $\mathbf h$ on point $\mathbf x$.

    For soundness, recall from step 1 that
    \[\Pr_{\mathbf x,\mathbf h \in \mathbb F^{2m}}\left[P_{\ell(\mathbf x,\mathbf h)}\neq \widetilde f_{|\ell(\mathbf x,\mathbf h)}\right] \le \frac{9}{2}\varepsilon,\]
    where $\widetilde f:\mathbb F^{2m} \to \mathbb F^{t \times \binom{2m+r}{r}}$ is a $(2m)$-variate degree-$d$ polynomial tuple.
    Thus, with probability at least $1-\frac{9}{2}\varepsilon$, we have $g(\lambda)=\widetilde f(\mathbf x+\lambda \mathbf h)$. If $\exists i \in [t],j \in [r]$ and $\mathbf v \in S^{2m}_j$ such that $\widetilde f(\cdot)[i,\mathbf v]\neq \frac{\partial^j \widetilde f(\cdot)[i,\mathbf 0]}{\partial \mathbf x^{\mathbf v}}$, then the LHS and RHS of (\ref{eq:pldtc2-ltc}) are different polynomials on variables $(\mathbf x,\mathbf h)$ with degree at most $d+r \le \frac{3}{2}d$. By Schwartz-Zippel Lemma, they differ on at least $1-\frac{3d}{2|\mathbb F|}$ fraction of points $(\mathbf x,\mathbf h)$. Therefore, the algorithm rejects with probability at least $1-\frac{3d}{2|\mathbb F|}-\frac{9}{2}\varepsilon$.

    \item At the end, we test if for every $i \in [t]$, $\widetilde f(\mathbf x,\mathbf y)$ is of the form $p(\mathbf x)+p(\mathbf y)$ for some $m$-variate degree-$d$ polynomial $p$. This is equivalent to testing if 
    \begin{equation}\label{eq:pldtc2:form}
        \left\{
        \begin{aligned}
            \widetilde f(\mathbf x,\mathbf y)[i,\mathbf 0] & =\widetilde f(\mathbf x,\mathbf 0)[i,\mathbf 0]+\widetilde f(\mathbf 0,\mathbf y)[i,\mathbf 0] \\
            \widetilde f(\mathbf x,\mathbf 0)[i,\mathbf 0] & =\widetilde f(\mathbf 0,\mathbf x)[i,\mathbf 0]
        \end{aligned}
        \right.
    \end{equation}
    holds for every $\mathbf x,\mathbf y \in \mathbb F^m$. Thus, we want to locally get the value of $(\widetilde f(\mathbf z)[1,\mathbf 0],\ldots,\widetilde f(\mathbf z)[t,\mathbf 0])$ for any $\mathbf z \in \mathbb F^{2m}$. To do this, we follow the 2-query local decoding algorithm in Section \ref{sec:pldtc1-ldc}. More specifically, we pick $\mathbf h \in \mathbb F^{2m}$ and different $\lambda_1,\lambda_2 \in \mathbb F \setminus \{0\}$ uniformly at random. Let $\mathbf z_1:=\mathbf z+\lambda_1 \mathbf h$ and $\mathbf z_2:=\mathbf z+\lambda_2 \mathbf h$. We know $\mathbf z_1,\mathbf z_2$ follow the uniform distribution over $\mathbb F^{2m}$. Thus, we further pick random $\mathbf h_1,\mathbf h_2 \in \mathbb F^{2m}$ and read out $P_{\ell(\mathbf z_1,\mathbf h_1)}(\mathbf z_1)$ and $P_{\ell(\mathbf z_2,\mathbf h_2)}(\mathbf z_2)$. By similar analysis above, with probability at least $1-9\varepsilon$, we have
    \[\left\{\begin{aligned}
        P_{\ell(\mathbf z_1,\mathbf h_1)}(\mathbf z_1) =\widetilde f(\mathbf z_1) = \widetilde f(\mathbf z+\lambda_1 \mathbf h) \\
        P_{\ell(\mathbf z_2,\mathbf h_2)}(\mathbf z_2) =\widetilde f(\mathbf z_2) = \widetilde f(\mathbf z+\lambda_2 \mathbf h)
    \end{aligned}\right..\]
    Therefore, we can solve the coefficients of $g(\lambda):=(\widetilde f(\mathbf z+\lambda \mathbf h)[1,\mathbf 0],\ldots,\widetilde f(\mathbf z+\lambda \mathbf h)[t,\mathbf 0])$ by Gaussian elimination, and obtain $(\widetilde f(\mathbf z)[1,\mathbf 0],\ldots,\widetilde f(\mathbf z)[t,\mathbf 0])=g(0)$ as described in Section \ref{sec:pldtc1-ldc}.

    If LHS and RHS are two different functions on $\mathbf x,\mathbf y$ in either of the two equations in (\ref{eq:pldtc2:form}), then they differ on at least $1-\frac{d}{|\mathbb F|}$ fraction of $(\mathbf x,\mathbf y)$ by Schwartz-Zippel Lemma. Thus, the rejection probability in this step is as least $1-\frac{d}{|\mathbb F|}-4\cdot 9\varepsilon$. 
\end{enumerate}
Overall, let $\varepsilon$ be the rejection probability of the line vs. point test. Assume $0.005<\varepsilon<0.01$. Our whole parallel local testing algorithm generates at most $2+2+13=17$ random elements in $\mathbb F^{2m}$ and $8$ elements in $\mathbb F \setminus \{0\}$, and makes $2+1+8=11$ queries, such that if the entire test rejects with probability at most $\min(\varepsilon,1-\frac{3d}{2|\mathbb F|}-\frac{9}{2}\varepsilon, 1-\frac{d}{|\mathbb F|}-4\cdot 9\varepsilon)=\varepsilon$, then the given string is at most $(\frac{13}{2}\varepsilon)$-far from a codeword in $C^{\odot t}$. In other words, if the rejection probability is at most $\frac{2}{13}\delta$ for some $0.035<\delta < 0.06$, then the given string is at most $\delta$-far from a codeword. Therefore, $C$ is $(11,|\mathbb F|^{34m+8},\delta,1-\frac{2}{13}\delta)$-parallel locally testable for any $0.035<\delta<0.06$.

\end{proof}

\subsubsection{Proof of Local Decodability}

\begin{lemma}
    The code $C$ defined in Section \ref{sec:pldtc2-construct} is $(2,|\mathbb F|^{6m+2},\delta,1-2\delta)$-smooth parallel locally decodable with respect to all $\chi_i$ and $\psi_{i,j}$ for every $0<\delta<\frac{1}{2}$.
\end{lemma}

\begin{proof}
    Let $\pi:\mathcal L^{2m} \to \mathbb F^{t \times \binom{2m+r}{r} \times (d+1)}$ be the function representation of a string which is $\delta$-close to $\text{Im}(C^{\odot t})$. For every $\ell \in \mathcal L^{2m}$, we use $P_{\ell}(\cdot):\ell \to \mathbb F^{t \times \binom{2m+r}{r}}$ to denote the degree-$d$ function decoded from $\pi(\ell)$. Since $\Delta(\pi,\text{Im}({C^{\odot t}})) \le \delta$, we have
    \[\Pr_{\ell \in \mathcal L^{2m}}\left[P_{\ell}=\widetilde f_{|\ell}\right]\ge 1-\delta.\]

    To decode $t$-parallelized messages, it suffices to locally correct $(\widetilde f(\mathbf z)[1,\mathbf 0],\ldots,\widetilde f(\mathbf z)[t,\mathbf 0]) \in \mathbb F^t$ for $\mathbf z=(\mathbf u_i,\mathbf 0)$ or $\mathbf z=(\mathbf u_i,\mathbf u_j)$. We follow a similar procedure and analysis in step 3 of our local testing algorithm, where we pick random $\mathbf h \in \mathbb F^{2m}$ and $\lambda_1,\lambda_2 \in \mathbb F \setminus\{0\}$, and define $\mathbf z_1:=\mathbf z+\lambda_1 \mathbf h, \mathbf z_2:=\mathbf z+\lambda_2 \mathbf h$. We first read out $P_{\ell}(\mathbf z_1)$, where $\ell$ is a random line in $\mathcal L^{2m}$ passing through $\mathbf z_1$, to get the value $\widetilde f(\mathbf z_1)$.
    We apply this to $\mathbf z_2$, too. Note that $\mathbf z_1,\mathbf z_2$ are 2 uniform points in $\mathbb F^{2m}$, thus the 2 queried lines are also uniform in $\mathcal L^{2m}$. The remaining proof is the same as that in Section \ref{sec:pldtc1-ldc}, and we omit the details. In short, we can recover $(\widetilde f(\mathbf z)[1,\mathbf 0],\ldots,\widetilde f(\mathbf z)[t,\mathbf 0]) \in \mathbb F^t$ with probability at least $1-2\delta$, by randomly drawing 1 element in $\mathbb F^{2m}$, 2 elements in $\mathbb F \setminus \{0\}$,  2 elements in $\{\ell \in \mathcal L^{2m} \mid \ell \text{ passes through a given point}\}$, and making 2 uniform queries. Therefore, $C$ is $(2,|\mathbb F|^{6m+2},\delta,1-2\delta)$-smooth parallel locally decodable with respect to all $\chi_i$ and $\psi_{i,j}$ for every $0<\delta<\frac{1}{2}$.
\end{proof}
\section{Applications}
\subsection{Stronger ETH Lower Bounds for Constant-Gap \kclique{}}
\label{sec:stronger}

In this section, we apply Theorem \ref{thm:vcsp2gapclique} with the PLDTC constructed in Section \ref{sec:pldtc2} to prove a stronger lower bound of constant-gap \kclique{} based on ETH (Theorem \ref{thm:eth-gap-clique-lowerbound}). By further applying a disperser to amplify the gap, we can rule out any FPT algorithm that approximates \kclique{} within $k/(\log k)^{\lambda\log\log\log k}$ ratio under ETH (Theorem \ref{thm:eth-gap-clique-fpt}). 

\begin{theorem}\label{thm:eth-gap-clique-lowerbound}
Assuming ETH, for some constant $\lambda>0$, constant-gap \kclique{} admits no algorithm with running time $f(k)n^{k^{\lambda/\log \log k}}$ for any computable function $f$.
\end{theorem}

\begin{theorem}\label{thm:eth-gap-clique-fpt}
Assuming ETH, for some constant $\lambda>0$, there is no FPT algorithm which can solve the $e(k)$-gap \kclique{} problem, where
\[e(k):=k/(\log k)^{\lambda\log\log\log k}.\]
\end{theorem}

Below we first prove Theorem \ref{thm:eth-gap-clique-lowerbound}.

\begin{proof}[Proof of Theorem~\ref{thm:eth-gap-clique-lowerbound}]
Fix a parameter $\alpha=2^{100}$. Pick any odd number $d>3072$ and set $m=\alpha d,k=\binom{m+d}{d}$. Since $\left(\frac{m}{d}\right)^d\le \binom{m+d}{d} \le \left(\frac{e(m+d)}{d}\right)^d$, we have $\frac{\log k}{\log e(1+\alpha)}\le d \le \frac{\log k}{\log \alpha}$ and $m \le d k^{\frac{1}{d}}$. Let $\mathbb F$ be the smallest field with prime size such that $|\mathbb F| > 3d$. By Bertrand's postulate, we have $|\mathbb F|\le 6d$.

By Corollary \ref{cor:eth}, assuming ETH, no algorithm can decide the satisfiability of a vector 2CSP instance with $k$ variables and dimension $\frac{5\log n}{\log |\mathbb F|}$ in $f(k)\cdot n^{o(k)}$ time for any computable function $f$. Below we reduce such an instance $\Pi=(X,\Sigma,\Phi)$ with size $n$ to constant-gap {$K$-\textsc{Clique}} via Theorem \ref{thm:vcsp2gapclique}.

Recall from Theorem \ref{thm:pldtc2} that the code $C$ encodes a message in $\mathbb F^{\binom{m+d}{d}}$ to $\left(\mathbb F^{(d+1)\cdot \binom{2m+r}{r}}\right)^{\mathcal L^{2m}}$ where $r=\frac{d-1}{2}$, and $C$ can be constructed in time depending only on $\mathbb F,m,d$. The alphabet size $|\Sigma_2|$ of a codeword is at most
\begin{align*}
|\mathbb F|^{(d+1)\cdot \binom{2m+r}{r}} &\le |\mathbb F|^{(d+1)\cdot\left(\frac{e\left(2m+\frac{d-1}{2}\right)}{\frac{d-1}{2}}\right)^{\frac{d-1}{2}}} \\
& \le  |\mathbb F|^{\left(6ek^{\frac{1}{d}}\right)^{\frac{d}{2}}} = |\mathbb F|^{k^{\frac{1}{2}+\varepsilon}}
\end{align*}
for $\varepsilon =\frac{\log (6e)}{200}<0.04$. 
By plugging the following parameters into Theorem \ref{thm:vcsp2gapclique}:
\begin{itemize}
    \item $|\Sigma_1|=|\mathbb F|$,
    \item $|\Sigma_2| \le |\mathbb F|^{k^{\frac{1}{2}+\varepsilon}}$,
    \item $t=\frac{5 \log n}{\log |\mathbb F|}$,
    \item $q=11$,
    \item $R_T = |\mathbb F|^{34m+8}$,
    \item $R_D = |\mathbb F|^{6m+2}$,
    \item $k'=|\mathcal L^{2m}| \le |\mathbb F|^{4m} \le R_T,R_D \le (k')^{\Theta(1)}$,
    \item $\delta=0.04$,
    \item $\varepsilon_T=1-\frac{2}{13}\delta < 0.994$,
    \item $\varepsilon_D=1-2\cdot 5\delta=0.6$,
\end{itemize}
we get a {$K$-\textsc{Clique}} instance $G=(V,E)$ with $K=R_T$, $N:=|V|=K \cdot n^{O(k^{\frac{1}{2}+\varepsilon})}$ in $N^{O(1)}$ time such that:
\begin{itemize}
    \item If $\Pi$ is satisfiable, then $G$ contains a size-$K$ clique.
    \item If $\Pi$ is not satisfiable, then $G$ contains no clique of size $\varepsilon_T K$.
\end{itemize}
Note that 
\begin{align*}
    & K=|\mathbb F|^{34m+8} = \Theta(\log k)^{\Theta(\log k)}=2^{\Theta(\log k\log \log k)},
\end{align*} 
and thus
\begin{align*}
    & k \ge K^{\Theta\left(\frac{1}{\log \log k}\right)} \ge K^{\Theta\left(\frac{1}{\log \log K}\right)}.
\end{align*} 

Therefore, for any computable function $g$, we set $f=g(K(k))\cdot K(k)^k$, then the $f(k)n^{\Omega(k)}$ lower bound for vector 2CSP implies the $g(K)N^{\Omega(k^{\frac{1}{2}-\varepsilon})}=g(K)N^{K^{\Omega\left(\frac{1}{\log \log K}\right)}}$ lower bound for constant-gap {$K$-\textsc{Clique}}.
\end{proof}

To prove Theorem~\ref{thm:eth-gap-clique-fpt}, we need to introduce the disperser, a combinatorial object used to amplify the gap of \kclique{}.

\begin{definition}[Disperser \cite{CW89, Zuc96a, Zuc96b}]
	    For positive integers $k,m,\ell, r\in \mathbb N$ and constant $\varepsilon\in(0,1)$, an $(k,m,\ell,r,\varepsilon)$-disperser is a collection $\mathcal I$ of $m$ subsets $I_1,\ldots,I_m \subseteq [k]$, each of size $\ell$, such that the union of any $r$ different subsets from the collection has size at least $(1-\varepsilon)k$.
    \end{definition}

    The following lemma shows the existence of dispersers, and how to construct them probabilistically.
 
    \begin{lemma}[\cite{CCK+17, DBLP:conf/icalp/LinRSW22}]
    \label{lem:disperser}
	For every $k,r \in \mathbb N$ and $\varepsilon \in (0,1)$, let $\ell=\lceil \frac{3k}{\varepsilon r}\rceil$ and let $I_1,\ldots,I_m$ be random $\ell$-subsets of $[k]$. If $\ln m \le \frac{k}{r}$ then $\mathcal I=\{I_1,\ldots,I_m\}$ is a $(k,m,\ell,r,\varepsilon)$-disperser with probability at least $1-e^{-k}$.
    \end{lemma}
    
    Next, we prove the following claim, which applies the disperser to amplify the gap in \kclique{} at the cost of increasing the instance size. This approach has been used in \cite{CCK+17,DBLP:conf/icalp/LinRSW22}.

\begin{claim}\label{claim:A}
For any constant $\lambda>0$ and $\varepsilon > 0$, given a $\left(\frac{1}{\varepsilon}\right)$-gap \kclique{} instance $G = (V_1\dot\cup \ldots \dot\cup V_k,E)$ with $n$ vertices, we can construct an $e(m)$-gap {$m$-\textsc{Clique}} instance $G' = (V_1'\dot\cup \ldots \dot\cup V_m',E')$ with $N\le n^{k^{\lambda / \log \log k}}$ vertices in $N^{O(1)}=n^{O\left(k^{\lambda/\log \log k}\right)}$ time, where \[e(m)\ge \frac{m}{(\log m)^{\frac{2}{\lambda}\cdot \log \log \log m}}\] for sufficiently large $m$.
\end{claim}
\begin{proof}
We construct a disperser with the following parameters:
\begin{itemize}
    \item $r=\frac{3}{1-\varepsilon}\cdot k^{1-\frac{\lambda}{\log \log k}}$,
    \item $\ell= k^{\frac{\lambda}{\log \log k}}$,
    \item $m=\exp\left(\frac{1-\varepsilon}{3}\cdot r^{\frac{\lambda}{\log \log r}}\right)$.
\end{itemize}
Note that $\ell = \frac{3k}{(1-\varepsilon)r}$ and $\ln m = \frac{1-\varepsilon}{3}\cdot r^{\frac{\lambda}{\log \log r}}\le \frac{1-\varepsilon}{3}\cdot k^{\frac{\lambda}{\log \log k}}= \frac{k}{r}$. Hence, our choice of parameters fulfills the constraints in Lemma~\ref{lem:disperser}, and the construction can be done by enumerating all possible set systems in time only depending on $k$. Thus, we obtain a collection of sets $\mathcal I = \{I_1,\dots,I_m\}$ over universe $[k]$, such that each $I_i$ has size $\ell$, and every size-$r$ subcollection of $\mathcal I$ has a union with size at least $\varepsilon k$.

Next, for every set $I_i = \{j_1,j_2,\ldots,j_{\ell}\}$ 
in $\mathcal I$, we construct $V_i'$ as all tuples of vertices $(u_{j_1},\ldots,u_{j_{\ell}})$, where for each $k \in [\ell]$, $u_{j_k}\in V_{j_k}$. We add an edge between two vertices $(u_{k_1},\dots, u_{k_{\ell}})\in V_{i_1}'$ and $(w_{j_1},\ldots, w_{j_{\ell}})\in V_{i_2}'$ if and only if the set $\{u_{k_1},\dots, u_{k_{\ell}},w_{j_1},\ldots, w_{j_{\ell}}\}$ is a clique in $G$. It is easy to see that $G'$ contains at most $N \le n^{\ell}=n^{k^{\lambda / \log \log k}}$ vertices, and the reduction has running time $N^{O(1)}$. 

If $G$ has a clique $D$ of size $k$, then for every $j \in [k]$, denote the vertex in $D \cap V_j$ as $v_j$. For each set $I_i = \{j_1,j_2,\dots,j_{\ell}\}$, we choose $(v_{j_1},\ldots,v_{j_{\ell}}) \in V_i'$. It is straightforward that these vertices form a clique in $G'$.

If $G$ has no clique with size $\varepsilon k$, then we prove that $G'$ has no clique with size $r$. 
Suppose by contradiction $G'$ has a clique $D$ with size at least $r$, then we consider the following set of vertices $T$:
\[T:=\{w\in G\mid \exists (u_{j_1},\dots,u_{j_{\ell}})\in D, w \in \{u_{j_1},\ldots,u_{j_{\ell}}\}\}.\]
By our construction, $T$ is a clique in $G$. Since the union of every $r$ distinct sets in $\mathcal I$ has size at least $\varepsilon k$, we have that $|T|\ge \varepsilon k$, contradicting the fact that $G$ has no clique with size $\ge \varepsilon k$. 

Therefore, $G'$ is an $\left(\frac{m}{r}\right)$-gap {$m$-\textsc{Clique}} instance. 

Note that $\frac{x}{\log x}\le y \implies  x \le 1.01 y\log y$ for sufficiently large $y$. Since 
\begin{align*}
    r^{\frac{\lambda}{\log \log r}} & =\frac{3}{1-\varepsilon}\cdot \ln m, \\
    \frac{\lambda\log r}{\log \log r} & = C+\log \log m
\end{align*}
where $C=\log \left(\frac{3}{(1-\varepsilon)\log e}\right)$ is a constant, we have
\begin{align*}
\log r & \le \frac{1.01}{\lambda} \left(C+\log \log m\right) \log (C+\log 
    \log m)\\
    &\le \frac{2}{\lambda} \log \log m \log \log \log m
\end{align*}
and thus
\[
    r \le (\log m)^{\frac{2}{\lambda}\cdot \log \log \log m}
\]
for sufficiently large $m$.
\end{proof}

\begin{proof}[Proof of Theorem~\ref{thm:eth-gap-clique-fpt}]

By Theorem~\ref{thm:eth-gap-clique-lowerbound}, assuming ETH, there exists constants $\lambda_0>0$ and $c_0>1$, such that $c_0$-gap \kclique{} cannot be solved in $f(k)n^{k^{\lambda_0/\log \log k}}$ time for any computable function $f$. 

Suppose for contradiction that for every $\lambda>0$, there exists an algorithm which can solve the $(k/(\log k)^{\frac{2}{\lambda}\cdot \log \log \log k})$-gap {$k$-\textsc{Clique}} with $N$ vertices in FPT time, then we apply the reduction in Claim \ref{claim:A} with $\lambda < \lambda_0$ and $\varepsilon=\frac{1}{c_0}$ to obtain an algorithm, which can solve $c_0$-gap \kclique{} with running time $f(k)n^{O(k^{\lambda/\log \log k})}<f(k)n^{k^{\lambda_0/\log \log k}}$ for some computable function $f$, contradicting Theorem \ref{thm:eth-gap-clique-lowerbound}.
\end{proof}

\subsection{Review of Previous Proofs}
\label{sec:review}

\noindent\textbf{Constant-gap \kclique.} Let $\mathbb{F}_2$ be the binary field.
Recall that the Hadamard Code over $\mathbb{F}_2$ is $\left(3,2^{2k}, \delta,1-\frac{\delta}{2}\right)$-parallel locally testable and $(2,2^k,\delta,1-2\delta)$-parallel smooth locally decodable   for any $0<\delta<\frac{1}{2}$. Using Theorem~\ref{thm:clique2vcsp} and Theorem~\ref{thm:vcsp2gapclique} with this code, we obtain a reduction which takes as input an integer $k$ and a \kclique{} instance $G$, outputs a {$k'$-\textsc{Clique}} instance $G'$ with $k'=2^{\Theta(k)}$ in $2^{\Theta(k)}\cdot |G|^{O(1)}$-time such that
\begin{itemize}
    \item If $G$ contains a size-$k$ clique, then $G'$ contains a size-$k'$ clique.
    \item If $G$ contains no size-$k$ clique, then $G'$ contains no clique of size $(1-\varepsilon) k'$ for some constant $0<\varepsilon<1$.
\end{itemize}
This reduction improves the parameter $k'=2^{k^6}$ in the gap-reduction of~\cite{Lin21} to $k'=2^{\Theta(k)}$. The reduction in~\cite{DBLP:conf/icalp/LinRSW22}  also achieved  this parameter by using the Reed-Muller Code. Our technique generalizes both results and simplifies the proof.

\medskip

\noindent\textbf{$k^{o(1)}$-Gap \kclique{}.} In~\cite{DBLP:conf/coco/SK22}, the authors proved the $k^{1/H(k)}$-inapproximability for \kclique{} for any increasing computable function $H$ assuming $\mathsf{W[1]}\neq \mathsf{FPT}$.  Theorem~\ref{thm:vcsp2gapclique} does not generalize this reduction because the reduction in~\cite{DBLP:conf/coco/SK22} relies on a property of Hadamard Code that is stronger than the smooth property in Definition~\ref{def:parallel-decode}. More precisely, let $\mathbb{F}$ be a finite field with $|\mathbb{F}|> k$. Suppose $f:\mathbb{F}^k\to\mathbb{F}^{\mathbb{F}^k}$ is a function that is $(1/k^{1/H(k)})$-close to some Hadamard code $C: \mathbb{F}^k\to\mathbb{F}^{\mathbb{F}^k}$ for some increasing computable function $H$. For any ${\alpha}\in\mathbb{F}^k$, we can still use $f$ to decode $C({{\alpha}})$ by sampling ${\beta}\in\mathbb{F}^k$  and querying the values of $f$ at any two points in the set $S_{{\beta}}:=\{{\beta}+t{\alpha}\mid t\in\mathbb{F}\}$. This is because if $f$ is $(1/k^{1/H(k)})$-close to $C$, then $|\{\mathbf x\in\mathbb{F}^k\mid f(\mathbf x)=C(\mathbf x)\}|>|\mathbb{F}|^{k-1}$. Observe that there are $|\mathbb{F}|^{k-1}$ distinct sets in the set family $\mathcal F:=\{S_{\beta}\mid \beta\in\mathbb{F}^k\}$ and any two of them are disjoint, by the pigeonhole principle, there must exist two good points in the same $S_{\beta}$. 

On the other hand, the smooth property in Definition~\ref{def:parallel-decode} only guarantees that to decode $C(\alpha)$, every $\mathbf{z}\in\mathbb{F}^k$ is equally likely to be queried. For example, to decode $C(\alpha)$, a decoder algorithm that samples $\beta\in\mathbb{F}^k,t\in\mathbb{F}$ and queries $f(\beta+t\alpha)$ and $f(\beta+(t+1)\alpha)$ is still smooth. But the existence of two points in the same $S_{\beta}$ is not sufficient to decode $C(\alpha)$. Nevertheless, our Theorem~\ref{thm:clique2vcsp} might be used to simplify the proof in \cite{DBLP:conf/coco/SK22} by avoiding the analysis of multiplying random matrices. 

\bibliographystyle{alpha}
\bibliography{main}

\newcommand{\etalchar}[1]{$^{#1}$}
\begin{thebibliography}{ALM{\etalchar{+}}98}

\bibitem[ALM{\etalchar{+}}98]{ASLMSS98}
Sanjeev Arora, Carsten Lund, Rajeev Motwani, Madhu Sudan, and Mario Szegedy.
\newblock Proof verification and the hardness of approximation problems.
\newblock {\em J. ACM}, 45(3):501–555, may 1998.

\bibitem[BGLR93]{bellare1993efficient}
Mihir Bellare, Shafi Goldwasser, Carsten Lund, and Alexander Russell.
\newblock Efficient probabilistically checkable proofs and applications to
  approximations.
\newblock In {\em Proceedings of the twenty-fifth annual ACM symposium on
  Theory of computing}, pages 294--304, 1993.

\bibitem[BLR93]{blum1993self}
Manuel Blum, Michael Luby, and Ronitt Rubinfeld.
\newblock Self-testing/correcting with applications to numerical problems.
\newblock {\em Journal of computer and system sciences}, 47(3):549--595, 1993.

\bibitem[BS94]{bellare1994improved}
Mihir Bellare and Madhu Sudan.
\newblock Improved non-approximability results.
\newblock In {\em Proceedings of the twenty-sixth annual ACM symposium on
  Theory of computing}, pages 184--193, 1994.

\bibitem[BS95]{bernstein1995some}
Mira Bernstein and Neil~JA Sloane.
\newblock Some canonical sequences of integers.
\newblock {\em Linear Algebra and its Applications}, 226:57--72, 1995.

\bibitem[CCK{\etalchar{+}}17]{CCK+17}
Parinya Chalermsook, Marek Cygan, Guy Kortsarz, Bundit Laekhanukit, Pasin
  Manurangsi, Danupon Nanongkai, and Luca Trevisan.
\newblock From gap-eth to fpt-inapproximability: Clique, dominating set, and
  more.
\newblock In Chris Umans, editor, {\em 58th {IEEE} Annual Symposium on
  Foundations of Computer Science, {FOCS} 2017, Berkeley, CA, USA, October
  15-17, 2017}, pages 743--754. {IEEE} Computer Society, 2017.

\bibitem[CFLL23]{CFLL23}
Yijia Chen, Yi~Feng, Bundit Laekhanukit, and Yanlin Liu.
\newblock Simple combinatorial construction of the k\({}^{\mbox{o(1)}}\)-lower
  bound for approximating the parameterized k-clique.
\newblock {\em CoRR}, abs/2304.07516, 2023.

\bibitem[CHKX06]{ChenHKX06}
Jianer Chen, Xiuzhen Huang, Iyad~A. Kanj, and Ge~Xia.
\newblock Strong computational lower bounds via parameterized complexity.
\newblock {\em J. Comput. Syst. Sci.}, 72(8):1346--1367, 2006.

\bibitem[CW89]{CW89}
Aviad Cohen and Avi Wigderson.
\newblock Dispersers, deterministic amplification, and weak random sources.
\newblock In {\em 30th Annual Symposium on Foundations of Computer Science},
  pages 14--19. IEEE Computer Society, 1989.

\bibitem[DF95]{downey1995fixed}
Rodney~G Downey and Michael~R Fellows.
\newblock Fixed-parameter tractability and completeness {II}: On completeness
  for {W}[1].
\newblock {\em Theoretical Computer Science}, 141(1-2):109--131, 1995.

\bibitem[Din16]{Din16}
Irit Dinur.
\newblock Mildly exponential reduction from gap 3sat to polynomial-gap
  label-cover.
\newblock {\em Electron. Colloquium Comput. Complex.}, page 128, 2016.

\bibitem[FGL{\etalchar{+}}96]{feige1996interactive}
Uriel Feige, Shafi Goldwasser, Laszlo Lov{\'a}sz, Shmuel Safra, and Mario
  Szegedy.
\newblock Interactive proofs and the hardness of approximating cliques.
\newblock {\em Journal of the ACM (JACM)}, 43(2):268--292, 1996.

\bibitem[FK00]{feige2000two}
Uriel Feige and Joe Kilian.
\newblock Two-prover protocols---low error at affordable rates.
\newblock {\em SIAM Journal on Computing}, 30(1):324--346, 2000.

\bibitem[FS95]{fstest}
K.~Friedl and M.~Sudan.
\newblock Some improvements to total degree tests.
\newblock In {\em Proceedings Third Israel Symposium on the Theory of Computing
  and Systems}, pages 190--198, 1995.

\bibitem[Gol98]{goldwasser1998introduction}
Shafi Goldwasser.
\newblock Introduction to special section on probabilistic proof systems.
\newblock {\em SIAM Journal on Computing}, 27(3):737, 1998.

\bibitem[GS06]{OM06}
Oded Goldreich and Madhu Sudan.
\newblock Locally testable codes and pcps of almost-linear length.
\newblock {\em J. ACM}, 53(4):558–655, jul 2006.

\bibitem[Has96]{hastad1996clique}
Johan Hastad.
\newblock Clique is hard to approximate within $n^{1-\epsilon}$.
\newblock In {\em Proceedings of 37th Conference on Foundations of Computer
  Science}, pages 627--636. IEEE, 1996.

\bibitem[Has01]{hastad3querypcp}
Johan Hastad.
\newblock Some optimal inapproximability results.
\newblock {\em J. ACM}, 48(4):798–859, jul 2001.

\bibitem[IP01]{IP01}
Russell Impagliazzo and Ramamohan Paturi.
\newblock On the complexity of k-sat.
\newblock {\em Journal of Computer and System Sciences}, 62:367--375, 2001.

\bibitem[Kar72]{kar72}
R.~M. Karp.
\newblock Reducibility among combinatorial problems.
\newblock In {\em Proceedings of a symposium on the Complexity of Computer
  Computations, held March 20-22, 1972, at the {IBM} Thomas J. Watson Research
  Center, Yorktown Heights, New York.}, pages 85--103, 1972.

\bibitem[KK22]{DBLP:conf/coco/SK22}
{Karthik {C. S.}} and Subhash Khot.
\newblock Almost polynomial factor inapproximability for parameterized
  k-clique.
\newblock In Shachar Lovett, editor, {\em 37th Computational Complexity
  Conference, {CCC} 2022, July 20-23, 2022, Philadelphia, PA, {USA}}, volume
  234 of {\em LIPIcs}, pages 6:1--6:21. Schloss Dagstuhl - Leibniz-Zentrum
  f{\"{u}}r Informatik, 2022.

\bibitem[Koz06]{kozen2006theory}
Dexter~C Kozen.
\newblock {\em Theory of computation}, volume 170.
\newblock Springer, 2006.

\bibitem[KST54]{KST54}
Tam{\'a}s K{\~o}v{\'a}ri, Vera~T. S{\'o}s, and Paul Tur{\'a}n.
\newblock On a problem of k. zarankiewicz.
\newblock {\em Colloquium Mathematicum}, 3:50--57, 1954.

\bibitem[Lin21]{Lin21}
Bingkai Lin.
\newblock Constant approximating k-clique is w[1]-hard.
\newblock In Samir Khuller and Virginia~Vassilevska Williams, editors, {\em
  {STOC} '21: 53rd Annual {ACM} {SIGACT} Symposium on Theory of Computing,
  Virtual Event, Italy, June 21-25, 2021}, pages 1749--1756. {ACM}, 2021.

\bibitem[LRSW22]{DBLP:conf/icalp/LinRSW22}
Bingkai Lin, Xuandi Ren, Yican Sun, and Xiuhan Wang.
\newblock On lower bounds of approximating parameterized k-clique.
\newblock In Mikolaj Bojanczyk, Emanuela Merelli, and David~P. Woodruff,
  editors, {\em 49th International Colloquium on Automata, Languages, and
  Programming, {ICALP} 2022, July 4-8, 2022, Paris, France}, volume 229 of {\em
  LIPIcs}, pages 90:1--90:18. Schloss Dagstuhl - Leibniz-Zentrum f{\"{u}}r
  Informatik, 2022.

\bibitem[LRSZ20]{LRSZ20}
Daniel Lokshtanov, M.~S. Ramanujan, Saket Saurabh, and Meirav Zehavi.
\newblock Parameterized complexity and approximability of directed odd cycle
  transversal.
\newblock In Shuchi Chawla, editor, {\em Proceedings of the 2020 {ACM-SIAM}
  Symposium on Discrete Algorithms, {SODA} 2020, Salt Lake City, UT, USA,
  January 5-8, 2020}, pages 2181--2200. {SIAM}, 2020.

\bibitem[MR17]{MR16}
Pasin Manurangsi and Prasad Raghavendra.
\newblock {A Birthday Repetition Theorem and Complexity of Approximating Dense
  CSPs}.
\newblock In Ioannis Chatzigiannakis, Piotr Indyk, Fabian Kuhn, and Anca
  Muscholl, editors, {\em 44th International Colloquium on Automata, Languages,
  and Programming (ICALP 2017)}, volume~80 of {\em Leibniz International
  Proceedings in Informatics (LIPIcs)}, pages 78:1--78:15, Dagstuhl, Germany,
  2017. Schloss Dagstuhl--Leibniz-Zentrum fuer Informatik.

\bibitem[Mul54]{Muller54}
David~E. Muller.
\newblock Application of boolean algebra to switching circuit design and to
  error detection.
\newblock {\em Trans. {I} {R} {E} Prof. Group Electron. Comput.}, 3(3):6--12,
  1954.

\bibitem[PS94]{PS94}
Alexander Polishchuk and Daniel~A. Spielman.
\newblock Nearly-linear size holographic proofs.
\newblock In Frank~Thomson Leighton and Michael~T. Goodrich, editors, {\em
  Proceedings of the Twenty-Sixth Annual {ACM} Symposium on Theory of
  Computing, 23-25 May 1994, Montr{\'{e}}al, Qu{\'{e}}bec, Canada}, pages
  194--203. {ACM}, 1994.

\bibitem[RS92]{rs-test}
Ronitt Rubinfeld and Madhu Sudan.
\newblock Self-testing polynomial functions efficiently and over rational
  domains.
\newblock In {\em Proceedings of the Third Annual ACM-SIAM Symposium on
  Discrete Algorithms}, SODA '92, page 23–32, USA, 1992. Society for
  Industrial and Applied Mathematics.

\bibitem[RS96]{RS96}
Ronitt Rubinfeld and Madhu Sudan.
\newblock Robust characterizations of polynomials with applications to program
  testing.
\newblock {\em {SIAM} J. Comput.}, 25(2):252--271, 1996.

\bibitem[WY07]{WY07}
David Woodruff and Sergey Yekhanin.
\newblock A geometric approach to information-theoretic private information
  retrieval.
\newblock {\em SIAM Journal on Computing}, 37(4):1046--1056, 2007.

\bibitem[Zuc96a]{Zuc96a}
David Zuckerman.
\newblock On unapproximable versions of np-complete problems.
\newblock {\em {SIAM} J. Comput.}, 25(6):1293--1304, 1996.

\bibitem[Zuc96b]{Zuc96b}
David Zuckerman.
\newblock Simulating {BPP} using a general weak random source.
\newblock {\em Algorithmica}, 16(4/5):367--391, 1996.

\bibitem[Zuc07]{Zuc07}
David Zuckerman.
\newblock Linear degree extractors and the inapproximability of max clique and
  chromatic number.
\newblock {\em Theory Comput.}, 3(1):103--128, 2007.

\end{thebibliography}

\appendix
\section{Optimizing the Parameters in~\cite{DBLP:conf/icalp/LinRSW22}}\label{app:lrsw22}

\begin{theorem}\label{thm:app-b}
If constant-gap \kclique{} has lower bound $f(k)\cdot n^{\omega(\sqrt{k})}$, then PIH is true.
\end{theorem}

Recall that a $(k,m,\ell,r,\varepsilon)$-disperser is a collection $\mathcal I$ of $m$ subsets $I_1,\ldots,I_m\subseteq [k]$, each of size $\ell$, such that the union of any $r$ different subsets has size at least $(1-\varepsilon)k$. For every sufficiently large $k$, we can construct a $(k,m=2^{\sqrt k},\ell=\Theta_{\varepsilon}(\sqrt{k}),r=\sqrt{k},\varepsilon/2)$-disperser by Lemma~\ref{lem:disperser}. To prove this theorem, we need the following lemma, which is implicitly extracted from~\cite{DBLP:conf/icalp/LinRSW22}. 
\begin{lemma}\label{lem:A}
 For every constants $\beta>0,\varepsilon\in (0,1)$, if there exists:
\begin{itemize}
    \item a $(k,m(k),\ell(k),r(k),\varepsilon/2)$-disperser such that $m(k)\ge 2^{\beta r(k)}$ for every sufficiently large $k$.
    \item a lower bound $f(k)\cdot n^{\omega(\ell(k))}$ for $(1-\varepsilon)$-approximation of \kclique{} for any computable function $f$.
\end{itemize}
Then {\sc Multi-colored Densest $k$-Subgraph} does not have any $f(k)\cdot n^{O(1)}$-time $(1-\varepsilon')$-approximation algorithm, for some constant $\varepsilon'>0$ and any computable function $f$.
\end{lemma}

It is not hard to see that Theorem~\ref{thm:app-b} follows from  Lemma~\ref{lem:A} and Lemma~\ref{lem:disperser}  by setting $\ell(k)=\sqrt{k}$. Below we first introduce the K{\~o}v{\'a}ri-S{\'o}s-Tur{\'a}n Theorem, then provide a proof for Lemma~\ref{lem:A}.

\begin{lemma}[K{\~o}v{\'a}ri-S{\'o}s-Tur{\'a}n Theorem \cite{KST54}]
    For any graph $G$ on $n$ vertices and any $a \ge 2$, if $G$ does not contain $K_{a,a}$ as a subgraph, then $G$ has at most $\frac{1}{2}(a-1)^{\frac{1}{a}}n^{2-\frac{1}{a}}+\frac{1}{2}(a-1)n$ edges.
\end{lemma}

\begin{proof}[Proof for Lemma~\ref{lem:A}]
Let $G=(V=V_1\dot\cup\ldots\dot\cup V_k,E)$ be a \kclique{} instance, it suffices to construct a colored graph $H$ in $f(k)\cdot n^{O(\ell(k))}$-time, with $m$ bounded by some function of $k$, such that 
\begin{itemize}
\item if $G$ contains a size-$k$ clique, then $H$ contains a size-$m$ multi-colored clique;
\item if $G$ does not contain a size-$(1-\varepsilon)k$ clique, then any $m$-vertex multi-colored subgraph in $H$ has at most $(1-\varepsilon')\binom{m}{2}$ edges.
\end{itemize}

Let $\mathcal I=\{I_1,\ldots,I_m\}$ be a $(k,m,\ell,r,\varepsilon/2)$-disperser. We construct $H$ using $\mathcal I$ as follows.
\begin{itemize}
    \item Let $V(H)=\dot\bigcup_{i \in [m]}\mathcal V_i$.
    \item For every $i \in [m]$ and $I_i=\{a_1,\ldots,a_{\ell}\}$, let $\mathcal V_i=\{\{v_1,\ldots,v_{\ell}\}:v_1 \in V_{a_1},\ldots,v_\ell \in V_{a_\ell}\}$.
    \item For any $\mathbf u=\{u_1,\ldots,u_{\ell}\} \in \mathcal V_i$ and $\mathbf w=\{w_1,\ldots,w_{\ell}\} \in \mathcal V_j$, $(\mathbf u,\mathbf w) \in E(H)$ if and only if $\{u_1,\ldots,u_{\ell},w_1,\ldots,w_{\ell}\}$ forms a clique in $G$.
\end{itemize}

We say a subgraph $F$ of $H$ is multi-colored if it contains at most one vertex in each $\mathcal V_i$, i.e., $|V(F) \cap \mathcal V_i| \le 1$ for every $i \in [m]$. Next, we show $H$ satisfies the following properties.
\begin{itemize}
    \item If $G$ contains a size-$k$ clique, then $H$ contains a size-$m$ multi-colored clique. In fact, let $v_1 \in V_1,\ldots,v_k \in V_k$ be a clique in $G$, then for every $i \in [m]$, let $\mathbf u_i=\{v_j:j\in I_i\}$. The vertex set $\{\mathbf u_i:i \in [m]\}$ induces a complete subgraph in $H$.

    \item If $H$ contains a multi-colored $K_{r,r}$-subgraph, then there is a clique in $G$ of size at least $(1-\varepsilon)k$. Suppose the vertices $\mathbf u_{a_1}\in \mathcal V_{a_1},\ldots,\mathbf u_{a_r} \in \mathcal V_{a_r}$ and $\mathbf w_{b_1} \in \mathcal V_{b_1},\ldots,\mathbf w_{b_r}\in \mathcal V_{b_r}$ induce a multi-colored $K_{r,r}$-subgraph. Let $X=\bigcup_{i \in [r]} I_{a_i}$ and $Y=\bigcup_{i \in [r]} I_{b_i}$. By the property of $(k,m,\ell,r,\varepsilon/2)$-disperser, we have $|X|\ge (1-\varepsilon/2)k$ and $|Y|\ge (1-\varepsilon/2)k$. Thus $|X \cap Y| \ge (1-\varepsilon)k$. For every index $i \in X \cap Y$, let $U_i=\left(\bigcup_{i \in [r]} \mathbf u_{a_i}\right) \cap V_i$, and $W_i=\left(\bigcup_{i \in [r]} \mathbf w_{a_i}\right)\cap V_i$. By the fact that $\{\mathbf u_{a_1},\ldots,\mathbf u_{a_r},\mathbf w_{a_1},\ldots,\mathbf w_{a_r}\}$ forms a $K_{r,r}$-subgraph and our construction of $E(H)$, we know $U_i\cup W_i$ is a clique in $G$. Since each $V_i$ in $G$ is an independent set, we have $U_i=W_i=\{v_i\}$ for some $v_i \in V_i$. Using similar analysis, for every distinct $i,j\in X \cap Y$, there is an edge $(v_i,v_j)$ in $G$. Therefore, $G$ contains a clique of size $|X \cap Y|\ge (1-\varepsilon)k$.
\end{itemize}

Putting things together, if $G$ contains a size-$k$ clique, then $H$ contains a multi-colored clique of size $m$. On the other hand, if $G$ contains no size-$(1-\varepsilon)k$ clique, then $H$ contains no multi-colored $K_{r,r}$-subgraph. By the K{\~o}v{\'a}ri-S{\'o}s-Tur{\'a}n Theorem, any $m$-vertex multi-colored subgraph of $H$ has at most $\frac{1}{2}(r-1)^{\frac{1}{r}}m^{2-\frac{1}{r}}+\frac{1}{2}(r-1)m\le(1-\varepsilon')\binom{m}{2}$ edges for some constant $\varepsilon'\in(0,1)$ given $r\le \frac{1}{\beta}\log m$, which proves the hardness of $(1-\varepsilon')$-approximating {\sc Multi-colored Densest $m$-Subgraph} problem.

\end{proof}

\section{Proof of Lemma \ref{lem:det-maintext}}
\label{lem:det-appendix}
Before proving Lemma \ref{lem:det-maintext}, we first state a folklore lemma about binomial transform. See e.g. \cite{bernstein1995some} for a reference.
\begin{lemma}
\label{lem:binom-transform}
    Let $\mathbb{F}$ be a field and $f: \mathbb{N}\to \mathbb{F}, g: \mathbb{N}\to \mathbb{F}$ be two functions. Then
    \[
        f(n) = \sum_{i=0}^n (-1)^{n-i} \binom{n}{i} g(i), \forall n\in \mathbb{N},
    \]
    if and only if
    \[
        g(n) = \sum_{i=0}^n \binom{n}{i} f(i), \forall n\in \mathbb{N}.
    \]
\end{lemma}

\begin{lemma}
    Let $r > 0$ be an integer and $d = 2r+1$, $\mathbb F := \mathbb F_q$ be a finite field with prime size $q > r$, and $\lambda_1,\lambda_2$ be two distinct elements in $\mathbb F$. The values $g(\lambda_1), g'(\lambda_1), \ldots, g^{(r)}(\lambda_1)$ and $g(\lambda_2), g'(\lambda_2), \ldots, g^{(r)}(\lambda_2)$ uniquely determine a univariate polynomial $g: \mathbb{F}\to \mathbb{F}$ with degree at most $d$.
\end{lemma}
\begin{proof}
    Let $f(\lambda) = \sum_{i=0}^{d} a_i \lambda^i$. The constraints can be rewritten as:
    \[
        \begin{pmatrix} f(\lambda_1) \\ f'(\lambda_1) \\ \vdots \\ f^{(r)}(\lambda_1) \\ f(\lambda_2) \\ f'(\lambda_2) \\ \vdots \\ f^{(r)}(\lambda_2)\end{pmatrix} = \begin{pmatrix} 1 & \lambda_1 & \lambda_1^2 & \ldots & \lambda_1^{d} \\ 0 & 1 & 2\lambda_1 & \ldots & d\lambda_1^{d-1} \\ \vdots & \vdots & \vdots & \ddots & \vdots \\ 0 & 0 & 0 & \ldots & \binom{d}{r} r! \lambda_1^{d-r} \\ 1 & \lambda_2 & \lambda_2^2 & \ldots & \lambda_2^{d} \\ 0 & 1 & 2\lambda_2 & \ldots & d\lambda_2^{d-1} \\ \vdots & \vdots & \vdots & \ddots & \vdots \\ 0 & 0 & 0 & \ldots & \binom{d}{r} r! \lambda_2^{d-r} \end{pmatrix}  \begin{pmatrix} a_0 \\ a_1 \\ a_2 \\ \vdots \\ a_{d} \end{pmatrix} 
    \]
    We further denote the coefficient matrix as $A$, which can be formally defined as follows: For every $0\le i,j\le d$:
    \[
        A_{ij} = \begin{cases} \binom{j}{i} i! \lambda_1^{j-i} & 0\le i\le r \\ \binom{j}{i - (r+1)} (i - (r+1))! \lambda_2^{j - (i-(r+1))} & r < i\le d \end{cases}
    \]
    We claim that $\det(A) = (\lambda_2 - \lambda_1)^{(r+1)^2} \prod_{i=1}^r (i!)^2\not\equiv 0\pmod q$ since $q > r$ is a prime. If our claim is true, then $A^{-1}$ exists and hence $a_0, a_1, \ldots, a_{d}$ are uniquely determined.

    To prove our claim, we first divide the $i$-th row by $i!$ if $0\le i\le r$ or $(i-(r+1))!$ otherwise. The new matrix $A'$ becomes \[A'_{ij} = \begin{cases} \binom{j}{i} \lambda_1^{j-i} & 0\le i\le r \\ \binom{j}{i - (r+1)} \lambda_2^{j - (i-(r+1))} & r < i\le d \end{cases}, \forall 0\le i, j\le d\] and $\det(A) = \det(A') \prod_{i=1}^r (i!)^2$.

    Let $A_0', A_1', \ldots, A_{d}'$ be the columns of $A'$, and let $B = (B_0, B_1, \ldots, B_{d})\in \mathbb{F}_q^{(d+1)\times (d+1)}$ be a matrix defined as follows:
    \[
        B_j  = \sum_{k=0}^j (-1)^k \lambda_1^k \binom{j}{k} A_{j-k}'.
    \]
    Then $\det(B) = \det(A')$. For the $i$-th row where $0\le i\le r$, we have
    \begin{align*}
        B_{ij} &= \sum_{k=0}^j (-1)^k \lambda_1^k \binom{j}{k} \binom{j-k}{i} \lambda_1^{j-k-i} \\
        &= \sum_{k=0}^{j-i} (-\lambda_1)^k \lambda_1^{j-k-i} \binom{j}{k} \binom{j-k}{i} \\
        &= \sum_{k=0}^{j-i} (-\lambda_1)^k \lambda_1^{j-k-i} \binom{j}{i} \binom{j-i}{k} \\
        &= \binom{j}{i} \sum_{k=0}^{j-i} \binom{j-i}{k} (-\lambda_1)^k \lambda_1^{j-i-k} \\
        &= \binom{j}{i} (\lambda_1 - \lambda_1)^{j-i} \\
        &= \mathbf 1_{\{i = j\}}
    \end{align*}
    where the third equality is from $\binom{j}{k}\binom{j-k}{i} = \frac{j!}{k!i!(j-k-i)!} = \binom{j}{i} \binom{j-i}{k}$.

    Similarly, for the $i$-th row where $r < i\le d$, we have
    \begin{align*}
        B_{ij} &= \sum_{k=0}^j (-1)^k \lambda_1^k \binom{j}{k} \binom{j-k}{i-(r+1)} \lambda_2^{j-k-(i-(r+1))} \\&= \binom{j}{i-(r+1)} (\lambda_2 - \lambda_1)^{j-(i-(r+1))}.
    \end{align*}
    Note that $B$ is in the form of
    \[
        B =: \begin{pmatrix} I_{(r+1)\times (r+1)} & O_{(r+1)\times (r+1)} \\ C_{(r+1)\times (r+1)} & D_{(r+1)\times (r+1)} \end{pmatrix}
    \]
    and hence $\det(B) = \det(D)$, where
    \[
        D_{ij} = \binom{j + (r+1)}{i} (\lambda_2 - \lambda_1)^{j-i + (r+1)}, \forall 0\le i,j\le r.
    \]
    Let $D' = (\lambda_2 - \lambda_1)^{-(r+1)}D$, then $\det(D) = \det(D') (\lambda_2 - \lambda_1)^{(r+1)^2}$. Let $D_0', D_1', \ldots, D_r'$ be the columns of $D'$, and let $E = (E_0, E_1, \ldots, E_r)\in \mathbb{F}_q^{(r+1)\times (r+1)}$ be a matrix defined as follows:
    \[
        E_j =  \sum_{k=0}^j (-1)^k (\lambda_2 - \lambda_1)^k \binom{j}{k} D_{j-k}'.
    \]
    Then $\det(E) = \det(D')$. Moreover, we have
    \begin{align*}
        E_{ij} &= \sum_{k=0}^j (-1)^k (\lambda_2 - \lambda_1)^k \binom{j}{k} \binom{j-k+(r+1)}{i} (\lambda_2 - \lambda_1)^{j-k-i} \\
        &= (\lambda_2 - \lambda_1)^{j-i} \sum_{k=0}^j (-1)^k \binom{j}{k} \binom{j-k+(r+1)}{i} \\
        &= (\lambda_2 - \lambda_1)^{j-i} \sum_{k=0}^j (-1)^{j-k} \binom{j}{k} \binom{k+(r+1)}{i}.
    \end{align*}

    For any $j\in \mathbb{N}$, define
    \begin{align*}
        f(j) &:= \sum_{k=0}^j (-1)^{j-k} \binom{j}{k} \binom{k+(r+1)}{i},\\g(j) &:= \binom{j+(r+1)}{i}.
    \end{align*}
    By Lemma~\ref{lem:binom-transform}, we have
    \[
        g(j) = \sum_{k=0}^j \binom{j}{k} f(k).
    \]
    Note that
    \[
        \binom{j+(r+1)}{i} = \sum_{k=0}^{\min\{i, j\}} \binom{j}{k} \binom{r+1}{i-k},
    \]
    so $f(j) = \begin{cases} \binom{r+1}{i-j} & j\le i \\ 0 & j > i\end{cases}$ is the unique solution. Hence
    \[
        E_{ij} = \begin{cases} (\lambda_2 - \lambda_1)^{j-i} \binom{r+1}{i-j} & j\le i \\ 0 & j > i\end{cases}
    \]
    is a lower triangular matrix, and $\det(E) = 1$. Finally, we have
    \begin{align*}
        \det(A) &= \det(B) \prod_{i=1}^r (i!)^2 = \det(D') (\lambda_2 - \lambda_1)^{(r+1)^2} \prod_{i=1}^r (i!)^2 \\&= (\lambda_2 - \lambda_1)^{(r+1)^2} \prod_{i=1}^r (i!)^2,
    \end{align*}
    finishing the proof.
\end{proof}
\section{Proof of Lemma \ref{lem:line-vs-point}}
\label{app:line-vs-point}
The completeness case is straightforward, and we focus on the soundness case.

\begin{lemma}\label{thm:ldt-lvp}
If $\mathcal T$ rejects with probability $\delta$, then $\Delta(f,\widetilde f) \le 2\delta$ for some $\widetilde f\in \mathcal F$.
\end{lemma}

Note that without changing $f$, the prover can optimize the accepting probability for each $\ell \in \mathcal L^m$ individually. Therefore, the rejecting probability for a fixed function $f$ is minimized when for every $\ell \in \mathcal L^m$, $P_{\ell}:\ell \to \mathbb F^t$ is the closest degree-$d$ function to $f_{|\ell}$. We denote such a degree-$d$ function as $P^f_{\ell}$, and denote the rejecting probability of $\mathcal T$ on $(f,P^f_{\ell})$ as $\delta(f)$. We only need to prove the case where $P=\{P^f_\ell | \ell \in \mathcal L^m\}$, since otherwise the rejecting probability will be larger.

We define $\MAJ^f:\mathbb F^m\to \mathbb F^t$ to be a function that maps every $\mathbf x \in \mathbb F^m$ to the most frequent value in $\{P^f_{\ell(\mathbf x,\mathbf h)}(\mathbf x)\left|\mathbf h \in \mathbb F^m\right.\}$. In other words, the most frequent value when we consider all directions $\mathbf h \in \mathbb F^m$ and try to decode $f(\mathbf x)$ from its restriction on the line $\ell(\mathbf x,\mathbf h)$.

\begin{lemma}\label{lem:ldt-lvp-1}
$\Delta(f,\MAJ^f)\le 2\delta(f)$.
\end{lemma}
\begin{proof}
 Recall that $\delta(f)=\Ex_{\mathbf x\in \mathbb F^m}[\Pr_{\mathbf h \in \mathbb F^m}[f(\mathbf x)\ne P^f_{\ell(\mathbf x,\mathbf h)}(\mathbf x)]]$. Let \[B=\left\{\mathbf x \in \mathbb F^m \mid \Pr_{\mathbf h \in \mathbb F^m}[f(\mathbf x)\ne P^f_{\ell(\mathbf x,\mathbf h)}(\mathbf x)]\ge \frac{1}{2}\right\}.\] By Markov's inequality, $\frac{|B|}{|\mathbb F|^m}\le 2\delta(f)$. For any $\mathbf x \notin B$, by definition we have $\MAJ^f(\mathbf x)=f(\mathbf x)$. Thus the lemma holds.
\end{proof}

\begin{lemma}\label{lem:ldt-lvp-2}
    Assume $|\mathbb F|> \max(3d,9216)$. If there exists $\widetilde f\in \mathcal F$ such that $\Delta(f,\widetilde f) \le \frac{1}{12}$, then $\MAJ^f=\widetilde f$.
\end{lemma}

\begin{proof}
    Fix any $\mathbf x \in \mathbb F^m$, we want to prove $\MAJ^f(\mathbf x) =\widetilde f(\mathbf x)$. 
    
    Consider the distribution of $\mathbf y$ defined by drawing a random $\mathbf h \in \mathbb F^m$ and picking a random $\mathbf y \in \ell(\mathbf x,\mathbf h)$.  
    When $\mathbf h = \mathbf 0$, $\mathbf y\equiv \mathbf x$; otherwise, $\mathbf y = \mathbf x$ happens with probability $\frac{1}{|\mathbb F|}$. Thus \[\Ex_{\mathbf h\in \mathbb F^m}\left[\Pr_{\mathbf y \in \ell(\mathbf x,\mathbf h)} \left[ \mathbf y = \mathbf x \right]  \right] = \frac{1}{|\mathbb F|^m} + \left(1 - \frac{1}{|\mathbb F|^m}\right)\cdot \frac{1}{|\mathbb F|} =: p,\] and $\mathbf y$ is drawn from $\mathbb F^m\setminus \{\mathbf x\}$ uniformly at random conditioned on $\mathbf y \neq \mathbf x$, which happens with probability $1-p$.
    \begin{align*}
    &\Ex_{\mathbf h\in \mathbb F^m}\left[\Pr_{\mathbf y \in \ell(\mathbf x,\mathbf h)} \left[ f(\mathbf y) \ne \widetilde f(\mathbf y) \right]  \right]\\
    &\le (1 - p)\cdot \left(1 - \frac{1}{|\mathbb F|^m}\right)\cdot \Delta(f, \widetilde f) + 1 - (1 - p)\cdot \left(1 - \frac{1}{|\mathbb F|^m}\right) \\
    &\le \Delta(f, \widetilde f) + 1 - \left(1 - \frac{1}{|\mathbb F|}\right) \left(1 - \frac{1}{|\mathbb F|^m}\right)^2 \\
    &\le \Delta(f, \widetilde f) + 1 - \left(1 - \frac{1}{|\mathbb F|}\right)^3 \\
    &\le \frac{1}{12} + \frac{1}{12} = \frac{1}{6}.
    \end{align*}

Let \[B=\left\{\mathbf h \in \mathbb F^m|\Pr_{\mathbf y \in \ell(\mathbf x,\mathbf h)}\left[f(\mathbf y) \ne \widetilde f(\mathbf y)\right]\ge \frac{1}{3}\right\}.\] By Markov's inequality we have $\frac{|B|}{|\mathbb F|^m} < \frac{1}{6}\times 3 = \frac{1}{2}$. 

For any $\mathbf h \not \in B$, we have $\Pr_{\mathbf y \in \ell(\mathbf x, \mathbf h)}[f(\mathbf y) \ne \widetilde f(\mathbf y)]< \frac{1}{3}$. 
Note that $\widetilde f$ is a degree-$d$ multivariate polynomial tuple, thus $\widetilde f_{|\ell(\mathbf x,\mathbf h)}$ is a degree-$d$ univariate polynomial tuple. By Schwartz-Zippel lemma, any two different degree-$d$ polynomials differ on $1-\frac{d}{|\mathbb F|}$ fraction of points, and so do two degree-$d$ polynomial tuples. Since the distance between $f_{|\ell(\mathbf x,\mathbf h)}$ and $\widetilde f_{|\ell(\mathbf x,\mathbf h)}$ is less than $\frac{1}{3}<\frac{1-\frac{d}{|\mathbb F|}}{2}$, we can conclude $\widetilde f_{|\ell(\mathbf x,\mathbf h)}$ is the closest univariate degree-$d$ polynomial tuple to $f_{|\ell(\mathbf x,\mathbf h)}$. By definition we have $P^f_{\ell(\mathbf x,\mathbf h)}=\widetilde f_{|\ell(\mathbf x,\mathbf h)}$. Furthermore, since $\Pr_{\mathbf h \in \mathbb F^m} [\mathbf h \not \in B] >\frac{1}{2}$, we have $\MAJ^f(\mathbf x)=\widetilde f(\mathbf x)$.

\end{proof}

\begin{lemma}\label{lem:ldt-lvp-3}
    Assume $|\mathbb F|> \max(3d,9216)$. If $0 < \delta(f) \le \frac{1}{48}$, then $\delta(\MAJ^f) <\delta(f)$.
\end{lemma}
\begin{proof}
    Draw $\mathbf x, \mathbf h_1,\mathbf h_2, \mathbf h_3 \in \mathbb F^m$ uniformly at random. Define the following random polynomial tuples and random variables:
    \begin{itemize}
        \item $a(i,j)=f(\mathbf x+i\mathbf h_1+j\mathbf h_2+ij\mathbf h_3)$,
        \item $r_i(j)=P^f_{\ell(\mathbf x+i\mathbf h_1,\mathbf h_2+i\mathbf h_3)}(\mathbf x+i\mathbf h_1+j\mathbf h_2+ij\mathbf h_3)$,
        \item $c_j(i)=P^f_{\ell(\mathbf x+j\mathbf h_2,\mathbf h_1+j\mathbf h_3)}(\mathbf x+i\mathbf h_1+j\mathbf h_2+ij\mathbf h_3)$,
        \item $\delta^R_i=\Pr_{j \in \mathbb F}[r_i(j)\ne a(i,j)]$,
        \item $\delta^C_j=\Pr_{i \in \mathbb F}[c_j(i)\ne a(i,j) ]$,
        \item $\eta_i^R=\mathbf 1[r_i(0)\ne a(i,0)]$,
        \item $\eta_j^C=\mathbf 1[c_j(0)\ne a(0,j)]$.
    \end{itemize}

    Let $\mathcal D$ be the distribution of $\ell(\mathbf x,\mathbf h)$ where $\mathbf x,\mathbf h$ are sampled uniformly at random in $\mathbb F^m$. It's easy to see when $\mathbf x,\mathbf h_1,\mathbf h_2,\mathbf h_3$ are uniform, $\ell(\mathbf x+i\mathbf h_1,\mathbf h_2+i\mathbf h_3)$ follows distribution $\mathcal D$ for every $i\in \mathbb F$. Furthermore, for different $i,i' \in \mathbb F$, $\ell(\mathbf x+i\mathbf h_1,\mathbf h_2+i\mathbf h_3)$ and $\ell(\mathbf x+i'\mathbf h_1,\mathbf h_2+i'\mathbf h_3)$ are independent. 

    Note that for every $i \in \mathbb F$, $\Ex[\delta_i^R]=\delta(f)$, and the $\delta_i^R$'s are pairwise independent. Let $X=\frac{1}{|\mathbb F|}\sum_{i \in \mathbb F}\delta_i^R$, we bound the variance of $X$ as follows:
    \begin{align*}
        \Var[X] & = \Var\left[\frac{1}{|\mathbb F|}\sum_{i \in \mathbb F}\delta_i^R\right] \\
        & = \frac{|\mathbb F|}{|\mathbb F|^2}\Var[\delta_0^R] \\
        & \le \frac{1}{|\mathbb F|}\Ex[X](1-\Ex[X]) \\
        & \le \frac{1}{|\mathbb F|} \delta(f),
    \end{align*}
    where the inequality on the third line holds because $X$ takes value in $[0,1]$ and \[\Var[X] = \Ex[X^2] - \Ex[X]^2\le \Ex[X] - \Ex[X]^2 = \Ex[X] (1 - \Ex[X]).\] By Chebyshev's inequality, $\Pr\left[X \ge \delta(f)+\frac{1}{48}\right] \le \frac{\Var[X]}{(1/48)^2}\le \frac{2304\delta(f)}{|\mathbb F|}$. Similar analysis applies to $\delta_j^C$'s, and more generally, to $\eta_i^R$'s and $\eta_j^C$'s, too. By union bound, with probability at least $1-\frac{9216\delta(f)}{|\mathbb F|}$ over random choices of $\mathbf x,\mathbf h_1,\mathbf h_2,\mathbf h_3$, we have 
    \[\left\{
    \begin{aligned}
        & \Pr_{i,j \in \mathbb F} [r_i(j)\ne a(i,j)]\le \frac{1}{24} \\
        & \Pr_{i,j \in \mathbb F} [c_j(i)\ne a(i,j)]  \le \frac{1}{24} \\
        & \Pr_{i \in \mathbb F}[r_i(0)\ne a(i,0)]  \le \frac{1}{24} \\
        & \Pr_{j \in \mathbb F}[c_j(0)\ne a(0,j)]  \le \frac{1}{24} \\
    \end{aligned}
    \right..\]
    Suppose this happens, then
    \begin{align*}
        & \Pr_{i,j\in \mathbb F}[r_i(j)\ne c_j(i)]\\
        \le & \Pr_{i,j \in \mathbb F}[r_i(j)\ne a(i,j)\lor c_j(i) \ne a(i,j)] \\
        \le & \Pr_{i,j \in \mathbb F}[r_i(j) \ne a(i,j)]+\Pr_{i,j \in \mathbb F}[c_j(i)\ne a(i,j)] \\
        \le & \frac{1}{24}+\frac{1}{24} = \frac{1}{12}.
    \end{align*}

    For a bivariate polynomial $F(x,y)$, we say it has degree $(d,e)$ if it has degree at most $d$ in $x$ and degree at most $e$ in $y$. 
    Note that $r_i(j)$ can be viewed as a $t$-tuple of bivariate polynomials $R(i,j) = (R_1,\ldots, R_t)(i,j)$, where each $R_i$ has degree $(|\mathbb F|, d)$. Similarly, $c_j(i)$ can also be viewed as $t$-tuple of bivariate polynomials $C(i,j)=(C_1,\ldots,C_t)(i,j)$, where each $C_i$ has degree $(d,|\mathbb F|)$. Since $\Pr_{i,j \in \mathbb F} \left[R(i,j) \ne C(i,j)\right] \le \frac{1}{12}$, by Theorem 9 of \cite{PS94}, there exists a non-zero bivariate polynomial $E(i,j)$ of degree $\left(\frac{n}{\sqrt {12}},\frac{n}{\sqrt {12}}\right)$ and a $t$-tuple of polynomials $Q(i,j) = (Q_1,\ldots,Q_t)(i,j)$, where each $Q_i$ has degree $(d,d)$, and for every $i,j \in \mathbb F$,
    \[R(i,j)E(i,j)=C(i,j)E(i,j)=Q(i,j)E(i,j).\]
    
    This implies that in any row on which $E(i,\cdot)$ is non-zero, $R$ agrees with $Q$ on the entire row. Since $E$ is non-zero, $E$ can be identically zero on at most $\frac{n}{\sqrt{12}}$ rows.
    Therefore, $R$ agrees with $Q$ on at least $\left(1-\frac{1}{\sqrt{12}}\right) n$ rows. Hence, there are at least $1-\frac{1}{\sqrt{12}}-\frac{1}{24}$ fraction of $i$'s satisfying:
    \[
    \left\{
        \begin{aligned} r_i(\cdot) &= Q(i, \cdot) \\ r_i(0) &= a(i, 0)
        \end{aligned}
    \right.
    \]
    which means $\Delta(Q(\cdot,0),a(\cdot,0))\le \frac{1}{\sqrt{12}}+\frac{1}{24}<\frac{1}{3}<\frac{1-\frac{d}{|\mathbb F|}}{2}$. The distance between any two different degree-$d$ polynomials is at least $1-\frac{d}{|\mathbb F|}$, thus $Q(\cdot ,0)$ is the closest degree-$d$ univariate polynomial tuple to $a(\cdot,0)$, which by definition, equals to $c_0(\cdot)$. By similar analysis, we also have $Q(0,\cdot)=r_0(\cdot)$. Therefore, with probability at least $1-\frac{9216\delta(f)}{|\mathbb F|}$, \[P^f_{\ell(\mathbf x,\mathbf h_2)}(\mathbf x)=r_0(0)=Q(0,0)=c_0(0)=P^f_{\ell(\mathbf x,\mathbf h_1)}(\mathbf x),\]
    which implies
    \begin{align*}
        \delta(\MAJ^f) & = \Pr_{\mathbf x,\mathbf h \in \mathbb F^m}\left[\MAJ^f(\mathbf x)\ne P^{\MAJ^f}_{\ell(\mathbf x,\mathbf h)}(\mathbf x)\right] \\
        & \le \Pr_{\mathbf x,\mathbf h \in \mathbb F^m}\left[\MAJ^f(\mathbf x) \ne P^f_{\ell(\mathbf x,\mathbf h)}(\mathbf x)\right]\\
        & \le \Pr_{\mathbf x,\mathbf h_1,\mathbf h_2 \in \mathbb F^m}\left[P^f_{\ell(\mathbf x,\mathbf h_1)}(\mathbf x)\ne P^f_{\ell(\mathbf x,\mathbf h_2)}(\mathbf x)\right] \\
        & \le \frac{9216\delta(f)}{|\mathbb F|} \\
        & < \delta(f).
    \end{align*}
\end{proof}

\begin{proof}[Proof of Lemma \ref{thm:ldt-lvp}]
    $\delta(f)$ can take at most $|\mathbb F|^{2m}$ different values, so we use induction on $\delta(f)$ to prove the lemma. The base case where $\delta(f)=0$ is handled in \cite{RS96}, which proves a multivariate polynomial is of degree $d$ if and only if its restriction on every line $\ell(\mathbf x,\mathbf h)$ is a univariate polynomial of degree $d$. Suppose $0 < \delta(f) \le \frac{1}{48}$. By Lemma \ref{lem:ldt-lvp-3}, $\delta(\MAJ^f)<\delta(f)$. From induction hypothesis we know $\Delta(\MAJ^f,\widetilde f) \le 2\delta(\MAJ^f) \le 2\delta(f)$ for some $\widetilde f \in \mathcal F$. By Lemma \ref{lem:ldt-lvp-1}, $\Delta(f,\MAJ^f) \le 2\delta(f)$. Thus $\Delta(f,\widetilde f) \le 4\delta(f) \le \frac{1}{12}$ by triangle inequality. By Lemma \ref{lem:ldt-lvp-2}, we have $\MAJ^f = \widetilde f$. Finally by Lemma \ref{lem:ldt-lvp-1} again, we have $\Delta(f,\widetilde f) \le 2\delta(f)$.
\end{proof}

\end{document}